\documentclass[11pt]{article}
\usepackage[usenames,dvipsnames,svgnames,table]{xcolor}
\usepackage{enumitem}
\usepackage{graphicx}
\usepackage{fancybox}
\usepackage{comment}
\usepackage{xcolor}
\usepackage{nameref}
\usepackage{bbm}
\usepackage[T1]{fontenc}

\definecolor{ForestGreen}{rgb}{0.1333,0.5451,0.1333}
\definecolor{DarkRed}{rgb}{0.8,0,0}
\definecolor{Red}{rgb}{1,0,0}
\usepackage[linktocpage=true,
colorlinks,
linkcolor=ForestGreen,citecolor=ForestGreen,
bookmarks,bookmarksopen,bookmarksnumbered]
{hyperref}
\usepackage[nottoc]{tocbibind}

\usepackage{amsmath}
\usepackage{amssymb}
\usepackage{amsthm}

\usepackage[ruled, noend, linesnumbered]{algorithm2e}

\usepackage{subfig}

\usepackage{thm-restate}

\usepackage{float}
\usepackage{cleveref}

\newtheorem{theorem}{Theorem}[section]
\newtheorem{informaltheorem}[theorem]{Informal Theorem}

\newtheorem{corollary}[theorem]{Corollary}
\newtheorem{lemma}[theorem]{Lemma}

\newtheorem{claim}[theorem]{Claim}

\newtheorem{definition}[theorem]{Definition}

\newtheorem{remark}[theorem]{Remark}

\newtheorem*{theorem*}{Theorem}
\newtheorem*{corollary*}{Corollary}
\newtheorem*{conjecture*}{Conjecture}
\newtheorem*{lemma*}{Lemma}
\newtheorem*{thm*}{Theorem}
\newtheorem*{prop*}{Proposition}
\newtheorem*{obs*}{Observation}
\newtheorem*{definition*}{Definition}

\newtheorem*{remark*}{Remark}
\newtheorem*{rec*}{Recommendation}

\newenvironment{fminipage}%
  {\begin{Sbox}\begin{minipage}}%
  {\end{minipage}\end{Sbox}\fbox{\TheSbox}}

\def\defeq{\stackrel{\mathrm{def}}{=}}

\def\ceil#1{\left\lceil #1 \right\rceil}

\def\abs#1{\left|#1  \right|}

\def\norm#1{\left\| #1 \right\|}

\newcommand\grad{\boldsymbol{\nabla}}

\DeclareMathOperator{\mincut}{mincut}
\DeclareMathOperator{\vol}{vol}

\newcommand\DDelta{\boldsymbol{\mathit{\Delta}}}

\newcommand\pphi{\boldsymbol{\mathit{\phi}}}

\def\aa{\pmb{\mathit{a}}}
\newcommand\bb{\boldsymbol{\mathit{b}}}
\newcommand\cc{\boldsymbol{\mathit{c}}}
\newcommand\dd{\boldsymbol{\mathit{d}}}

\newcommand\ff{\boldsymbol{\mathit{f}}}
\renewcommand\gg{\boldsymbol{\mathit{g}}}

\newcommand\rr{\boldsymbol{\mathit{r}}}
\renewcommand\ss{\boldsymbol{\mathit{s}}}

\newcommand\uu{\boldsymbol{\mathit{u}}}

\newcommand\yy{\boldsymbol{\mathit{y}}}
\newcommand\zz{\boldsymbol{\mathit{z}}}
\newcommand\xx{\boldsymbol{\mathit{x}}}

\newcommand\veczero{\boldsymbol{0}}
\newcommand\vecone{\boldsymbol{1}}

\renewcommand{\deg}{\operatorname{deg}}

\renewcommand\AA{\boldsymbol{\mathit{A}}}
\newcommand\BB{\boldsymbol{\mathit{B}}}

\newcommand\UU{\boldsymbol{\mathit{U}}}

\renewcommand\O{\widetilde{O}}

\newcommand\R{\mathbb{R}}

\newcommand{\expct}[2]{{}\mathop{\mathbb{E}}_{#1}\left[#2\right]}

\DeclareMathOperator*{\val}{val}

\DeclareMathOperator*{\diag}{diag}

\newcommand{\polylog}{\text{ polylog}}

\newcommand\Z{\mathbb{Z}}

\newcommand{\eps}{\epsilon}
\renewcommand{\O}{\widetilde{O}}

\renewcommand{\l}{\langle}
\renewcommand{\r}{\rangle}

\renewcommand{\forall}{\mathrm{\text{ for all }}}

\newcommand{\g}{\nabla}

\newcommand{\bd}{\boldsymbol{d}}
\newcommand{\bg}{\boldsymbol{g}}

\newcommand{\by}{\boldsymbol{y}}

\newcommand{\bDelta}{\boldsymbol{\Delta}}

\renewcommand{\hat}{\widehat}
\renewcommand{\tilde}{\widetilde}

\DeclareFontFamily{U}{mathb}{\hyphenchar\font45}
\DeclareFontShape{U}{mathb}{m}{n}{<5> <6> <7> <8> <9> <10> gen * mathb
<10.95> mathb10 <12> <14.4> <17.28> <20.74> <24.88> mathb12}{}
\DeclareSymbolFont{mathb}{U}{mathb}{m}{n}
\DeclareMathSymbol{\rcirclearrow}{\mathbin}{mathb}{'367}

\newcommand{\wt}{\widetilde}

\renewcommand{\bar}{\overline}

\newif\ifrandom
\randomtrue

\renewcommand{\l}{\langle}
\renewcommand{\r}{\rangle}

\newcommand{\todolater}[1]{}

\newcommand{\cA}{\mathcal{A}}
\newcommand{\cD}{\mathcal{D}}

\DeclareUnicodeCharacter{2113}{$\ell$}

\usepackage[margin=1in]{geometry}
\usepackage{orcidlink} 
\usepackage{lscape}

\renewcommand{\paragraph}[1]{
\medskip \noindent \textbf{#1}}

\usepackage[backend=biber, isbn=false, style=alphabetic, backref=true, doi=false, url=false, maxcitenames=10, mincitenames=5, maxalphanames=10, maxbibnames=10, minbibnames=5, minalphanames=5, defernumbers=true, sortlocale=en_US]{biblatex}
\title{Almost-Linear Time Algorithms for Decremental Graphs:\\ Min-Cost Flow and More via Duality}

\date{}
\newcommand*\samethanks[1][\value{footnote}]{\footnotemark[#1]}
\author{
}
\author{Jan van den Brand\thanks{Jan van den Brand was supported by NSF Award CCF-2338816.} \\ Georgia Tech \\ vdbrand@gatech.edu
\and Li Chen\thanks{Li Chen was supported by NSF Grant CCF-2330255.} \\ Carnegie Mellon University \\ lichenntu@gmail.com
\and Rasmus Kyng\thanks{The research leading to these results has received funding from grant no. 200021 204787 of the Swiss National Science Foundation.} \\ ETH Zurich \\ kyng@inf.ethz.ch \and Yang P. Liu\thanks{This work is partially supported by
 NSF DMS-1926686.}\\ Institute for Advanced Study \\ yangpliu@ias.edu \and Simon Meierhans\samethanks[3] \\ ETH Zurich \\ mesimon@inf.ethz.ch \and Maximilian Probst Gutenberg\samethanks[3] \\
ETH Zurich \\ maximilian.probst@inf.ethz.ch
\and Sushant Sachdeva\thanks{Sushant Sachdeva’s research is supported by an Natural Sciences and Engineering Research Council of Canada (NSERC) Discovery Grant RGPIN-2018-06398, an Ontario Early Researcher Award (ERA) ER21-16-283, and a Sloan Research Fellowship} \\ University of Toronto \\ sachdeva@cs.toronto.edu}

\begin{document}

\maketitle

\begin{abstract}

We give the first almost-linear total time algorithm for deciding if a flow of cost at most $F$ still exists in a directed graph, with edge costs and capacities, undergoing decremental updates, i.e., edge deletions, capacity decreases, and cost increases.
This implies almost-linear time algorithms for approximating the minimum-cost flow value and $s$-$t$ distance on such \emph{decremental} graphs. 
Our framework additionally allows us to maintain decremental strongly connected components in almost-linear time deterministically. 
These algorithms also improve over the current best known runtimes for \emph{statically} computing minimum-cost flow, in both the randomized and deterministic settings. 

We obtain our algorithms by taking the dual perspective, which yields cut-based algorithms.
More precisely, our algorithm computes the flow via a sequence of $m^{1+o(1)}$ dynamic \emph{min-ratio cut} problems, the dual analog of the dynamic min-ratio cycle problem that underlies recent fast algorithms for minimum-cost flow. Our main technical contribution is a new data structure that returns an approximately optimal min-ratio cut in amortized $m^{o(1)}$ time by maintaining a \emph{tree-cut sparsifier}. This is achieved by devising a new algorithm to maintain the dynamic expander hierarchy of [Goranci-R\"{a}cke-Saranurak-Tan, SODA 2021] that also works in \emph{capacitated} graphs. All our algorithms are deterministc, though they can be sped up further using randomized techniques while still working against an adaptive adversary. 
\end{abstract}

\pagenumbering{gobble}

\pagebreak

\tableofcontents

\pagebreak

\pagenumbering{arabic}

\pagebreak



\pagebreak
\section{Introduction}
\label{sec:intro}

The study of \emph{dynamic graph algorithms} involves designing efficient algorithms for graphs  undergoing edge updates.
In this paper, we focus on solving the challenging \emph{minimum-cost flow} problem on directed graphs in the \emph{decremental} setting, where the graph undergoes updates that guarantee that the optimal cost is non-decreasing. Henceforth, decremental updates consist of edge deletions, cost increases, and capacity decreases.
The minimum-cost flow problem generalizes the $s$-$t$ shortest path and the more general single-source shortest-path (SSSP) problem that have received significant attention in the decremental setting \cite{HenzingerKN14, HenzingerKN15, HenzingerKN18, BernsteinGW20, BernsteinPW19, bernstein2022deterministic}, which are all not known to admit almost-linear-time algorithms, even against oblivious adversaries. In this paper, we give almost-linear-time algorithms for several problems in decremental graphs, which are primarily derived by solving the more general problem of \emph{decremental thresholded min-cost flow}.
\begin{definition}
\label{def:threshold}
The thresholded min-cost flow problem is defined on a directed graph $G = (V, E)$ with capacities $\uu$ and costs $\cc$, undergoing decremental updates (edge deletions, edge capacity decreases, and cost increases) along with a threshold $F$ and demands $\dd \in \R^V$. 
A dynamic algorithm solves the problem if, after each update, the algorithm outputs whether there is a feasible flow $\ff$ routing demand $\dd$ with cost $\cc^\top \ff$ at most $F$, or answers that no such flow exists.
\end{definition}
The thresholded min-cost flow problem in \emph{incremental} graphs (undergoing edge insertions) was recently shown to have an almost-linear-time algorithm in \cite{chen2023almost}.
In this paper, we show that the decremental version can also be solved in almost-linear-time (see \Cref{thm:main} for a formal statement).
\begin{informaltheorem}
\label{ithm:main}
There is a deterministic algorithm that solves the decremental thresholded min-cost flow problem on graphs with $m$ edges initially, undergoing $Q$ updates in total time $(m+Q)m^{o(1)}$, provided that costs, capacities, and demands are integral and polynomially bounded in $m$.
\end{informaltheorem}
This result and its extensions give almost-linear-time deterministic algorithms for decremental approximate min-cost flow value, single-source reachability, strongly connected component maintenance, and $s$-$t$ distance. 

Towards proving this result, let us recall the approach of \cite{chen2023almost}, which builds on the almost-linear-time min-cost flow algorithm of \cite{chen2022maximum}. 
The algorithm of \cite{chen2022maximum} used an $\ell_1$-based interior point method (IPM) to solve min-cost flow via a sequence of dynamic \emph{min-ratio cycle} problems, with approximation quality $\alpha = m^{o(1)}$. Later, \cite{vdBrand23incr} showed that giving an algorithm with amortized $m^{o(1)}$ update time for approximate dynamic min-ratio cycle against adaptive adversaries (which was not achieved in \cite{chen2022maximum}) suffices for incremental thresholded min-cost flow. Such a data structure for dynamic min-ratio cycle was developed in \cite{chen2023almost}.

There is one key difference between the incremental and decremental settings: a feasible flow $\ff$ continues to be feasible under edge insertions, but not under edge deletions. To handle this, we instead work with a dual version of the min-cost flow problem. More precisely, we first give a standard reduction in \Cref{sec:reduceToTrans} between min-cost flow and transshipment: $\min_{\BB^\top\ff=\dd,\ff \ge 0} \cc^\top\ff$, where $\BB$ is the edge-vertex incidence matrix of the underlying graph $G$. The dual of this problem, computed via strong duality (see \Cref{lem:transDual}), is 
\begin{align}
    \max_{\cc-\BB\yy\ge0} \dd^\top\yy.
    \label{eq:dual}
\end{align} 
Note that in this dual formulation, if a solution $\yy$ is feasible, i.e., $\cc-\BB\yy \ge 0$, then it continues to be feasible after an edge deletion or cost increase. It turns out that $\l \dd, \yy\r$ is also monotone increasing in the transshipment instance under capacity decreases (see \Cref{lemma:dec}). Thus, it is natural to work with the dual problem in the decremental setting.

To give an almost-linear-time algorithm for solving \eqref{eq:dual}, we broadly follow the approach set forth by \cite{chen2022maximum}. We first design an $\ell_1$-IPM which solves \eqref{eq:dual} via a sequence of $m^{1+o(1)}$ dynamic \emph{min-ratio cut} problems (see \Cref{sec:IPM}), defined below.
\begin{definition}[$\alpha$-approximate dynamic min-ratio cut]
\label{def:minratiocut}
The dynamic min-ratio cut problem is defined on an undirected graph $G$ with capacities $\uu \in \R^E_{\ge0}$, and vertex gradient $\gg \in \R^V$. At each time step, the gradient of a single vertex or the length of a single edge may be updated, in a fully-dynamic manner.
We let $\BB$ be the edge-vertex incidence matrix of $G$ after assigning an arbitrary orientation to each edge.

A dynamic algorithm solves the problem, if after the $i$-th update, it identifies a cut $\zz \in \{0, 1\}^V$, such that
\[ \frac{\l \gg, \zz\r}{\|\UU\BB\zz\|_1} \le \frac{1}{\alpha} \min_{\pphi \neq 0} \frac{\l \gg, \pphi\r}{\|\UU\BB\pphi\|_1}. \]
\end{definition}
It is worth pointing out that $\min_{\pphi \neq 0} \frac{\l \gg, \pphi\r}{\|\UU\BB\pphi\|_1}$ is a non-positive quantity by symmetry. 
We also stress that the min-ratio cut problem does not depend on the orientations chosen for edges when defining the edge-vertex incidence matrix $\BB$.
Our main technical contribution is an algorithm that solves dynamic min-ratio cut in amortized $m^{o(1)}$ time with $\alpha = m^{o(1)}$ approximation, under fully-dynamic updates against adaptive adversaries.

To solve the min-ratio cut problem approximately, we fully-dynamically maintain an $\ell_{\infty}$-oblivious routing for $G$ which is realized by a single tree $T$, often referred to as a \emph{tree cut sparsifier}.
We then show that on this tree, we can solve the min-ratio cut problem, and these cuts are good approximate solutions to the min-ratio cut problem on $G$.
We give a formal definition of these tree cut sparsifiers, since they are crucial to our result. 

%
\begin{definition}[Tree Cut Sparsifier]
\label{def:tree_cut_sparsifer_intro}
Given graph $G = (V,E, \uu)$, a \emph{tree cut sparsifier} $T = (V', E', \uu')$ of quality $q$ is a tree graph with $V \subseteq V'$ such that for every pair of disjoint sets $A, B \subseteq V$, we have that $\mincut_{G}(A, B) \leq \mincut_{T}(A, B) \leq q \cdot \mincut_{G}(A, B)$.
\end{definition}

Tree cut sparsifiers, in turn, are associated with dynamic expander hierarchies as introduced in \cite{goranci2021expander}. Loosely speaking, an expander hierarchy computes an expander decomposition of an underlying graph, contracts each expander piece to a single vertex, and recursively computes more expander decompositions, contractions, etc. This naturally induces a tree structure, which \cite{goranci2021expander} proves is a tree cut-sparsifier of quality $q = m^{o(1)}$ in unit capacity graphs. We give the first non-trivial algorithm for maintaining dynamic expander hierarchies and thus tree cut sparsifiers in \emph{capacitated} graphs. In fact, our algorithm is optimal up to subpolynomial factors.

\begin{informaltheorem}\label{infthm:treecutSparsifier}
Given an $m$-edge graph $G=(V, E, \uu)$ with polynomially bounded capacities that undergoes $\tilde{O}(m)$ edge insertions/deletions, then there is a deterministic algorithm that maintains a tree cut sparsifier $T$ of quality $m^{o(1)}$ in total update time $m^{1+o(1)}$.
\end{informaltheorem}

\subsection{Comparison to Earlier Minimum-Cost Flow Algorithms}

Our results build on 
the $\ell_1$-interior point method introduced in the first almost-linear time minimum-cost flow algorithm \cite{chen2022maximum}.
The primal $\ell_1$-IPM of \cite{chen2022maximum} and later works \cite{vdBrand23incr,brand2023incremental, chen2023almost} requires  solving a dynamic min-ratio cycle problem.
This problem is solved using data structures that fundamentally center around distance approximation in undirected graphs.
Our dual $\ell_1$-IPM requires us to solve a dynamic min-ratio cut problem, which instead requires cut approximation in undirected graphs.
The dual perspective turns out to be very natural in retrospect, and has two striking consequences:
Firstly, our dual approach enables us to solve decremental graph problems, similar to how incremental graph problems were solved in \cite{vdBrand23incr,brand2023incremental,chen2023almost} using primal algorithms, essentially  because dual solutions stay feasible under edge deletions while primal solutions stay feasible under edge insertions.
Secondly, the dual approach yields methods based on cut geometry instead of distance geometry, motivating us to develop fully-dynamic tree cut sparsifiers for weighted graphs, a powerful data structure for answering cut queries.
Notably, our cut data structures are substantially simpler than earlier approaches.
We expand on this comparison in Section~\ref{subsec:app} below.

In seeking to develop our dynamic cut approximation data structures, we encounter challenges that are morally similar to those of \cite{chen2022maximum, BrandCPKLPSS23, chen2023almost} which developed extensive new machinery for maintaining fully-dynamic low-stretch trees and $\ell_1$-oblivious routings in weighted graphs.
While fast dynamic algorithms to maintain LSSTs for unit capacity graphs existed previously \cite{forster2019dynamic, chechik2020dynamic, forster2021dynamic}, a central technical contribution in each of \cite{chen2022maximum, BrandCPKLPSS23, chen2023almost} is a dynamic algorithm to maintain LSSTs or
$\ell_1$-oblivious routings respectively in capacitated graphs.
This turns out to require a very different set of tools, and the resulting algorithms deviate heavily from algorithms designed for unit capacity graphs,
and instead build on ideas from \cite{M10,S13,KLOS14,rozhon2022undirected}.

We are faced with a similar challenge in constructing fully-dynamic tree cut sparsifiers.
A striking prototype data structure for the unit capacity case was built via the expander hierarchy in \cite{goranci2021expander},
but their methods face major obstacles in extending to the weighted case.
Our construction is motivated by their result but takes as its starting point a later generation of expander decomposition methods \cite{hua2023maintaining, sulser2024}.
These methods yield particularly clean expander hierarchies in unweighted decremental graphs, and we show how to extend these methods to weighted graphs using a new reduction from weighted to unweighted graphs in this setting.
Using \emph{core graph} techniques motivated by \cite{Madry10, chen2022maximum, chen2023almost}, we finally reduce the fully-dynamic tree cut sparsifier problem to the decremental case.

\subsection{Applications}
\label{subsec:app}

\paragraph{Application \#1: Faster static min-cost flow.} Somewhat surprisingly, the expander hierarchy data structure has fewer recursive levels than those for min-ratio cycle. 
This results in a faster runtime for static min-cost flow for both randomized and deterministic algorithms.
In particular, we give a randomized algorithm that statically solves exact min-cost flow on graphs with polynomially bounded costs and capacities in time $m \cdot e^{O((\log m)^{3/4} \log \log m)}$, and a deterministic version that runs in time $m \cdot e^{O((\log m)^{5/6} \log \log m)}$. This should be compared to randomized $m \cdot e^{O((\log m)^{7/8} \log \log m)}$ time \cite{chen2022maximum}, and deterministic $m \cdot e^{O((\log m)^{17/18} \log \log m)}$ time \cite{BrandCPKLPSS23} respectively.

Our algorithm initially solves the dual \eqref{eq:dual} and is able to use the final IPM dual solution to extract an optimal flow (see \Cref{sec:flow}).
Our approach is arguably the simplest almost-linear time algorithm for computing minimum-cost flows yet:
Our data structure only needs one main component, namely a fully-dynamic tree cut sparsifier, obtained from our dynamic expander hierarchy. In the randomized setting, the tools required to implement this expander hierarchy primarily involve a direct reduction from capacitated expander decomposition to the unit capacity setting.
Finally, on top of this, we build a data structure for detecting the best tree cut, using standard techniques.
%
In contrast, the data structure for solving min-ratio cycle in the first  almost-linear time min-cost flow algorithm of \cite{chen2022maximum} is quite involved.
It relies on complex fully-dynamic spanners, core graphs, all-pairs shortest path data structures, and a delicate restarting procedure to manage the interaction between a non-fully-adaptive data structure and its `adversary' coming from the interior point method.
This need for reasoning about the interaction between a data structure and an adversary was removed in \cite{chen2023almost}, which gave a (deterministic) fully-adaptive data structure for solving min-ratio cycle problems, but introduced other complexities by using extensive machinery to maintain fully-dynamic $\ell_1$-oblivious routings using dynamic terminal spanners and low-diameter trees \cite{KMP23}.

\paragraph{Application \#2: Decremental min-cost flow.} By designing an $\ell_1$-IPM for \eqref{eq:dual}, and an efficient data structure for min-ratio cuts (\Cref{def:minratiocut}), we show the following. 
\begin{restatable}{theorem}{mainTheorem}
\label{thm:main}
There is a randomized algorithm that given a decremental graph $G = (V, E, \uu, \cc)$ with integer capacities $\uu$ in $[1, U]$ and integer costs $\cc$ in $[-C, C]$, with $U, C \le m^{O(1)}$, where $m$ is the initial number of edges in $G$, a demand $\dd \in \Z^V$, and parameter $F \in \R$ reports after each edge deletion, capacity decrease, or cost increase, whether there is a feasible flow $\ff$ with cost $\cc^\top \ff$ at most $F$. Over $Q$ updates, the algorithm runs in total time $(m+Q) \cdot e^{O((\log m)^{3/4} \log \log m)}$, and can be made deterministic with time $(m+Q) \cdot e^{O((\log m)^{5/6} \log \log m)}$.
\end{restatable}
A corollary of \Cref{thm:main} is an algorithm for approximately maintaining the flow value.
Because \Cref{thm:main} solves a thresholded problem, it succeeds against adaptive adversaries. Because \Cref{thm:approx} reduces to the thresholded problem, it also succeeds against adaptive adversaries.
\begin{theorem}
\label{thm:approx}
There is a randomized algorithm that given a decremental graph $G = (V, E, \uu, \cc)$ with integer capacities $\uu$ in $[1, U]$ and integer costs $\cc$ in $[1, C]$, with $U, C \le m^{O(1)}$, where $m$ is the initial number of edges in $G$, a demand $\dd \in \Z^V$, maintains a $(1+\eps)$-approximation to the cost of the current min-cost flow. Over $Q$ updates, the algorithm runs in total time $\eps^{-1}(m+Q) \cdot e^{O((\log m)^{3/4} \log \log m)}$, and can be made deterministic in time $\eps^{-1}(m+Q) \cdot e^{O((\log m)^{5/6} \log \log m)}$.
\end{theorem}
\begin{proof}
We run a thresholded decremental min-cost flow algorithm (\Cref{thm:main}) for thresholds $F = (1+\eps)^i$. Note that the cost the min-cost flow is monotonically increasing because the graph is decremental. The cost is lower bounded by $1$, and upper bounded by $mCU$, so it suffices to set $i = 0, 1, \dots, O(\eps^{-1} \log(mCU))$. The result thus follows from \Cref{thm:main}.
\end{proof}
This also implies approximation and threshold algorithms for maintaining the value of the decremental maximum flow and the size of a weighted bipartite matching. It is worth noting that the dependence on $\eps$ is optimal under the online matrix-vector ($\mathsf{OMv}$) conjecture \cite{HenzingerKNS15}. Indeed, exact decremental matching in unweighted graphs requires time at least $mn^{1-o(1)}$ under $\mathsf{OMv}$. Our result should be compared to previous algorithms \cite{bernstein2022deterministic,BKS23,JambulapatiJST22}, with runtimes\footnote{We use $\tilde{O}(\cdot)$ to hide $\polylog(m)$ factors.} $\O(m\eps^{-4})$, $\O(m\eps^{-3})$, and $m^{1+o(1)}\eps^{-2}$ respectively.

\paragraph{Application \#3: Deterministic decremental single-source reachability and strongly connected components.} A long line of work resulted in near-linear time algorithms for decremental single-source reachability (SSR) and strongly connected components (SCC) \cite{HenzingerKN14,HenzingerKN15,ChechikHILP16,ItalianoKLS17,BernsteinPW19}. All these algorithms are randomized. They technically work against adaptive adversaries because the SCC decomposition or reachability structure does not reveal any randomness. However, they do not work against ``non-oblivious'' adversaries that can see all internal randomness of the algorithm.
In general, non-oblivious or even deterministic algorithms are often more desirable so that they can be used within an optimization framework (such as IPMs or multiplicative weights).
To the best of our knowledge, the current fastest deterministic algorithms for SSR and SCC require time $mn^{1/2+o(1)}$, achieved by \cite{bernstein2020deterministic}.\footnote{\cite{bernstein2020deterministic} claims a runtime of $mn^{2/3+o(1)}.$ Using the almost-linear-time deterministic max-flow algorithm from \cite{BrandCPKLPSS23} to speed up a routine to embed directed expanders in \cite{bernstein2020deterministic}, improves the runtime to $mn^{1/2+o(1)}$.} We improve this runtime to $m^{1+o(1)}$.

\begin{restatable}{theorem}{thmSCC}
\label{thm:scc}
There is a deterministic algorithm that given a directed graph $G = (V, E)$ undergoing edge deletions, explicitly maintains the strongly connected components of $G$ in total time $m \cdot e^{O((\log m)^{5/6} \log \log m)}$.
\end{restatable}
There is a simple reduction from SSR to SCCs.
\begin{corollary}
\label{cor:ssr}
There is a deterministic algorithm that, given a directed graph $G = (V, E)$ undergoing edge deletions and vertex $s \in V$, explicitly maintains the set of vertices reachable from $s$ in $G$ in total time $m \cdot e^{O((\log m)^{5/6} \log \log m)}$.
\end{corollary}
\begin{proof}
For a graph $G = (V, E)$ and vertex $s \in V$, consider the graph $\hat{G}$ which contains the edges in $E$, plus $(t, s)$ for $t \in V$. The SCC containing $s$ in $\hat{G}$ is exactly the set of vertices reachable from $s$. If $G$ is decremental, then so is $\hat{G}$. Thus, the result follows from \Cref{thm:scc}.
\end{proof}
It is worth mentioning that the previous deterministic result in \cite{bernstein2020deterministic} allows for querying paths between vertices in the same SCC, in time proportional to the length of the returned path. We do not currently see how to use our methods to achieve this.

\paragraph{Application \#4: Decremental $s$-$t$ distance.} The $s$-$t$ shortest path problem is of particular interest, partly because a classic algorithm of Garg and K\"{o}neman \cite{GargK07} shows that a data structure which solves approximate $s$-$t$ shortest path in decremental directed graphs (against an adaptive adversary, with path reporting) can be used to design a high-accuracy maximum flow algorithm with nearly the same runtime. Their algorithm is based on multiplicative weights, and hence initially only achieves a constant factor approximation. However, it works in directed graphs, so one can take the residual graph and repeat the argument to boost to high accuracy. Versions of this MWU framework have been instantiated in several settings \cite{Madry10,bernstein2022deterministic}. Even though we now know an almost-linear-time maximum flow algorithm, such an approach may provide an alternate algorithm not based on IPMs. Recently Chuzhoy and Khanna achieved a $n^{2+o(1)}$ runtime for bipartite matching via this approach \cite{ChuzhoyK2024,ChuzhoyK2024b}, by leveraging specific properties of the residual graphs encountered in a bipartite matching algorithm. Curiously, \Cref{thm:main,thm:approx} indeed give an almost-linear-time algorithm for reporting the distance of the decremental $s$-$t$ shortest path, though not a witness approximate shortest path itself. However, the algorithm uses an IPM, so even having access to a witness path would not lead to a more ``combinatorial'' maximum flow algorithm based only on MWU.

Before our result, the previous best-known runtimes against oblivious adversaries were $\O(n^2)$ in dense graphs \cite{BernsteinGW20} and $\O(mn^{3/4})$ in general \cite{GutenbergW20a}, and deterministically/adaptively only runtimes of $n^{2+2/3+o(1)}$ and $O(mn)$ are known \cite{bernstein2020deterministic, shiloach1981line}
Our result is deterministic and hence works against adaptive adversaries. It should be noted that these prior works solve the more general problem of single-source shortest path, i.e., approximate shortest path lengths from a source $s$ to every other vertex. They also support reporting approximate shortest paths. 
We do not know how to achieve either currently, for reasons similar to the ones discussed above regarding why we cannot report paths for SSR and SCC. However, we are mildly optimistic that it may be achievable with additional insights.

\paragraph{Application \#5: Dynamic flow algorithms.} Our data structure to maintain the expander hierarchy and tree cut sparsifier runs in graphs undergoing edge insertions and deletions with polynomially bounded edge capacities with randomized amortized update time and approximation quality $e^{O(\log^{3/4}\log\log m)}$. We further give a deterministic algorithm that achieves amortized update time and approximation quality $e^{O(\log^{5/6}\log\log m)}$. These bounds match the respective runtimes claimed in \cite{goranci2021expander} but extend their result also to capacitated graphs.

By the same reductions as in \cite{goranci2021expander}, we obtain the first algorithm with subpolynomial update time and approximation ratio for various important flow and cut problems.
\begin{theorem}
There is a deterministic algorithm on a capacitated $m$-edge graph undergoing edge insertions and deletions with amortized update time $m^{o(1)}$ that can return an $m^{o(1)}$-approximation to queries for the following properties:
\begin{itemize}
    \item $s$-$t$ maximum flow, $s$-$t$ minimum cut for any input pair $(s,t) \in V^2$;
    \item lowest conductance cut, sparsest cut; and
    \item multi-commodity flow, multi-cut, multi-way cut, and vertex cut sparsifiers.
\end{itemize}
The former two queries are answered in worst-case time $\tilde{O}(1)$, the last type of queries are answered in time $ \tilde{O}(k)$ where $k$ is the number of multi-commodity flow pairs; $k$ is the number of required sets in the multi-cut; $k$ is the number of terminals in the multi-way cut; or $k$ is the number of terminal vertices over which the vertex sparsifier is required.
\end{theorem}

Previously, similar results were obtained in unit capacity graphs by \cite{goranci2021expander}. \cite{CGHPS20} gave algorithms  for the first problem that achieve nearly-logarithmic quality while achieving sub-linear update time $\tilde{O}(n^{2/3})$ in an $n$-vertex graph against an oblivious adversary and $\tilde{O}(m^{3/4})$ time against an adaptive adversary. In \cite{incrTreeCutSparsifier}, a deterministic algorithm with $m^{o(1)}$-approximation quality and update time was given which works in capacitated graphs undergoing edge insertions only.

\subsection{Related Work}

\paragraph{Dual-based flow algorithms.} The work \cite{HenzingerJPW23} provided a cut toggling alternative to the cycle toggling Laplacian solver of \cite{KelnerOSZ13}. For the problem of decremental approximate bipartite matching, \cite{BKS23} provided an MWU algorithm on dual vectors that run in nearly-linear time. Additionally, \cite{zuzic2021simple} (see also \cite{Li20}) gave a framework for undirected transshipment that was partially based on adjusting dual variables.
The dual perspective is also crucial for the communication complexity of transshipment \cite{BlikstadBEMN22}.

\paragraph{Previous min-cost flow algorithms.} Following a long line of work \cite{ChristianoKMST11,Madry13,S13,KLOS14,LeeS14,Madry16,Peng16,van2020bipartite,van2021minimum,GaoLP21,bernstein2022deterministic, BrandGJLLPS22}, the work \cite{chen2022maximum} gave an almost-linear time algorithm for solving minimum-cost flow exactly in graphs with polynomially bounded integral costs and capacities. A series of works since then \cite{vdBrand23incr, BrandCPKLPSS23, brand2023incremental, chen2023almost} has made the algorithm deterministic and has given an algorithm for maintaining a minimum-cost flow in incremental graphs \cite{chen2023almost}. Earlier works primarily used electrical flows to make progress, and recent works use approximate minimum-ratio cycles. Our work provides an alternative approach that instead solves minimum-ratio cut problems. Our algorithm has fewer recursive layers and results in a faster runtime for exact minimum-cost flow in both the randomized and deterministic settings.

\paragraph{Decremental graph algorithms.}
One of the first decremental graph algorithms was given in the 80's, when Even and Shiloach gave an algorithm to maintain decremental BFS trees~\cite{shiloach1981line}.
Since then, there has been significant work on maintaining fundamental properties of decremental graphs.
Decremental $s$-$t$ or single source shortest path (SSSP) is particularly important problems that have applications such as efficient implementation of numerical methods on graphs~\cite{GargK07, Madry10, bernstein2022deterministic, chen2023simple}.
On undirected decremental graphs, it is known how to achieve a $(1+\eps)$-approximation ratio in $m^{1+o(1)}$ total time deterministically~\cite{bernstein2016deterministic, bernstein2017deterministic, HenzingerKN18, CK19, gutenberg2020deterministic,ChuzhoyS21, bernstein2022deterministic, KMP23}.
For directed graphs, achieving an almost-linear runtime remains open and the state of the art is either $n^{2+2/3+o(1)}$ deterministically, $\O(m^{3/4}n^{5/4})$ assuming adaptive adversaries, or $\O(m^{2/3}n^{4/3})$ assuming oblivious adversaries~\cite{bernstein2013maintaining, HenzingerKN14, HenzingerKN15, BernsteinGW20, bernstein2020deterministic, GutenbergW20a}.
With a large $m^{o(1)}$ approximation factor, deterministic almost-linear total time is achievable for even fully-dynamic all-pair shortest path (APSP) on undirected graphs~\cite{chechik2018near, ChuzhoyS21, Chu21, bernstein2022deterministic}.

Matchings are another graph property that have attracted significant attention from the dynamic graph algorithms community.
In decremental graphs, it is possible to maintain $(1+\eps)$-approximate maximum (weighted or cardinality) matching on either bipartite~\cite{bernstein2020deterministic, JambulapatiJST22} or general graphs~\cite{assadi2022decremental, dudeja2023decremental, chen2023entropy}.

Most of the aforementioned results are purely combinatorial and are focused on maintaining discrete structures such as graph decompositions, neighborhood coverings, and search trees.
On the other hand, our work maintains the solution through the lens of feasibility and optimality of a continuous optimization problem.
Such idea also appears in some previous works such as incremental matchings~\cite{Gupta14}, decremental matchings~\cite{JambulapatiJST22,chen2023entropy}, partially dynamic LPs~\cite{bhattacharya2023dynamic}, and decremental max eigenvectors~\cite{adil2024decremental}.









\paragraph{Dynamic flow algorithms.}
As a direct implication of the Ford-Fulkerson's maxflow algorithm, one can maintain exact max flow on fully dynamic unweighted graphs with $O(m)$ update time (see \cite{gupta2018simple} for a discussion on the matter).
On planar graphs, an improved update time of $\O(n^{2/3})$ can be achieved~\cite{italiano2010improved, karczmarz2021fully}.
However, under the strongly exponential time hypothesis (SETH), there is no sublinear update time algorithm for partially dynamic general graphs~\cite{dahlgaard:LIPIcs.ICALP.2016.48}.
As a result, much attention is devoted to maintaining approximate solutions.

In the fully dynamic case, $m^{o(1)}$-approximation ratio with $m^{o(1)}$ update time can be achieved via expander hierarchy on unit capacity, undirected graphs~\cite{goranci2021expander}. 
\cite{CGHPS20} shows how to maintain a $\O(1)$-approximation on capacitated graphs.
Our dynamic tree cut sparsifier improves these results to $m^{o(1)}$-approximation in $m^{o(1)}$-amortized time on capacitated graphs.
In the incremental case, $(1+\eps)$-approximate solutions can be maintained with $m^{1+o(1)}\eps^{-1}$ total time for undirected $p$-norm flows as well as directed min-cost flows~\cite{brand2023incremental, chen2023almost, vdBrand23incr}.
For unweighted graphs, a runtime of $m^{3/2+o(1)}\eps^{-1/2}$ total time was previously achieved by~\cite{GH22}.
The algorithm of \cite{GH22} can also maintain exact max flows on incremental unweighted graphs in a $n^{5/2+o(1)}$ total time, which corresponds to a sublinear update time when the graph is sufficiently dense.

\paragraph{Lower bounds. }
In this paragraph, we give a more detailed discussion of related lower bounds. Since the current lower bounds do not distinguish between the incremental and decremental settings, we refer the reader to \cite{chen2023almost} for an analogous and slightly more expansive discussion in the incremental setting.  

\begin{itemize}
    \item \underline{Flows and Matchings:} Under the online matrix-vector (OMv) conjecture \cite{HenzingerKNS15}, there are bipartite graphs where performing $\Theta(n^2)$ deletions and $\Theta(n)$ size queries requires $\Omega(n^{3-\delta})$ total time for any fixed constant $\delta > 0$ to maintain exact matching size \cite{dahlgaard:LIPIcs.ICALP.2016.48}. Therefore, $\Omega(n^{2 - \delta})$ amortized update time is necessary for $\Theta(n)$ updates and one size query under OMv. Thus, our dependence on $\epsilon$ is optimal for algorithms with sub-polynomial update time because a $(1 - \frac{1}{n + 1})$-approximate matching is a maximum cardinality matching.  

    Furthermore, under the strongly exponential time hypothesis (SETH), every decremental algorithm for the weighted and directed exact maximum flow value problem on a sparse graph with $n$ vertices requires $O(n^{1 - \delta})$ amortized update time  \cite{dahlgaard:LIPIcs.ICALP.2016.48}. 
    
    \item \underline{SCCs:} We discuss the hardness of deciding if a directed graph contains a cycle in the fully-dynamic and worst-case decremental settings. 
    In the fully-dynamic setting $\Theta(n^2)$ updates and $\Theta(n)$ cycle detection queries take $\Omega(n^{3-\delta})$ time for an arbitrary constant $\delta > 0$ under OMv \cite{HenzingerKNS15}. 
    
    By a straightforward reduction to deciding if the $s$-$t$ shortest path has length $3$ or $5$, there are graphs for which $\Theta(n)$ edge deletions and a single cycle detection query take total time $\Omega(n^{2-\delta})$ under OMv \cite{HenzingerKNS15}. This rules out sub-linear worst-case update time for decremental algorithms.
    \item \underline{Decremental $s$-$t$ Shortest Path:} Under the OMv conjecture the exact decremental $s$-$t$ shortest distance problem requires amortized update time $m^{0.5 - \delta}$ on an unweighted graph with $n$ vertices and $m = O(n^2)$ edges for any fixed $\delta > 0$ \cite{HenzingerKNS15}. We remark that our dependence on $\epsilon$ is optimal for algorithms with sub-polynomial update time because a $(1 + \frac{1}{n})$-approximate shortest distance on a unweighted graph is an exact shortest distance. Furthermore, algorithms with sub-polynomial worst-case update time are ruled out for obtaining a $3/5 - \delta$ approximation under OMv, again via distinguishing $s$-$t$ distances $3$ and $5$ \cite{HenzingerKNS15}. 
\end{itemize}

\paragraph{Paper Organization.}
In \Cref{sec:overview}, we give an overview of our algorithm. Then, we describe our algorithm for maintaining a tree cut sparsifier in \Cref{sec:tree_cut_spars}. In \Cref{sec:min_ratio}, we show that tree cut sparsifers can be used to detect min-ratio cuts, and we describe that approximate edge potential differences can be maintained efficiently via standard techniques in \Cref{sec:detection}. Finally, in \Cref{sec:IPM}, we show that such a min-ratio cut data-structure suffices to solve decremental threshold min-cost flow. 
\section{Overview}
\label{sec:overview}

To convey the workings of our algorithm, it is natural to present the sections in
a top-down manner to better highlight and motivate why we need to solve certain subproblems.
The later main text will give the formal proof in the bottom-up order, as our proofs build on the precise properties and guarantees of the subroutines derived before.

\subsection{Min-Cost Flow, Transshipment and its Dual}
\label{sec:transship_flow_dual}

Our algorithm for min-cost flow first reduces the min-cost flow problem to transshipment on a sparse bipartite graph $G = (V, E, \cc)$ with some vertex demands $\dd$ (\Cref{sec:reduceToTrans}). The transshipment problem 
\begin{equation}
    \label{eq:transship}
    \min_{\BB^\top \ff = \dd, \ff \geq 0} \cc^\top \ff
\end{equation}
is a special case of min-cost flow where all the capacities are unbounded. Because our ultimate goal is to handle edge deletions, the reduction to this form does not address the central issue that arises for algorithms in flow space: Deleting an edge causes the current flow to no-longer route the demands. Therefore, we take the dual of \eqref{eq:transship} which translates the problem to voltage space (i.e. vertex potentials)
\begin{equation}
    \label{eq:transship_dual}
    \max_{\cc - \BB \yy \geq \veczero} \dd^\top \yy.
\end{equation}
This form is more amenable to edge deletions, since the vertex potentials $\yy$ remain feasible under edge deletions. Finally, we consider the thresholded version of \eqref{eq:transship_dual} and simply aim to decide if $\max_{\cc - \BB \yy \geq \veczero} \dd^\top \yy \geq F$ instead of maximizing the dual. 

\subsection{Solving the Thresholded Transshipment Dual via Min-Ratio Cuts}

In this section, we outline how to solve the transshipment dual problem by repeatedly solving the min-ratio cut problem on a fully dynamic graph $G = (V, E, \uu, \gg)$, where $\uu \in \R_{\geq 0}^E$ are best interpreted as edge capacities (different from the capacities in the original min-cost flow instance) and $\gg \in \R^V \perp \vecone$ are referred to  as vertex gradients. We denote $\UU = \diag(\uu)$. Then, the min-ratio cut problem is given by
\begin{equation}
    \label{eq:min_ratio_cut}
    \min_{\bDelta \in \R^V} \frac{\l\gg, \bDelta\r}{\norm{\UU \BB \bDelta}_1}. 
\end{equation}
Notice, that the solution of \eqref{eq:min_ratio_cut} is always negative and that it therefore maximizes the absolute value of the ratio.
We show that there is always a optimal solution $\bDelta = \pm \vecone_{C}$, i.e. there exist optimal $\bDelta$ which indicate cuts in the graph.
Furthermore, despite being used to solve the transshipment dual on a directed graph, this problem is undirected in that only the signs of the gradients depend on the side of the cut.

To show that \eqref{eq:min_ratio_cut} can be used to solve the transshipment dual problem \eqref{eq:transship_dual}, we closely follow the $\ell_1$-IPM framework introduced by \cite{chen2022maximum} for the first almost-linear time algorithm for minimum-cost flow, adapted to dual space, and apply it to the transshipment dual. Following \cite{chen2022maximum} we introduce a potential $\Phi : \R^V \to \R$
\begin{equation}
    \Phi(\yy) \defeq 100m\log(F - \l \dd, \yy \r) + \sum_{e = (u,v) \in G} (\cc(e) - (\BB\yy)(e))^{-\alpha}
\end{equation}
for $\alpha \approx 1/\log(mC)$ where $m$ denotes the initial number of edges in $G$ and all costs are integers in the interval $[-C, C]$. If a solution of cost $F$ exists, then the potential $\Phi(\yy)$ is unbounded and goes to $-\infty$ as $\l \dd, \yy \r$ approaches $F$.

The barrier $(\cdot)^{-\alpha}$ can be thought of as the more standard $\log(\cdot)$ barrier to ensure that $\yy$ remains feasible, but it penalizes approaching the boundary more harshly and thus ensures that $(\cc(e) - (\BB\yy)(e)) \geq 1/n^{\tilde{O}(1)}$ as long as $(\cc(e) - (\BB\yy)(e))^{-\alpha} \leq \tilde{O}(m)$. This ensures that the bit-complexity remains bounded by $\tilde{O}(1)$, which directly follows from the following description of vertex gradients and edge capacities respectively.  
We let
\begin{equation*}
    \gg \defeq \grad \Phi(\yy) = \frac{-100m}{F - \l \dd, \yy \r} \dd + \alpha \BB^{\top} (\cc - (\BB\yy))^{-1 - \alpha}
\end{equation*}
and $\uu(e) \defeq (\cc(e) - (\BB\yy)(e))^{-1-\alpha}$ where the $-1 - \alpha$ exponent is applied to every element in the vectors separately. The Taylor-expansion 
\begin{equation*}
    \Phi(\yy + \bDelta) \approx \Phi(\yy) + \l \gg, \bDelta \r + \norm{\UU\BB \bDelta}_2^2 \le \Phi(\yy) + \l \gg, \bDelta \r + \norm{\UU\BB \bDelta}_1^2
\end{equation*}
implies that solving the min-ratio cut problem to $1/\kappa$ accuracy yields an update reducing the potential by approximately $1/\kappa^2$ if there is a solution to \eqref{eq:transship_dual} with cost $F$ because the optimum ratio is then $\approx 1$.

This can be turned into an algorithm for decremental transshipment with the following observations. First, if there is a feasible $\yy$ with $\l \dd, \yy \r \ge F$, then there exists a solution to the min-ratio cut problem that decreases the potentials by at least $m^{-o(1)}$ (\Cref{lem:existSmallRatioCut}). Thus if our min-ratio cut algorithm cannot find a good solution, we conclude that $\max_{\cc-\BB\yy\ge0} \l \dd,\yy \r < F$, and continue.

It is not difficult to initialize $\yy$ so that the potential is initially $\O(m)$, and when $\Phi(\yy) \le \O(m)$, one can show that $\l \dd, \yy \r \ge F - m^{-O(1)}$ (\Cref{lem:smallPot}). Finally, edge deletions cannot increase the potential, and each edge deletion only causes $O(1)$ updates to the gradients and capacities. Overall, the algorithm only makes $m^{1+o(1)}$ calls to the dynamic min-ratio cut data structure. We refer the reader to \Cref{sec:IPM} for a detailed description of the interior point method.

\subsection{Min-Ratio Cuts on Trees}

Despite its description involving an arbitrary update vector $\bDelta$ in \eqref{eq:min_ratio_cut}, the min-ratio cut problem always has a solution that updates along a single cut, i.e., we have
\begin{equation*}
    \min_{C \subseteq V} \frac{\l\gg, \vecone_C\r}{\norm{\UU \BB \vecone_{C}}_1} = \min_{\bDelta \in \R^V} \frac{\l\gg, \bDelta\r}{\norm{\UU \BB \bDelta}_1}
\end{equation*}
which explains the nomenclature and allows us to focus our efforts on cuts from here on out. We refer the reader to \Cref{lem:cut_suffices} for a short proof of this fact.

To describe how we repeatedly solve the min-ratio cut problem approximately on a fully dynamic graph $G$, we assume that the problem is posed on a dynamic tree $T$ instead. We will later reduce to this case using dynamic tree-cut sparsifiers, the main data structure we develop in this paper.

In an analogue to the cycle decomposition of circulations in flow-space, we next show by induction that it suffices to consider cuts induced by a single tree edge. Consider a min-ratio tree cut $C$ that cuts $k > 1$ tree edges. We show that there is a cut with at most $k - 1$ edges achieving at least as good quality. Because the graph is a tree there is at least one connected component $C'$ (after removing the cut edges) that is only incident to a single cut edge. Since shifting the $\bDelta = \vecone_{C}$ vector does not change its ratio, we may assume that this connected component receives value $1$, i.e., $\bDelta(v) = 1$ for $v \in C'$. Now notice that $\vecone_{C} = \vecone_{C'} + \vecone_{C \setminus C'}$ where $C'$ is a cut induced by removing a single edge, and $C \setminus C'$ is a cut induced by removing $k - 1$ edges. But then we obtain
\begin{equation}
    \label{eq:one_less_edge}
    \min\left(\frac{\l\gg, \vecone_{C'}\r}{\norm{\UU \BB \vecone_{C'}}}, \frac{\l\gg, \vecone_{C \setminus C'}\r}{\norm{\UU \BB _{C \setminus C'}}} \right) \leq  \frac{\l\gg, \vecone_{C'} + \vecone_{C \setminus C'}\r}{\norm{\UU \BB \vecone_{C'}} + \norm{\UU \BB \vecone_{C \setminus C'}}} = \frac{\l\gg, \vecone_{C}\r}{\norm{\UU \BB \vecone_{C}}}
\end{equation}
where the inequality follows from the well known fact that $\min(a/b, c/d) \leq \frac{a + c}{b + d}$ given $b, d > 0$. Iterating \eqref{eq:one_less_edge} directly yields that it suffices to consider cuts induced by single tree edges. 

Given this insight, it suffices to maintain the ratio achieved by every tree edge under updates to the tree, edge capacities, and gradients, where we are guaranteed that $\gg \perp \vecone$ at all times. It turns out that the tree-cut sparsifiers we maintain have hop diameter bounded by $\O(1)$\footnote{The hop-diameter of a graph is the diameter of its uncapacitated version.}. This allows us to maintain the quality of each single edge cut in $\O(1)$ time: whenever vertices $u$ and $v$ undergo an update in the form of an edge insertion, deletion, or gradient change, only edges in the path $T[u, v]$ connecting $u, v$ have their qualities change.

We refer the reader to \Cref{sec:min_ratio} for a detailed description of our min-ratio cut data structure on trees. This section also contains an additional component necessary to our overall algorithm. We must maintain approximations to the true gradient and capacities to know which edges to update in the dynamic min-ratio cut data structure, and this involves detecting edges which have accumulated large potential differences across the cuts we have returned. We build a standard data structure for this problem in \Cref{sec:detection}.

\subsection{Reducing to Trees via Tree-Cut Sparsifiers}

In this section, we explain our construction of tree-cut sparsifiers $T$ for a dynamic graph $G = (V, E, \uu)$. These are trees on a larger vertex set that capture every cut up to some multiplicative factor $q$. This allows us to approximate the min-ratio cut in $G$ with a tree cut up to a multiplicative loss $q$.

\begin{restatable}[(Tree/Forest) Cut Sparsifier]{definition}{TreeCutSparsifer}
\label{def:tree_cut_sparsifer}
Given graph $G = (V,E, \uu)$, a \emph{cut sparsifier} $G' = (V', E', \uu')$ of quality $q$ is a graph with $V \subseteq V'$ such that for every pair of disjoint sets $A, B \subseteq V$, we have that $\mincut_{G}(A, B) \leq \mincut_{G'}(A, B) \leq q \cdot \mincut_{G}(A, B)$. We say that $G'$ is a \emph{forest cut sparsifier} if $G'$ is a cut sparsifier and a forest graph; and we say $G'$ is a \emph{tree cut sparsifier} if $G'$ is a cut sparsifier and a tree graph. 
\end{restatable}

At a high level, our algorithm wishes to maintain an \emph{expander hierarchy} on a dynamic capacitated graph $G$, introduced by \cite{goranci2021expander}. Broadly, an expander hierarchy is constructed by first finding an expander decomposition of $G$. In fact, a stronger notion called \emph{boundary-linkedness} (\Cref{def:boundaryLinkedExp}) is necessary, but this generally follows from most expander decomposition constructions. Then each expander piece is contracted, and the algorithm then finds an expander decomposition on the contracted graph, and repeats. Note that this naturally leads to a tree structure. \cite{goranci2021expander} proves that this tree is a tree-cut sparsifier of quality $\O(1)^k/\phi$, where $k = O(\log_{1/\phi} m)$ is the number of layers in the expander hierarchy. This is $\O(1)^k/\phi \le m^{o(1)}$, for $\phi = 2^{-\sqrt{\log m}}$.

The work of \cite{goranci2021expander} showed how to maintain an expander hierarchy in unit capacity graphs. Our first goal is to extend this to capacitated graphs that only undergo edge \emph{deletions}. Later, we show how to construct a tree-cut sparsifier on fully dynamic graphs by using the \emph{core-graph} technique and batching. We start by discussing how to maintain an expander decomposition on a capacitated graph undergoing edge deletions.

\paragraph{Capacitated Decremental Expander Decomposition.}
An expander decomposition is a partition $\mathcal{X}$ of the vertices in $G = (V, E, \uu)$, such that for every $X \in \mathcal{X}$ the induced subgraph $G[X]$ is a $\phi/\tilde{O}(1)$ expander with respect to conductance, i.e., $\uu_G(S, X \setminus S)/\min(\vol_{G}(S), \vol_G(X \setminus S) \geq \phi$ for all $S \subseteq X$. Futhermore, the total capacity of the crossing edges is bounded with $\tilde{O}(\phi \cdot U^{\text{total}})$, where we denote the total capacity of all the edges in $G$ with $U^{\text{total}}$. 

While this problem has been studied before using more involved techniques \cite{li2021deterministic}, we give the simplest imaginable reduction to the uncapacitated setting. This is important for us because we require an additional property of the expander decomposition we maintain: the vertex sets of the expanders refine over time, and the total volume of all edges that are ever cut is bounded by $\tilde{O}(\phi \cdot U^{\text{total}})$ over all edge deletions.

We first fix a value $U^{\text{cutoff}} = \phi \cdot \frac{U^{\text{total}}}{m}$, and let $G'$ be the sub-graph of $G$ that only contains edges with capacity at least $U^{\text{cutoff}}$, and additionally contains self loops of capacity $\vol_G(v)$ for every vertex $v$. Notice that computing a weighted expander decomposition for this graph $G'$ suffices, since the same decomposition has at most $\phi U^{\text{total}}$ extra crossing edge capacity in $G$, and we have $\vol_{G}(S) \leq \vol_{G'}(S)$ for every set $S \subseteq V$ due to the additional self loops. 

We now exploit that all the non-self loop edges in $G'$ have high capacity to replace $G'$ by an unweighted multi-graph $G''$. We simply replace every edge with  $\ceil{\uu(e)/U^{\text{cutoff}}}$ uncapacitated multi-edges. Notice that the capacity of every cut in $G''$ 2-approximates the capacity of the cut in $G'$ (after scaling with $U^{\text{cutoff}}$), and that the volume of $G''$ is lower bounded by the volume of $G'$. Furthermore $G''$ only contains $O(m/\phi)$ edges. 

An uncapacitated decremental expander decomposition that refines over time under edge deletions can then be computed using recent works on expander decompositions, specifically \cite{sulser2024} adapted using ideas from \cite{hua2023maintaining} to enable vertex splits and self-loop insertions. This \emph{refining property} of the expander decomposition then ensures that the total amount of capacity on all edges cut at any point in time is $\tilde{O}(\phi \cdot U^{total})$.

Overall, we have given an algorithm to maintain an expander hierarchy, and thus a tree-cut sparsifier of quality $2^{O(\sqrt{\log m} \log \log m)}$ in $m^{1+o(1)}$ time in decremental capacitated graphs.

\paragraph{Fully Dynamic Tree-Cut Sparsifiers. }
Finally, we reduce from the fully dynamic case to the decremental case using batching. To describe the main ideas used in our batching scheme, we consider a current tree-cut sparsifier $T$ of some graph $G$ that receives a batch of insertions $I$. We show that we can compute a new tree-cut sparsifier of $G \cup I$ in time proportional to $|I|$ without losing too much quality. Batching the updates appropriately then turns such an algorithm into a fully dynamic tree-cut sparsifier data structure. 

To compute a tree-cut sparsifier of $G \cup I$, we instead consider the graph $T \cup I$. It is not surprising that $T \cup I$ is a cut-sparsifier of $G \cup I$ given that $T$ is a tree-cut sparsifier of $G$. Then, we instantiate a set of terminals $B \subseteq V(T)$. Initially, every endpoint of an edge in $I$ is added to $B$. Then, $B$ is extended to a branch-free set, i.e. a set such that the set of paths $\mathcal{P}$ containing all tree paths $T[a, b]$ for $a, b \in B$ such that it does not intersect any other terminal is edge-disjoint. This extension can be achieved by doubling the size of the terminal set.

Then, we remove the minimum capacity edge from each such path and refer to the trees in the leftover forest as cores. Thereafter, the algorithm contracts the cores (forest pieces) and computes a tree-cut sparsifier on the contracted graph merely containing the identified min-capacity edges and the inserted edges in $I$. This graph contains approximately $|I|$ edges, and therefore computing the tree-cut sparsifier takes time roughly proportional to $|I|$. Then, this tree is mapped back to a tree on the whole graph via un-contracting the cores. 

Since the procedure described above loses an $q$ factor in quality every time it is applied, we make sure that the sequential depth $k$ of this operation in the final batching scheme handling insertions is very low, i.e., $q^k = \hat{O}(1)$. 

See \Cref{sec:tree_cut_spars} for a full description of our tree-cut sparsifier data structure. 
\section{Preliminaries}

\paragraph{Linear Algebra. } We denote vectors as lower case bold letters $\aa$, and matrices as upper case bold letters $\AA$. Given a vector $\aa \in \R^X$ and a subset $Y \subseteq X$ we let $\aa[X]$ denote the vector $\aa$ restricted to the coordinates in $X$, and we let $\aa(X) = \sum_{x \in X} \aa(x)$. For a vector $\uu \in \R^n$, we let $\diag(\uu) \in \R^{n \times n}$ denote the diagonal matrix with entries of $\uu$ on the diagonal.

\paragraph{Graphs.} We work with a capacitated input graph $G =(V,E,\uu)$ where $\uu$ is the function that assigns each edge $e \in E$ a capacity $\uu(e) \geq 1$. We define $\vol_G(v)$ for every vertex $v \in V$ as the weighted degree, i.e. $\vol_G(v) = \sum_{e \in E, v \in e} \uu(e)$ and denote by $\deg_G(v)$ the combinatorial degree, i.e. $\deg_G(v) =  \sum_{e \in E : v \in e} 1$. We extend this notion to sets where $X \subseteq V$, $\vol_G(X) = \sum_{v \in X} \vol_G(v)$ and $\deg_G(X) = \sum_{v \in X} \deg_G(v)$. For uncapacitated graphs, note that degrees and volumes coincide.

For directed graphs, we let the in-degree of vertex $v$ be equal to the number of edges $(w, v)$ whose head is $v$, and we let the out-degree of $v$ be the number of edges $(v, w)$ whose tail is $v$. 

We say that a graph $G$ is a $\phi$-expander if for every $S \subseteq V$ with $\vol_G(S) \leq \vol_G(V)/2$, we have $\uu(E(S, V \setminus S)) \geq \phi \cdot \vol_G(S)$.

Given a tree $T$, we denote with $T[u, v]$ the unique tree path from vertex $u$ to vertex $v$. 

Finally, we define the mincut between two sets of vertices in a graph $G$.
\begin{definition}
Given a graph $G=(V,E, \uu)$ and two disjoint sets $A, B \subseteq V$, we denote by $\mincut_G(A,B)$ the minimum value $\uu(E_G(A', V \setminus A'))$ achieved by any set $A'$ with $A \subseteq A' \subseteq V \setminus B$.
\end{definition}
\section{Fully-Dynamic Tree Cut Sparsifiers}
\label{sec:tree_cut_spars}

The main graph-theoretic object in this paper is the notion of a tree cut sparsifier. 

\TreeCutSparsifer*

In this section, we show that tree cut sparsifiers can be maintained efficiently in a fully-dynamic graph. Previously, this result was only known for uncapacitated graphs \cite{goranci2021expander}. Our main result is summarized in \Cref{thm:mainTreeSparsifier}. 

\begin{restatable}{theorem}{mainTreeCutSparisifer}\label{thm:mainTreeSparsifier}
Given an $m$-edge graph $G=(V, E, \uu)$ where $\uu \in [1, U = m^{O(1)}]^E$. Let $G$ be undergoing up to $\tilde{O}(m)$ edge deletions/edge insertions and vertex splits. Then, there is a randomized algorithm that maintains a tree  $T = (V', E', c')$ undergoing insertions and deletions of edges and isolated vertices, such that $T$ is a tree cut sparsifier of quality $\gamma_q = 2^{O(\log^{3/4}(m) \log\log(m))}$ with total update time $m \cdot 2^{O(\log^{3/4}(m) \log\log(m))}$. The algorithm succeeds w.h.p.
\end{restatable}

We further augment the above theorem to maintain a dependency graph $H$ that allows us to track approximately which edges are in the cut induced by each tree edge of $T$. This graph $H$ is crucial in our final min-cost flow algorithm as it allows us to implicitly maintain flow and potentials in the IPM. 

\begin{restatable}{definition}{DirectedLayerGraphDef}
\label{def:directed_layer_graph}
    Given a tree cut sparsifier $T$ of quality $q$, a directed layer graph $H = (V_0 \cup V_1 \cup \cdots \cup V_k, E_H)$ has $k$ layers where $V_0$ has a vertex for each edge $e \in E$, and all edges $e_H \in E_H$ have their tail in $V_{i+1}$ and head in $V_i$ for some $0 \leq i < k$, such that every vertex $v \in V(H)$ has in-degree $d = O(\log^{c'} m)$ for some constant $c' > 0$.

    For every edge $e_T \in T$, let $E_{e_T}$ be the set of edges in $G$ that cross the cut induced by $T \setminus \{e_T\}$, i.e. let $A, B$ be the connected components of $T \setminus \{e_T\}$, then $E_{e_T} = E_G(A \cap V, B \cap V)$. Let $E'_{e_T}$ be the set of edges in $G$ whose corresponding vertices in $V_0$ are reached by the vertex $v_{e_T}$ that represents the edge $e_T$ in the graph $H$. Then, we have at any time that $E_{e_T} \subseteq E'_{e_T}$ and $\uu_G(E'_{e_T}) \leq q \cdot \uu_G(E_{e_T})$. 
\end{restatable}




\begin{restatable}{lemma}{DirectedLayerGraphRand}
\label{rem:directed_layer_graph}
The algorithm in \Cref{thm:mainTreeSparsifier} can be extended to explicitly maintain a directed layer graph $H = (V_0 \cup V_1 \cup \cdots \cup V_k, E_H)$ where $k = O(\log^{1/4}(m) \log\log (m))$. 

The additional total runtime for maintaining the graph $H$ is again $m \cdot 2^{O(\log^{3/4}(m) \log\log(m))}$. The total number of updates to $H$ consisting of insertions/deletions of edges and isolated vertices is bounded by $m \cdot 2^{O(\log^{3/4}(m) \log\log m)}$.
\end{restatable}

Finally, we discuss how to derandomize the above result at the cost of obtaining a slightly worse approximation guarantee and runtime.

\begin{theorem}
\label{thm:det_mainTreeSparsifier}
Given an $m$-edge graph $G=(V, E, \uu)$ where $\uu \in [1, U = m^{O(1)}]^E$. Let $G$ be undergoing up to $\tilde{O}(m)$ edge deletions/edge insertions and vertex splits. Then, there is a \textbf{deterministic} algorithm that maintains a tree cut sparsifier $T = (V', E', c')$ of quality $\gamma_q = 2^{O(\log^{5/6}(m) \log\log(m))}$ with total update time $m \cdot 2^{O(\log^{5/6}(m) \log\log(m))}$.
\end{theorem}

\begin{restatable}{lemma}{DirectedLayerGraphDet}
\label{rem:directed_layer_graph_det}
The deterministic algorithm in \Cref{thm:det_mainTreeSparsifier} can be extended to explicitly maintain a directed layer graph $H = (V_0 \cup V_1 \cup \cdots \cup V_k, E_H)$ where $k = O(\log^{1/6}(m) \log\log (m))$. 

The additional total runtime for maintaining the graph $H$ is again $m \cdot 2^{O(\log^{5/6}(m) \log\log(m))}$. The total number of updates to $H$ consisting of insertions/deletions of edges and isolated vertices is bounded by $m \cdot 2^{O(\log^{5/6}(m) \log\log(m))}$.
\end{restatable}

\begin{remark}
    \label{rem:DiamTreeSparsifier}
    The tree cut sparsifers maintained by \Cref{thm:mainTreeSparsifier} and \Cref{thm:det_mainTreeSparsifier} have hop diamter $\tilde{O}(1)$.
\end{remark}

For the rest of the section, we implicitly assume that all (dynamic) graphs $G$ under consideration are connected (at all times). We obtain our main result summarized in \Cref{thm:mainTreeSparsifier} in three steps: first, in \Cref{subsec:capTounCapExpander}, we give a reduction that allows us to maintain a decremental expander decomposition of \textbf{capacitated} graphs by using existing techniques to maintain an expander decomposition of a decremental, \textbf{un-capacitated} graph. We then show that we can maintain a tree cut sparsifier of a decremental graph via expander decompositions in \Cref{subsec:decrTreeSpars}. Finally, we reduce the problem of maintaining a tree cut sparsifier on a dynamic graph to a decremental graph problem in \Cref{subsec:fullydyntoDecrTreeSpars}. We then discuss how to derandomize our result to obtain \Cref{thm:det_mainTreeSparsifier} in \Cref{subsec:detTreeCutSparsifiers}.

\subsection{Decremental Expander Decompositions for Capacitated Graphs}
\label{subsec:capTounCapExpander}

In this section, we generalize a recent result about the maintenance of expander decompositions to graphs with capacities. We summarize our result in \Cref{thm:weightedExpDecomp} below. We point out that our proof techniques in this section can be used to obtain expander decompositions of directed, capacitated graphs, however, here we focus only on undirected graphs.

\begin{theorem}[Capacitated Expander Decomposition]\label{thm:weightedExpDecomp}
Given a parameter $0 <\phi \leq 1$ and a capacitated $m$-edge graph $G = (V,E, \uu)$, where $\uu \in [1, U]^E$ and $U$ being any positive number, undergoing a sequence of $\tilde{O}(m)$ updates consisting of edge deletions, vertex splits and self-loop insertions.

There is a randomized algorithm that explicitly maintains tuple $(\mathcal{X}, E^{\text{cut}})$ where $\mathcal{X}$ is a partition of the vertex set of $G$ that refines over time and $E^{\text{cut}}$ is a monotonically increasing set of intercluster edges with $E^{\text{cut}} \subseteq E$ such that: 
\begin{enumerate}
    \item at any stage, for every \emph{cluster} $X \in \mathcal{X}$, we have that the current graph $G[X]$ is a $(\phi/c_0)$-expander for some fixed $c_0 = \tilde{O}(1)$, and
    \item at any stage, for every edge $e$ in the current graph $G$, we have that if its endpoints are not in the same cluster $X \in \mathcal{X}$, then the edge is intercluster and therefore in $E^{\text{cut}}$, and at any time $\uu(E^{\text{cut}}) \leq c_1 \cdot \phi \cdot U^{\text{total}}$ where $U^{\text{total}}$ is the total capacity of all edges present in $G$ at any point in time and $c_1 = \tilde{O}(1)$. 
\end{enumerate}
The algorithm takes total time $\tilde{O}(m / \phi^3)$ and succeeds w.h.p.
\end{theorem}

To prove \Cref{thm:weightedExpDecomp}, we give a reduction to the uncapacitated setting and then use the following result. We point out that the theorem below generalizes the theorem in \cite{sulser2024} as it also allows for vertex splits and self-loop insertions. This generalization can be obtained straightforwardly by combining the framework from \cite{sulser2024} with standard techniques from \cite{hua2023maintaining} to deal with vertex splits and self-loop insertions.

\begin{theorem}[Expander Decomposition \cite{sulser2024}]\label{thm:uncapExpDecomp}
Given a parameter $0 <\phi \leq 1$ and an \textbf{un-capacitated} $m$-edge \mbox{(multi-)graph} $G = (V,E)$ undergoing a sequence of $\tilde{O}(m)$ updates consisting of edge deletions, vertex splits and self-loop insertions.

There is a randomized algorithm that explicitly maintains tuple $(\mathcal{X}, E^{\text{cut}})$ where $\mathcal{X}$ is a partition of the vertex set of $G$ that refines over time and $E^{\text{cut}}$ is a monotonically increasing set of intercluster edges with $E^{\text{cut}} \subseteq E$ such that:
\begin{enumerate}
    \item at any stage, for every \emph{cluster} $X \in \mathcal{X}$, we have that the current graph $G[X]$ is a $(\phi/c_0)$-expander for $c_0 = \tilde{O}(1)$, and
    \item at any stage, for every edge $e$ in the current graph $G$, we have that if its endpoints are not in the same cluster $X \in \mathcal{X}$, then the edge is intercluster and therefore in $E^{\text{cut}}$, and at any time $|E^{\text{cut}}| \leq c_1 \cdot \phi m$ for $c_1 = \tilde{O}(1)$.
\end{enumerate}
The algorithm takes total time $\tilde{O}(m/ \phi^2)$ and succeeds  w.h.p.
\end{theorem}

\paragraph{The Algorithm.} For the proof of \Cref{thm:weightedExpDecomp}, we first assume that the total capacity of all edges inserted since the start of the algorithm is at most equal to the total capacity $U^{\text{total}}$ of the initial graph. This is w.l.o.g. as otherwise the algorithm can be restarted with edges in the set $E^{\text{cut}}$ removed from the graph and added to the new set of intercluster edges produced.\footnote{Because capacities are not polynomially-bounded, the number of restarts could be large, however, using the techniques introduced below, an edge can effectively be ignored if its capacity is below $\phi \cdot U^{\text{total}}/m$ and thus any edge is only considered by the algorithm during $O(\log m)$ restarts.}

Then, consider the dynamic graph $G'$ obtained from the graph $G$ by deleting/not inserting all edges with capacity less than $\phi \cdot \frac{U^{\text{total}}}{m}$. Throughout, let $G''$ be the uncapacitated dynamic graph obtained from graph $G'$ by replacing each edge $e$ of capacity $\uu(e)$ by $\lceil \frac{m \cdot \uu(e)}{U^{\text{total}}\phi} \rceil$ multi-edges between the same endpoints and by additionally having $\lceil \frac{m \cdot \vol_G(v)}{U^{\text{total}} \phi} \rceil$ self-loops at each vertex $v \in V$.

Finally, maintain the tuple $(\mathcal{X}, E'')$ by running the algorithm from \Cref{thm:uncapExpDecomp} on graph $G''$. Maintain the output tuple $(\mathcal{X}, E^{\text{cut}})$ to have the same partition and let $E^{\text{cut}}$ be the union of all edges that appear at any time in $G \setminus G''$ and all edges in $E$ such that a corresponding multi-edge is in $E''.$

\paragraph{Analysis.} We prove the two main properties claimed in \Cref{thm:weightedExpDecomp} and then analyze the remaining properties claimed.

\begin{claim}\label{clm:sizeOfGPrimePrime}
The total number of edges to ever appear in $G''$ is at most $\tilde{O}(m/\phi)$. Thus, the total capacity of all edges in $G$ that become intercluster for $\mathcal{X}$ is at most $\tilde{O}(\phi \cdot U^{\text{total}})$.

\end{claim}
\begin{proof}
The total capacity of all edges that ever appear in $G$ is by assumption at most $2 \cdot U^{\text{total}}$. Since we replace each edge of capacity $\uu(e)$ by $\lceil \frac{m \cdot \uu(e)}{U^{\text{total}}\phi} \rceil$ multi-edges, we can thus upper bound the number of such multi-edges by $\frac{m \cdot U^{\text{total}}}{U^{\text{total}}\phi} + \tilde{O}(m) = \tilde{O}(m/\phi)$ since we can charge each edge $e$ its capacity $\uu(e)$ and where the second term $\tilde{O}(m)$ stems from the fact that we are rounding up $\tilde{O}(m)$ terms.

Let us next bound the number of self-loops added to $G''$. We have that the total volume at all vertices is at most $4 \cdot U^{\text{total}}$ at any time by assumption, and we have that there are at most $\tilde{O}(m)$ vertices. Thus, there are again at most $\tilde{O}( \frac{m \cdot U^{\text{total}}}{U^{\text{total}} \cdot \phi}) + \tilde{O}(m) = \tilde{O}(m/\phi)$ self-loops added this way, as desired.

Finally, it suffices to observe that at most a $\tilde{O}(\phi)$-fraction of the edges in $G''$ ever become intercluster for the partition $\mathcal{X}$ by \Cref{thm:uncapExpDecomp}. But for each edge $e$ in the graph $G'$, we add $\lceil \frac{m \cdot \uu(e)}{U^{\text{total}}\phi} \rceil$ corresponding multi-edges to $G''$. Thus, the total capacity of all edges in $G'$ that becomes intercluster for $\mathcal{X}$ is at most $\tilde{O}(\phi \cdot U^{\text{total}})$. Further, the capacity of all edges in $G$ that do not appear in $G'$ is at most $\tilde{O}(m) \cdot \phi \cdot \frac{U^{\text{total}}}{m} = \tilde{O}(\phi \cdot U^{\text{total}})$ by our construction of $G'$.
\end{proof}

\begin{claim}
The partition $\mathcal{X}$ is such that at any time, for any $X \in \mathcal{X}$, we have that $G[X]$ is a $\tilde{\Omega}(\phi)$-expander.
\end{claim}
\begin{proof}
Consider at any time, any cluster $X \in \mathcal{X}$. Let $S \subseteq X$ such that $\vol_{G''[X]}(S) \leq \vol_{G''[X]}(X)/2$. Then, we have from \Cref{thm:uncapExpDecomp}, that $|E_{G''[X]}(S, X \setminus S)| = \tilde{\Omega}(\phi) \cdot \vol_{G''[X]}(S)$. 

Since we have a one-to-one correspondence between non-self-loop multi-edges $e'$ of multiplicity $a$ in $G''$ and edges $e$ in $G'$ such that $\lceil \frac{m \cdot \uu(e)}{U^{\text{total}}\phi} \rceil = a$ and since all edges in $G'$ have capacity at least $\phi \cdot \frac{U^{\text{total}}}{m}$, we have that $\lceil \frac{m \cdot \uu(e)}{U^{\text{total}}\phi} \rceil = a \leq 2\frac{m \cdot \uu(e)}{U^{\text{total}}\phi}$. We further have that 
\begin{align*}
|E_{G''[X]}(S, X \setminus S)| &= \sum_{e \in E_{G'[X]}(S, X \setminus S)} \left\lceil \frac{m \cdot \uu(e)}{U^{\text{total}}\phi} \right\rceil\\
&\leq 2 \uu(E_{G'[X]}(S, X \setminus S)) \cdot \frac{m}{U^{\text{total}}\phi}
\\
&\leq 2 \uu(E_{G[X]}(S, X \setminus S)) \cdot \frac{m}{U^{\text{total}}\phi}
\end{align*}
where we use that $G' \subseteq G$ in the last inequality. Thus, we obtain
\[
\uu(E_{G[X]}(S, X \setminus S)) \geq |E_{G''[X]}(S, X \setminus S)| \cdot \frac{U^{\text{total}}\phi}{2m} = \tilde{\Omega}(\phi) \cdot \vol_{G''[X]}(S) \cdot \frac{U^{\text{total}}\phi}{2m} = \tilde{\Omega}(\phi) \cdot \vol_G(S)
\]
where we use in the last inequality that since each vertex $v \in V$ has at least degree $\lceil \frac{m \cdot \vol(v)}{U^{\text{total}} \phi} \rceil$ in $G''$ and since we add self-loops, we also have that $\vol_{G''[X]}(S) \geq \frac{m \cdot \vol_G(S)}{U^{\text{total}} \phi}$.
\end{proof}

Given the two claims above, it suffices for a proof of \Cref{thm:weightedExpDecomp} to observe that $\mathcal{X}$ is refining by \Cref{thm:uncapExpDecomp}, that $E^{\text{cut}}$ is monotonically increasing by adding only edges that become intercluster for $\mathcal{X}$ and that the time to maintain $\mathcal{X}$ on $G''$ is at most $\tilde{O}(m/\phi^3)$ by the size upper bound from \Cref{clm:sizeOfGPrimePrime} on $G''$ and again by \Cref{thm:uncapExpDecomp}, and that all additional operations take time at most $\tilde{O}(m/\phi^3)$. 

\subsection{Decremental Tree Cut Sparsifiers}
\label{subsec:decrTreeSpars}

In this section, we prove the following result that was previously obtained in \cite{goranci2021expander} for uncapacitated graphs. Our proof follows a similar high-level strategy, however, we require more refined building blocks and arguments to obtain our result.

\begin{theorem}\label{thm:decrTreeCutSparsifier}
Given an $m$-edge graph $G=(V,E,\uu)$ undergoing up to $\tilde{O}(m)$ edge deletions, vertex splits and self-loop insertions where $\uu \in [1, U = m^{O(1)}]^E$. 

Then, there is a randomized algorithm that maintains a tree cut sparsifier $T = (V', E', \uu')$ of $G$ of quality $\gamma_{quality} = 2^{O(\sqrt{\log m} \log\log m)}$ such that $T$ is a graph consisting of at most $\tilde{O}(m)$ vertices and undergoing at most $\tilde{O}(m)$ edge weight decreases and edge un-contractions where the latter is an update that splits a vertex $v$ into vertices $v'$ and $v''$ and inserts an edge $(v', v'')$. The algorithm takes total time $m \cdot 2^{O(\sqrt{\log m})}$. The algorithm succeeds w.h.p. 

Furthermore, the hop diameter of $T$ is at most $O(\log m)$ throughout. 
\end{theorem}

To obtain the above result, we maintain a decremental boundary-linked expander hierarchy as defined in \cite{goranci2021expander}.

\begin{definition}[Dynamic Boundary-Linked Expander Decomposition]\label{def:boundaryLinkedExp}
  Given a dynamic graph $G$ and parameters $\phi \in
  (0,1], \beta > 0, s \geq 1$, we say that a partition $\mathcal{X}$ of the vertex set of $G$ is an $(\beta,
  \phi, s)$ boundary-linked expander decomposition
  of $G$ if 
  \begin{enumerate}
        \item at any stage, for every edge $e$ in the current graph $G$, we have that if its endpoints are not in the same cluster $X \in \mathcal{X}$, then the edge is intercluster and therefore in $E^{\text{cut}}$, and at any time $\uu(E^{\text{cut}}) \leq \beta \cdot \phi \cdot U^{\text{total}}$ where $U^{\text{total}}$ is the total capacity of all edges present in $G$ at any point in time.
      \item at any time, for any $X \in \mathcal{X}$, we have that the graph $G^{1/(s\beta\phi)}_{\mathcal{X}}[X]$ is $\phi$-expander where we have a one-to-one correspondence between edges $e = (u,v) \in E^{cut}$ and self-loops at $u$ and $v$ of capacity $\frac{1}{s\beta\phi} \uu(e)$. Here, $G^{1/(s\beta\phi)}_{\mathcal{X}}$ is the graph $G$ plus self-loops of total capacity $\frac{1}{s\beta\phi} \cdot \uu(E_G(v, V) \cap E^{cut})$ at each vertex $v \in V$.
      \label{prop:boundaryLinkedNumSelfLoops}
  \end{enumerate}
\end{definition}

We next define the crucial concept of expander hierarchies.

\begin{definition}[Dynamic Expander Hierarchy]\label{def:exphierarchy}
Given a dynamic graph $G$ and parameters $\phi \in
  (0,1], \beta > 0$, $s \geq 1$, we define an $(\beta, \phi, s)$-expander hierarchy recursively to consist of levels $0 \leq i \leq k$ where for each level $i$, we maintain a dynamic graph $G_i$ and an $(\beta,
  \phi, s)$ boundary-linked expander decomposition $\mathcal{X}_i$ of $G_i$. We let $G_0 = G$, and for $i \geq 0$, we define $G_{i+1}$ to be the dynamic graph obtained from $G_i$ after contracting all vertices in the same partition set in $\mathcal{X}_i$ into a single super node and removing all self-loops. We let $k$ be the first index such that $G_k$ consists of only a single vertex.
\end{definition}
\begin{remark}
We point out that the partitions $\mathcal{X}_0, \mathcal{X}_1, \ldots, \mathcal{X}_k$ can be extended to partitions of $V$ and it is straightforward to see that the extension of $\mathcal{X}_{i}$ refines the extension of $\mathcal{X}_{i+1}$ and that $G_{i+1}$ can be obtained from contracting $\mathcal{X}_i$ in $G_i$ or from contracting the extension of $\mathcal{X}_i$ in $G$. We use these partitions and their extensions interchangeably when the context is clear. Further, again when the context is clear, we refer to the sets $X \in \mathcal{X}_{i}$ as the vertices of $G_{i+1}$. 
\end{remark}

In our algorithm, for $\phi =  1/2^{\sqrt{\log m}}$, we maintain a $(2c_1, \phi/c_0, 2)$ expander hierarchy for our decremental input graph $G$ as described in \Cref{def:exphierarchy} (the values $c_0, c_1$ are defined in \Cref{thm:weightedExpDecomp}). To maintain the boundary-linked expander decomposition $\mathcal{X}_i$ for each graph $G_i$, we simply run the algorithm from \Cref{thm:weightedExpDecomp} on the graph $\tilde{G}_i = (G_i)^{1/(4c_1 \cdot \beta\phi)}_{\mathcal{X}_{i}}$ where $\beta = 2 c_i$ with parameter $\phi$.\footnote{Note that technically, \Cref{thm:weightedExpDecomp} requires capacities to be at least $1$ while some of the self-loops might be smaller. However, since correctness is not affected by scaling all capacities and all capacities in $\tilde{G}_i$ are polynomially lower bounded in $m$, we can simply scale up all capacities by a large polynomial factor to increase them to be at least of size $1$.} We denote by $E^{cut}_{\mathcal{X}_i}$ the set of cut edges maintained by the algorithm in \Cref{thm:weightedExpDecomp} that is run on $\tilde{G}_i$.

To obtain a tree cut sparsifier $T$ from our dynamic expander hierarchy, we finally appeal to the following theorem. We note that the theorem below from \cite{goranci2021expander} was proven only in the uncapacitated setting, however, their proof extends seamlessly.

\begin{theorem}[see Theorem 5.2 in \cite{goranci2021expander}]\label{thm:GRSTcorrectness}
Given a (dynamic) $(\beta, \phi, s)$-expander hierarchy $\mathcal{H} = \{ (G_0, \mathcal{X}_0), (G_1, \mathcal{X}_k), \ldots, (G_k, \mathcal{X}_k)\}$, and let $\mathcal{X}_{-1}$ denote the partition of the vertex set of $G$ into singleton sets. Let $T_{\mathcal{H}}$ denote the tree that has a node for each set $X$ in any of the partitions $\mathcal{X}_{i}$ and if $i < k$ then the node in $T_{\mathcal{H}}$ associated with $X$ is a child of the node $Y \in \mathcal{X}_{i+1}$ where $X \subseteq Y$ where the capacity of the edge $(X,Y)$ in $T_{\mathcal{H}}$ is $\vol_{G_{i}}(X)$. 

Then, $T_{\mathcal{H}}$ is a tree cut sparsifier of $G$ with quality $O((s\beta)^{O(k)}/\phi)$.
\end{theorem}

From the definitions, maintenance of tree $\mathcal{T}_{\mathcal{H}}$ is straightforward. We note that to reduce the number of updates to the tree cut sparsifier $T$ that we output, we let $T$ be a version of $\mathcal{T}_{\mathcal{H}}$ where all edge capacities are rounded up to the nearest power of two, and enforce that all edge capacities in $T$ are monotonically decreasing by using the smallest capacity value of an edge in $\mathcal{T}_{\mathcal{H}}$ that has been observed so far. Proving that $T$ is still a correct tree cut sparsifier (only worse in quality by a constant factor) is trivial since by the decremental nature of $G$ any fixed cut has monotonically decreasing capacity.

\begin{proof}[Proof of \Cref{thm:decrTreeCutSparsifier}.] We prove by induction on $i$ that
\begin{enumerate}
    \item \label{prop:validExp} it is valid to invoke \Cref{thm:weightedExpDecomp} on the graph $\tilde{G}_i$, i.e. $\tilde{G}_i$ is undergoing only edge deletions, vertex splits and self-loop insertions, and
    \item \label{prop:totalCapExp} the total capacity on all edges in $E_{\mathcal{X}_i}^{cut}$ is at most $(2c_1\phi)^{i+1} U_G^{total}$, and
    \item \label{prop:updateNumExp} the number of updates to $\tilde{G}_i$ is $\tilde{O}(m)$.
\end{enumerate}

\underline{Property \ref{prop:validExp}:} Since for every $i$, $\tilde{G}_i$ is obtained from $G_i$ by self-loop insertions/ deletions, we can conclude that Property \ref{prop:validExp} holds, if it holds for each graph $G_i$. For $i = 0$, it is vacuously true since $G_0 = G$ which is a decremental graph by assumption. For $i > 0$, we use that $\mathcal{X}_{i-1}$ is a refining partition which implies that $G_i$ which is obtained from contracting partition sets of $\mathcal{X}_{i-1}$ in $G$ and removing self-loops, only undergoes the deletions that $G$ undergoes if the corresponding edge in $G_i$, and vertex splits for when $\mathcal{X}_{i-1}$ refines, possibly preceded by insertions of the removed self-loop at the vertex that is split in the current update. 

\underline{Property \ref{prop:totalCapExp}:} By \Cref{def:boundaryLinkedExp}, we have that the total volume of all self-loops added to $\tilde{G}_i$ that are not in $G_i$ already is at most $\frac{1}{4 c_1 \phi} \cdot 2\uu(E^{cut})$ (since each edge $e \in E^{cut}$ adds a self-loop of capacity $\frac{1}{4 c_1 \phi} \uu(e)$ to both $u$ and $v$). Thus, $U^{\text{total}}_{\tilde{G}_i} \leq U^{\text{total}}_{G_i} + \frac{1}{2 c_1 \phi} \cdot \uu(E^{cut})$ (see \Cref{prop:boundaryLinkedNumSelfLoops} in \Cref{def:boundaryLinkedExp}).

On the other hand, since \Cref{thm:weightedExpDecomp} maintains $\mathcal{X}_i$ to be an expander decomposition of $\tilde{G}_i$ with parameter $\phi$, we have that the capacity of all cut edges $E^{cut}_{\mathcal{X}_i}$ is bounded by $c_1 \phi \cdot U^{\text{total}}_{\tilde{G}_i}$.

Combining these inequalities, we obtain
\begin{align*}
    U^{\text{total}}_{\tilde{G}_i} 
    &\leq U^{\text{total}}_{G_i} + \frac{1}{2 c_1 \phi} \cdot c_1 \phi \cdot U^{\text{total}}_{\tilde{G}_i} = U^{\text{total}}_{G_i} + \frac{1}{2} \cdot U^{\text{total}}_{\tilde{G}_i}.
\end{align*}
Subtracting $\frac{1}{2} \cdot U^{\text{total}}_{\tilde{G}_i}$ from both sides on the inequality thus yields $U^{\text{total}}_{\tilde{G}_i} \leq 2 \cdot U^{\text{total}}_{G_i}$. 

Finally, for $i = 0$, we have that $G_0 = G$ has total capacity $U_G^{total}$  by definition, and thus the capacity of all cut edges $E_{\mathcal{X}_0}^{cut}$ is at most $c_1\phi \cdot 2 U_G^{total}$. For $i > 0$, we have that $G_i$ can only obtain edges in $E^{cut}_{\mathcal{X}_{i-1}}$ as can be seen from \Cref{def:exphierarchy}. Thus, we have $U_{G_i}^{total} \leq \uu(E^{cut}_{\mathcal{X}_{i-1}}) \leq (2c_1\phi)^{i} U_G^{total}$ where we used the induction hypothesis for the last inequality. This yields by \Cref{thm:weightedExpDecomp} and our bound $U^{\text{total}}_{\tilde{G}_i} \leq 2 \cdot U^{\text{total}}_{G_i}$ that $\uu(E_{\mathcal{X}_i}^{cut}) \leq 2c_1\phi \cdot U^{\text{total}}_{G_i} \leq (2c_1\phi)^{i+1} \cdot U^{\text{total}}_{G}$, as desired.

\underline{Property \ref{prop:updateNumExp}:} For $G_0 = G$, we have at most $\tilde{O}(m)$ updates. For $i > 0$, we have that $G_i$ undergoes at most $\tilde{O}(m)$ updates since $\mathcal{X}_{i-1}$ is refining and thus, for $G^{final}$ being the final graph $G$, causes at most $|V(G^{final})|-1$ vertex splits and $|V(G^{final})|-1$ self-loop insertions (to compensate for earlier removals of self-loops that go between two vertices in the graph $G_i$ after the vertex splits) and additionally, undergoes the sequence of updates that $G$ is undergoing if the corresponding edges are present in $G_i$. Thus, $G_i$ undergoes $\tilde{O}(m)$ updates. 

Finally, we have that $\tilde{G}_i$ undergoes $2$ self-loop insertions whenever an edge is added to the set $E^{cut}_{\mathcal{X}_i}$. But since $E^{cut}_{\mathcal{X}_i}$ is monotonically increasing (see \Cref{thm:weightedExpDecomp}), this can cause at most twice as many updates as there are edges in $G_i$. Thus, $\tilde{G}_i$ undergoes at most $\tilde{O}(m)$ updates.

\underline{Putting it All Together:} From Property \ref{prop:totalCapExp}, we can conclude that the number of levels of the hierarchy is $O(\log_{1/(2c_1\phi)}(U^{\text{total}})) = O(\sqrt{\log m} \log\log m)$ by choice of $\phi$ and the fact that capacities are polynomially-bounded. 

Correctness of our algorithm thus follows immediately from \Cref{thm:weightedExpDecomp}.

Combining the bound on the number of levels of the expander hierarchy with the runtime bounds obtained by \Cref{thm:weightedExpDecomp} and the bound on the number of updates to each graph $\tilde{G}_i$ by Property \ref{prop:updateNumExp}, we obtain that the expander hierarchy can be maintained in time $\tilde{O}(m/\phi^3)$. From \Cref{thm:GRSTcorrectness}, it can be observed that maintenance of tree $\mathcal{T}_{\mathcal{H}}$ and also of our modified tree $T$ is straightforward and can be done in time $\tilde{O}(m/\phi^3)$ as only $\tilde{O}(1)$ operations suffice to update both trees after any update to the dynamic expander hierarchy. This yields the total runtime of our algorithm.

Finally, to bound the hop diameter of $T$ follows immediately from the fact that $T_{\mathcal{H}}$ and thus $T$ is a tree with $k$ levels where $k = O(\log m)$.
\end{proof}

\subsection{Fully Dynamic Tree Cut Sparsifiers}
\label{subsec:fullydyntoDecrTreeSpars}

Finally, we present an algorithm to maintain a tree cut sparsifier as described in \Cref{thm:mainTreeSparsifier} by giving a reduction to the decremental setting.

\mainTreeCutSparisifer*
\DirectedLayerGraphDef*
\DirectedLayerGraphRand*



\paragraph{Core Graphs.} Before we describe our reduction, let us introduce the concept of core graphs which have been crucial in the design of recent dynamic graph algorithms. 

\begin{definition}[Core graph]
  \label{def:coregraph}
 Given a graph $G=(V, E, \uu)$, a rooted forest $F$ (i.e. each component of $F$ has a dedicated root vertex) with $V(F) \supseteq V$. We define the \emph{core graph} $\mathcal{C}(G, F)$ to be the graph obtained from graph $G$ by contracting the vertices of every connected component in $F$ into a super-vertex that is then identified with the root vertex of the corresponding tree in $F$, i.e. the vertex set of $\mathcal{C}(G, F)$ is the set of roots of $F$. We let the capacities of edges in $\mathcal{C}(G, F)$ be equal to their capacities in $G$. 
\end{definition}

In our algorithm, we use induced core graphs. For the definition, we also need to define the notion of a branch-free set. 

\begin{definition}[Branch-Free Set]
Given a tree $T = (V,E,\uu)$, we say that $B \subseteq V$ is a branch-free set for $T$ if we have that $\mathcal{P}_{T, B}$, the collection of all paths $T[a,b]$ for $a,b \in B$ that contain no internal vertex in $B$, consists of pairwise edge-disjoint paths.
\end{definition}

\begin{definition}[Induced Core Graph]\label{def:inducedCoreGraph}
Given a graph $G=(V, E, \uu)$, a tree $T$ with $V \subseteq V(T)$, and a set of roots $B \subseteq V$ such that $B$ is a branch-free set for $T$. We let $F(T, B)$ denote the rooted forest obtained by removing from the tree $T$ the lexicographically-first edge $e_P$ of minimum capacity from each path $P \in \mathcal{P}_{T, B}$. Note, that this yields a forest $F(T,B)$ where each connected component contains exactly one node in $B$. We let the corresponding vertex in $B$ be the root of its component to make $F(T,B)$ a rooted forest. We define the \emph{induced core graph} $\mathcal{C}(G, T, B)$ to be the core graph $\mathcal{C}(G, F(T,B))$.
\end{definition}

Finally, we state the following algorithmic result that extends any set $R$ to a branch-free set $B$ that is not much larger. To unambiguously define the result, here, we require the notion of a monotonically increasing vertex set for a graph undergoing vertex splits.

\begin{definition}[Monotonically Increasing Set in Graph Undergoing Vertex Splits]
Given a graph $G=(V,E, \uu)$ undergoing a sequence of vertex splits. Whenever a vertex $v$ is split into vertices $v'$ and $v''$, we say that $v'$ and $v''$ descend from $v$ and we further extend this notion to be transitively closed, i.e. if $v'$ is further split into $v'''$ and $v''''$, then $v'''$ and $v''''$ also descend from $v$, and so on. 
Then, we say that a set $X \subseteq V$ is monotonically increasing if for any two time steps $t' < t$, every $v$ in $X$ at time $t'$ has a descendent in $X$ at time $t$.
\end{definition}

The following standard result is then obtained via link-cut trees \cite{ST83} (See e.g. \cite{chen2022maximum}).

\begin{theorem}\label{thm:maintainBranchFreeSet}
Given an $m$-vertex tree graph $T = (V,E,\uu)$ undergoing $\tilde{O}(m)$ edge un-contractions, i.e. updates that split a vertex $v$ into vertices $v'$ and $v''$ and add an edge $(v', v'')$, and a monotonically increasing set $R \subseteq V$. Then, there is a deterministic algorithm that maintains a monotonically increasing set $B$ such that at any time, $R \subseteq B$, $B$ is branch-free for the current tree $T$, and $|B| \leq 2|R|$. The algorithm outputs $B$ explicitly after every update to $T$ or $R$, and runs in total time $\tilde{O}(m)$. 
\end{theorem}

\begin{figure}[!ht]
    \centering
    \includegraphics[width=14cm]{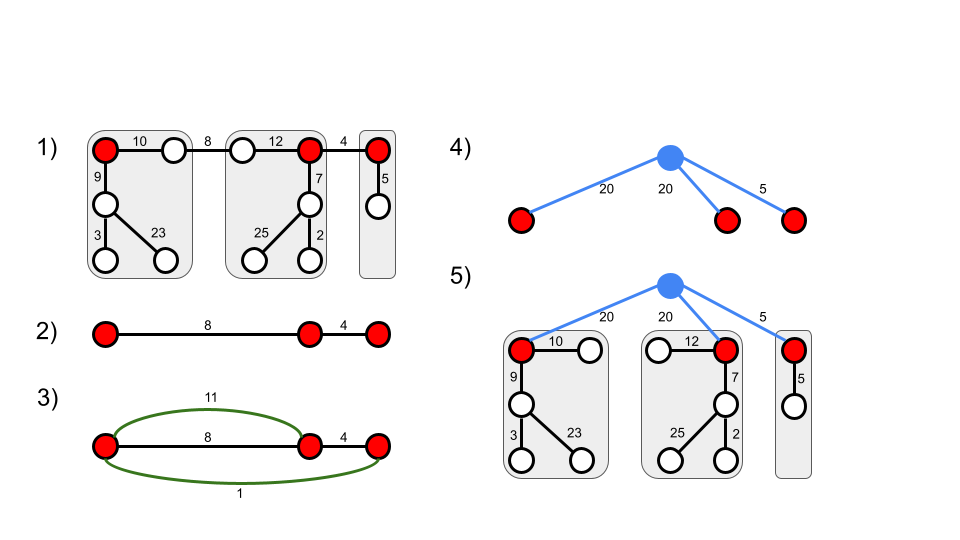}
    \caption{1) shows a tree cut sparsifier $T_{i-1}$ (for a graph $G_{i-1}$). Red vertices are the vertices in $B_i$. The grey components show the connected components of $F(T_{i-1}, B_i)$, edges crossing such components are of minimum capacity on a path in $\mathcal{P}_{T_{i-1}, B_i}$. \\
    2) shows the induced core graph $\mathcal{C}(T_{i-1}, T_{i-1}, B_i)$. \\
    3) shows the induced core graph $\hat{G}_i = \mathcal{C}(T_{i-1} \cup (I_{i-1} \setminus I_i), T_{i-1}, B_i)$, i.e., the previous graph with all edges that are in $G_i$ but not in $G_{i-1}$ (in green). \\
    4) shows a tree cut sparsifier $\hat{T}_i$ of the graph $\hat{G}_i$. \\
    5) shows the final tree cut sparsifier $ T_i$ of $G_i$ which is formed by the union of $F(T_{i-1}, B_i)$ and the tree cut sparsifier $\hat{T}_i$ of the induced core graph $\hat{G}_i$.}
    \label{fig:graphReduction}
\end{figure}

\paragraph{A Hierarchy of Tree Cut Sparsifiers.} We are now ready to describe our reduction (see also \Cref{fig:graphReduction}). Let $\tilde{m} = \tilde{O}(m)$ be a strict upper bound on the number of updates to $G$. Our algorithm maintains levels $0, 1, \ldots, L_{max} = \lceil \log^{1/4}(\tilde{m}) \rceil$. We use a simple batching approach over the update sequence where we associate with each level $i \in [0, L_{max}]$, at current time $t$, an associated time $t_i = \lfloor t / \tilde{m}^{(L_{max} - i)/L_{max}} \rfloor \cdot \tilde{m}^{(L_{max} -i)/ L_{max}}$ at which level $i$ was last re-built.\footnote{We assume here that $\tilde{m}^{(L_{max} -i)/L_{max}}$ is integer which is w.l.o.g.}

We further maintain with each level $i \in [0,L_{max}]$, a batch $I_i$ consisting of all edges in the current graph $G$ that were inserted after time $t_i$ (note in particular that edges added and deleted after time $t_i$ are not in $I_i$). We define $G_i = G \setminus I_i$ for all $0 \leq i \leq L_{max}$. We note in particular that $I_{L_{max}} = \emptyset$ since $t_{L_{max}} = t$ and therefore $G_{L_{max}} = G$. 

For each level $i \in [0,L_{max}]$, our goal is to maintain a tree cut sparsifier $T_i$ of the current graph $G_i$, thus in particular, $T_{L_{max}}$ is a tree cut sparsifier of the current graph $G$. For $i = 0$, we let $T_0$ be the tree cut sparsifier obtained by running the data structure from \Cref{thm:decrTreeCutSparsifier} on the graph $G_0 = G \setminus I_0$, that is, the initial graph where only decremental updates are applied. For $i > 0$, we let $B_i$ be the monotonically increasing set obtained by running the algorithm in \Cref{thm:maintainBranchFreeSet} on the tree $T_{i-1}$ for vertices $V(I_{i-1} \setminus I_{i})$ since time $t_{i}$. Let $\hat{T}_i$ be the tree obtained from running the data structure in \Cref{thm:decrTreeCutSparsifier} on the graph $\hat{G}_i = \mathcal{C}(T_{i-1} \cup (I_{i-1} \setminus I_{i}), T_{i-1}, B_i)$ (as defined in \Cref{def:inducedCoreGraph}) since time $t_i$. Then, we maintain $T_i = 2 \cdot (F(T_{i-1}, B_i) \cup \hat{T}_i)$. Note here in particular that we are not adding the pre-images of edges in $\hat{T}_i$ to $T_i$ but instead the real edges in $\hat{T}_i$ which are supported on $B_i$ only.

As previously mentioned, we output the tree $T_{L_{max}}$ as our tree cut sparsifier $T$ of $G$.

\paragraph{Analysis.} We first establish correctness of the algorithm.

\begin{claim}[Correctness]
The graph $T$ is a tree cut sparsifier of $G$ of quality $2^{O(\log^{3/4}(m) \log\log(m))}$ at all times.
\end{claim}
\begin{proof}
We prove by induction on $i$ that $T_i$ is a tree cut sparsifier of $G_i$ of quality $q_i = (2\gamma_{quality})^{i+1}$. For $i = 0$, we have, by definition of $\tilde{m}$, that all inserted edges since the start of the algorithm are in $I_0$. Thus, the data structure from  \Cref{thm:decrTreeCutSparsifier} maintains $T_0$ correctly to be a tree cut sparsifier of $G_0$ of quality $\gamma_{quality}$. 

For $i > 0$, we have by the induction hypothesis that $T_{i-1}$ is a tree cut sparsifier of $G_{i-1}$ of quality $q_{i-1}$. It is straightforward from the definition of $G_i$ to see that $G_i = G \setminus I_i = (G \setminus I_{i-1}) \cup (I_{i-1} \setminus I_i) = G_{i-1} \cup (I_{i-1} \setminus I_i)$ since $I_{i-1} \supseteq I_i$.

Thus, it is straightforward to verify that by the induction hypothesis, we have that $T_{i-1} \cup (I_{i-1} \setminus I_i)$ is a cut sparsifier of $G_i = G_{i-1} \cup (I_{i-1} \setminus I_i)$ of quality $q_{i-1}$. Finally, consider the graph $T_i$ as maintained by the hierarchy. To see that $T_i = 2\cdot (F(T_{i-1}, B_i) \cup \hat{T}_i)$ is a tree, we use the standard fact that the union of a tree in a graph contracted along a forest and the forest itself yields a tree spanning the original graph. It remains to verify the quality of $T_i$ w.r.t. $G_i$.

\underline{$\text{min-cut}_{G_i}(A, B) \leq \text{min-cut}_{T_i}(A, B)$:} Let us consider any disjoint sets $A, B \subseteq V(G_i) = V$. By sub-modularity of graph cuts in $G_i$, it suffices to focus on the special case that the $AB$-min-cut in $T_i$ consists only of a single edge $e$. 
To show this by induction on the number of cut tree edges, we first extend the $AB$-min-cut to a realization $(A', V_{T_i} \setminus A')$ in $T_i$. Then, we assume that the claimed inequality holds for cuts involving at most $k$ tree edges. Consider a cut involving $k + 1$ tree edges. Remove one of the $k + 1$ cut edges such that the remaining $k$ edges are in the same tree component. Then, let $A'_1$ be the cut induced by the remaining $k$ edges, and $A'_2$ be the cut induced by the removed edge such that $A' \subseteq A'_1$ and $A' \subseteq A'_2$. Then, we have 
\begin{align*}
    \text{min-cut}_{G_i}(A', V \setminus A') &= \text{min-cut}_{G_i}(A'_1 \cap A'_2, V \setminus (A'_1 \cap A'_2)) \\
    &\leq \text{min-cut}_{G_i}(A'_1, V \setminus A'_1 ) +\text{min-cut}_{G_i}( A'_2, V \setminus A'_2) \\
    &\leq \text{min-cut}_{T_i}(A'_1 \cap A'_2, V \setminus (A'_1 \cap A'_2))
\end{align*}
where the first inequality is by sub-modularity of cuts, and the second follows from the induction hypothesis. 
We then distinguish by cases:
\begin{itemize}
 \item \underline{if $e \in F(T_{i-1}, B_i)$:} We now give a formal proof of this case and discuss an example of such a proof in \Cref{fig:analysisGraphReduction}.
 
 Since $F(T_{i-1}, B_i)$ is a forest where each component contains exactly one vertex on $B_i$ (see \Cref{def:inducedCoreGraph}), we have that $F(T_{i-1}, B_i) \setminus \{e\}$ contains a single connected component $A'$ that contains no vertex in $B_i$. Further note that since $T_i$ consists of $F(T_{i-1}, B_i)$ and edges supported only on $B_i$, we have that also $T_i \setminus \{e\}$ contains $A'$ as one of its connected components. Thus, by assumption either $A \subseteq A' \subseteq V(T_i) \setminus B$, or $B \subseteq A' \subseteq V(T_i) \setminus A$. Let us assume for the rest of the proof that $A' \supseteq A$ (the case where $A' \supseteq B$ is analogous). 

 Our key claim is that for any edge $e_1, e_2, \ldots, e_k$ in $E_{T_{i-1}}(A', V(T_{i-1}) \setminus A') \setminus \{e\}$ the path $P_j \in \mathcal{P}_{T_{i-1}, B_i}$ that contains edge $e_j$ also contains edge $e$. Note that this implies that $k \leq 1$, since all paths in $\mathcal{P}_{T_{i-1}, B_i}$ are edge-disjoint. To see the claim, observe that for every path $P_j \in \mathcal{P}_{T_{i-1}, B_i}$ there is only a single edge on $P_j$ removed from $T_{i-1}$ to obtain $F(T_{i-1}, B_i)$ and thus if two edges $e_j$ and $e_{\ell}$ for some $\ell \neq j$ appear on $P_j$ then one of them would still be in the cut $E_{F(T_{i-1}, B_i)}(A', V(T_{i-1}) \setminus A')$ which contradicts that the latter set consists only of the edge $e$. But since all paths $P_1, P_2, \ldots, P_k$ must be distinct, each path $P_j$ enters the component $A'$ via edge $e_j$. But all paths in $\mathcal{P}_{T_{i-1}, B_i}$ start and end in a vertex in $B_i$. Since $A' \cap B_i = \emptyset$, the path $P_j$ must therefore use the edge $e$ to reach a vertex in $B_i$ as it is the only edge to leave $A'$ that is not already on another path.

    We observe that if $E_{T_{i-1}}(A', V(T_{i-1}) \setminus A') \setminus \{e\} = \emptyset$, then trivially \[ \uu_{T_{i-1}}(e) = \uu_{T_{i-1}}(A', V(T_{i-1}) \setminus A') = \uu_{F(T_{i-1}, B_i)}(A', V(T_{i-1}) \setminus A'). \] If it contains an additional edge $e_1$, then we have that the path $P_1$, defined as above, contains the edge $e$. But we have from \Cref{def:inducedCoreGraph} that removing $e_1$ instead of $e$ from $T_{i-1}$ to obtain $F(T_{i-1}, B_i)$ implies $\uu_{T_{i-1}}(e_1) \leq \uu_{T_{i-1}}(e)$. And thus, we have in this case, $\uu_{T_{i-1}}(A', V(T_{i-1}) \setminus A') =  \uu_{T_{i-1}}(e) + \uu_{T_{i-1}}(e_1) \leq 2 \cdot \uu_{T_{i-1}}(e) = 2 \cdot \uu_{F(T_{i-1}, B_i)}(A', V(T_{i-1}) \setminus A')$. We can thus finally use the induction hypothesis on $T_{i-1}$ to obtain that $2\uu_{F(T_{i-1}, B_i)}(A', V(T_{i-1}) \setminus A') \geq \uu_{G_{i-1}}(A', V \setminus A')$ from which we can conclude 
    \begin{align*}
    \uu_{T_{i}}(A', V(T_{i-1}) \setminus A') &= 2 \cdot \uu_{F(T_{i-1}, B_i) \cup \hat{T}_i}(A', V(T_{i}) \setminus A') \\
    &\geq 2 \cdot \uu_{F(T_{i-1}, B_i)}(A', V(T_{i}) \setminus A') \\ 
    &\geq \uu_{G_{i-1}}(A, V\setminus A) \\
    &= \uu_{G_{i}}(A, V \setminus A)
    \end{align*}
    where the last equality follows since no edge from $G_i \setminus G_{i-1}$ is incident to $A$ (since $A \cap B_i \subseteq A' \cap B_i = \emptyset$). 
 
\begin{figure}[!ht]
    \centering
    \includegraphics[width=9cm]{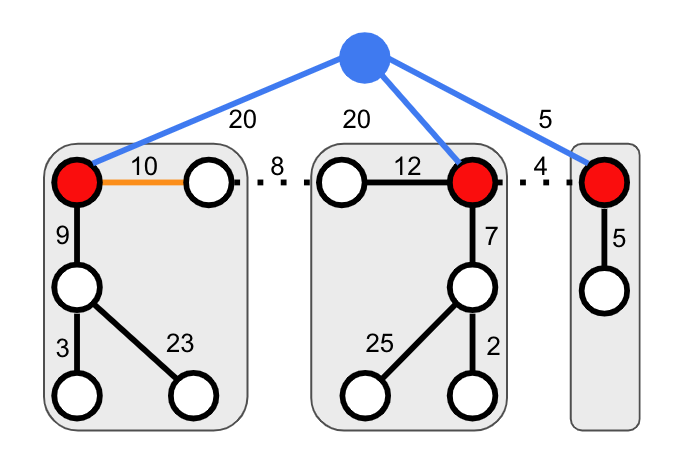}
    \caption{Consider the example from \Cref{fig:graphReduction} where edges in $T_{i-1}$ that are not in $T_i$ are dashed, red vertices are the vertices of $B_i$ and blue edges and vertices and the vertices of $B_i$ form $\hat{T}_i$. \\   Let us argue for the cut induced by the orange edge $e$ in $F(T_{i-1}, B_i)$, where $T_i \setminus \{e\}$ is the $AB$-min-cut of $T_i$. We have that $F(T_{i-1}, B_i) \setminus \{e\}$ contains a connected component $A'$ that has no vertex in $B_i$ since it contains exactly one more connected component than vertices in $B_i$. This component $A'$ is also a connected component of $T_i \setminus \{e\}$ since no edges of $\hat{T}_i$ are incident to $A'$ and $T_i = F(T_{i-1}, B_i) \cup \hat{T}_i$.\\     We prove that either, we are in a case where $T_{i-1}$ has no edge leaving $A'$ other than $e$ in which case we obtain a rather straightforward lower bound on the cut size, or, at most one such edge $e_1$ (in our case the dotted edge incident to the orange edge). But in this case, we have that $e$ and $e_1$ are on a common path $P \in \mathcal{P}_{T_{i-1}, B_i}$ and since from each such path only the edge of smallest capacity is removed, we have $\uu_{T_{i-1}}(e_1) \geq \uu_{T_{i-1}}(e)$. Thus, we can again bound the capacity of the cut $(A', V(T_i) \setminus A')$ by $2\uu(e)$.  }
    \label{fig:analysisGraphReduction}
\end{figure}

    \item \underline{otherwise:} in this case, we have $e \in \hat{T}_i$. Let $A'' \supseteq A$ and $B'' \supseteq B$ be the connected components of $T_i \setminus \{e\}$. Let $A' = A'' \cap V(T_{i-1})$ and $B' = B'' \cap V(T_{i-1})$. We clearly have that $A \subseteq A' \subseteq A''$ and $B \subseteq B' \subseteq B''$ because $V \subseteq V(T_{i-1})$ by induction on $T_{i-1}$. Observe further that $A, B$ partition $V$; $A', B'$ partition $V(T_{i-1})$; and $A'', B''$ partition $V(T_{i})$.  

    Next, let $\hat{A}, \hat{B}$ be the connected components of $\hat{T}_i \setminus \{e\}$ such that $\hat{A} \subseteq A'', \hat{B} \subseteq B''$. Then, we have by \Cref{thm:decrTreeCutSparsifier}, that $\uu_{\hat{T}_i}(\hat{A}, \hat{B}) \geq \uu_{\hat{G}_i}(\hat{A} \cap B_i, \hat{B} \cap B_i)$ where we use that $B_i = V(\hat{G_i})$. By construction, we have that $\uu_{\hat{G}_i}(\hat{A} \cap B_i, \hat{B} \cap B_i)) = \uu_{\mathcal{C}(T_{i-1} \cup (I_{i-1} \setminus I_i), T_{i-1}, B_i)}(\hat{A} \cap B_i, \hat{B} \cap B_i) = \uu_{T_{i-1} \cup (I_{i-1} \setminus I_i)}(A', B') = \uu_{T_{i-1}}(A', B') + \uu_{I_{i-1} \setminus I_i}(A', B')$. By induction, we have that $\uu_{T_{i-1}}(A', B') \geq \uu_{G_{i-1}}(A, B)$ and since $I_{i-1} \setminus I_i \subseteq G_i$, we have that $\uu_{I_{i-1} \setminus I_i}(A', B') = \uu_{I_{i-1} \setminus I_i}(A, B)$. 

    It remains to use that $T_i \supseteq 2 \cdot \hat{T}_i$ and to combine inequalities which yields $\uu_{T_i}(\hat{A}, \hat{B}) \geq 2\uu_{\hat{T}_i}(\hat{A}, \hat{B}) \geq 2\uu_{G_{i-1}}(A, B) + 2\uu_{I_{i-1} \setminus I_i}(A, B)$. 
\end{itemize}
\underline{$\text{min-cut}_{T_i}(A, B) \leq q_i \cdot \text{min-cut}_{G_i}(A, B)$:} For this claim, note that it suffices to prove for all sets $A \subseteq V$ that $\text{min-cut}_{T_i}(A, B) \leq q_i \cdot \uu_{G_i}(A,B)$ for $B = V \setminus A$. 

Let us fix such a cut $(A, B = V \setminus A)$ in $G_i$. We have for $T_i = 2 \cdot (F(T_{i-1}, B_i) \cup \hat{T}_i)$ that since $F(T_{i-1}, B_i) \subseteq T_{i-1}$ and $G_{i-1} \subseteq G_i$, we have that $\text{min-cut}_{F(T_{i-1}, B_i)}(A, B) \leq q_{i-1} \cdot \uu_{G_i}(A,B)$. 

It thus only remains to obtain an upper bound on $\text{min-cut}_{\hat{T}_i}(A, B)$. Since all edges in $\hat{T}_i$ are either incident to vertices in $B_i$ or to newly created vertices (not in $V$), we have that $\text{min-cut}_{\hat{T}_i}(A, B) = \text{min-cut}_{\hat{T}_i}(A \cap B_i, B \cap B_i)$. We have from \Cref{thm:decrTreeCutSparsifier} that $\text{min-cut}_{\hat{T}_i}(A \cap B_i, B \cap B_i) \leq \gamma_{quality} \cdot \uu_{\hat{G}_i}(A \cap B_i, B \cap B_i)$. But for every edge $e$ in $\hat{G}_i$, we either have that $e \in I_{i-1} \setminus I_i$ and thus the edge is also present in $G_i$ with the same quality. Or, there is a path $P_e \in \mathcal{P}_{T_{i-1}, B_i}$ between the two endpoints of $e$ in $T_{i-1}$ where each edge on the path has capacity at least $\uu_{\hat{G}_i}(e)$. Since all of these paths $P_e$ are edge-disjoint, we have that for every edge $e \in E_{\hat{G}_i}(A \cap B_i, B \cap B_i)$, we have at least one edge in $E_{T_{i-1}}(A \cap B_i, B \cap B_i)$ on $P_e$ that has no less capacity than $\uu_{\hat{G_i}}(e)$. Thus, $\text{min-cut}_{\hat{G}_i}(A \cap B_i, B \cap B_i) \leq \uu_{G_i}(A, B) + \uu_{T_{i-1}}(A, B)$. Combining these inequalities, we obtain that 
\begin{align*}
    \text{min-cut}_{\hat{T}_i}(A, B) &\leq \gamma_{quality} (\uu_{G_i}(A, B) + \uu_{T_{i-1}}(A, B)) \\
    &\leq \gamma_{quality} (\uu_{G_i}(A, B) + q_{i-1} \cdot \uu_{G_{i-1}}(A, B)) \\
    &\leq \gamma_{quality} \cdot (q_{i-1} + 1) \cdot \uu_{G_i}(A, B)
\end{align*}
where we used in the second inequality the induction hypothesis on $T_{i-1}$ and in the final inequality that $G_{i-1} \subseteq G_i$.

This yields that $\text{min-cut}_{T_i}(A, B) \leq  (\gamma_{quality} + 1) \cdot (q_{i-1} + 1) \cdot \uu_{G_i}(A, B)$ where we have that $(\gamma_{quality} + 1) \cdot (q_{i-1} + 1) < 2\gamma_{quality} \cdot q_{i-1} = q_i$, as required.
\end{proof}

It now only remains to bound the runtime.

\begin{claim}[Runtime]
The algorithm has amortized update time $2^{O(\log^{3/4}(m)\log\log(m))}$.
\end{claim}
\begin{proof}
The set $I_{i-1}$ contains at most $t - t_{i-1}$ edges at any time which can be bounded by $\tilde{m}_{i - 1}$ where $\tilde{m}_j := \tilde{m}^{(L_{\max} - j)/L_{max}}$ for $j = 0, \ldots L_{\max}$ by definition of $t_{i-1}$. 

But this implies that for level $i$, the set $V(I_{i-1}\setminus I_i)$ is of size at most $2 \tilde{m}_{i - 1}$, and thus the set $B_i$ is of size at most $4 \tilde{m}_{i - 1}$ (see  \Cref{thm:maintainBranchFreeSet}). Since each graph $\hat{G}_i$ consists only of a forest supported on the vertices in $B_i$, and the images of edges in $I_{i-1}\setminus I_i$ under the contractions (see \Cref{def:inducedCoreGraph}), we can bound the number of edges in $\hat{G}_i$ at any time by $4 \tilde{m}_{i - 1} - 1 + \tilde{m}_{i - 1} \leq 5 \tilde{m}_{i - 1}$. 

Thus, the runtime of the decremental tree cut sparsifier run on the graph $\hat{G}_i$ has total update time $2^{O(\sqrt{\log \tilde{m}_{i - 1}})} \cdot \tilde{m}_{i - 1}$ in-between rebuilds by \Cref{thm:decrTreeCutSparsifier}. Since $\widehat{G}_i$ gets re-built after $\tilde{m}_i$ updates and there are $\tilde{m}$ updates in total, the total time spend by all decremental tree cut sparsifier data structures for $\widehat{G}_i$ is at most $\tilde{m}_{i-1} \cdot 2^{O(\sqrt{\log \tilde{m}_{i - 1}})}  \cdot \frac{\tilde{m}}{\tilde{m}_i} = 2^{O(\sqrt{\log \tilde{m}_{i - 1}})} \cdot \tilde{m}^{1 + 1/L_{max}} = 2^{O(\log^{3/4}\tilde{m})} \cdot \tilde{m}$ since $L_{\max} = \ceil{\log^{3/4} m}$. The time to implement the remaining operations of our algorithm is asymptotically subsumed by this bound.
\end{proof}

\paragraph{Extending \Cref{thm:mainTreeSparsifier} to Maintain Cut Edges.} Finally, we will prove \Cref{rem:directed_layer_graph} and \Cref{rem:directed_layer_graph_det}.

To this end, we maintain the layer graph $L$ with levels $i \in [0, k]$ where we choose $k = L_{max} + 1$. We define $\hat{G}_0 = G$ and $\hat{T}_0 = T_0$, and recall the definition of $\hat{G}_i$ for $i > 0$ to be $\hat{G}_i = \mathcal{C}(T_{i-1} \cup (I_{i-1}\setminus I_i), T_{i-1}, B_i)$. From the runtime analysis, we have that every graph $\hat{G}_i$ undergoes at most $m \cdot 2^{\log^{3/4}(m) \log\log(m)}$ updates over the entire course of the algorithm.

Now, we maintain for the graph $L$ the vertex set $V_i$ at level $i$ to be in one-to-one correspondence with the edges of $\hat{G}_i$. We then add for every edge $e = (u,v) \in E(\hat{G}_i)$ and edge $e' \in \hat{T}_i[u,v]$, an edge from the vertex $v_{e'} \in V_{i+1}$ (in one-to-one correspondence with $e'$) to the vertex $v_e \in V_i$ (in one-to-one correspondence with $e$) to the graph $L$. 

For the analysis, let us first observe that the edges of $\hat{G}_0$ are exactly the edges in $G$ and thus $V_0$ is in one-to-one correspondence with the edges in $G$ as required. We further use that the hop diameter of every tree $\hat{T}_i$ is at most $O(\log m)$ by \Cref{thm:decrTreeCutSparsifier}. This implies that the out-degree of every vertex in $L$ is at most $O(\log m)$. Finally, we observe that each graph $\hat{G}_i$ undergoes at most $m \cdot 2^{\log^{3/4}(m)\log\log(m)}$ updates which follows trivially from our runtime analysis and the fact that each such graph is maintained explicitly. It remains to observe that once an edge $e \in E(\hat{G}_i)$ is embedded into an edge $e' \in \hat{T}_i$, it remains embedded until the end of the algorithm or until $e$ is deleted. And since each edge $e'$ that $e$ embeds into can be detected in constant time:
\begin{itemize}
    \item if $e$ is newly inserted, then it suffices to walk towards the root of $\hat{T}_i$ (which we can root arbitrarily for this purpose such that vertices in $\hat{G}_i$) from both endpoints of $e$ to detect all edges that $e$ currently embeds into. By walking in parallel and aborting once the two explorations meet, this operation can be implemented in constant time per detected edge, or
    \item if $\hat{T}_i$ is undergoing an un-contraction (see \Cref{thm:decrTreeCutSparsifier}), then for the newly created edge $e'$, it suffices to copy the set of edges embedded into $e''$, where $e''$ is the edge that is incident to $e'$ and closer to the root of $\hat{T}_i$.
\end{itemize}
This yields both runtime and recourse bounds for the maintenance of graph $L$, as required.

\paragraph{Bounding the hop-diameter. } Since the decremental trees have depth $O(\log n)$, composing $L_{max} + 1 = O(\log^{1/4} m)$ such trees as described above yields a tree of depth $\tilde{O}(1)$. This proves \Cref{rem:DiamTreeSparsifier}. 

\subsection{A Deterministic Algorithm to Maintain Fully Dynamic Tree Cut Sparsifiers}
\label{subsec:detTreeCutSparsifiers}

We finally describe how to derandomize our result. We first note that all algorithmic reductions presented in this paper are already deterministic. Thus, the only algorithm that uses randomization in our data structure above is the algorithm from \Cref{thm:uncapExpDecomp}. We note that a deterministic version of  \Cref{thm:uncapExpDecomp} was already given in  \cite{hua2023maintaining} with only subpolynomially worse runtime and approximation guarantees. Here, however, we describe how to derandomize \Cref{thm:uncapExpDecomp} more directly to obtain sligthly better subpolynomial factors. We note that both \cite{hua2023maintaining, sulser2024} work even in the directed setting while we describe a derandomization for the undirected setting only.

The algorithm from \Cref{thm:uncapExpDecomp} given in \cite{sulser2024} in turn is also deterministic except for $\tilde{O}(1/\phi)$ invocations of the static cut-matching game algorithm from \cite{khandekar2009graph} (the algorithm from \cite{sulser2024} in fact uses a generalization of  \cite{khandekar2009graph} to directed graphs, however, since we only work with undirected graphs, using \cite{khandekar2009graph} in their algorithm is sufficent for our purposes). This algorithm obtains approximation guarantees of $\tilde{O}(1)$ and runtime $\tilde{O}(m/\phi)$ on a graph with $m$ edges. Recently, this algorithm was derandomized in \cite{chuzhoy2020deterministic}, however, the authors obtained a slightly weaker result: their approximation guarantee is $e^{O(\log^{1/3}(m)\log\log m)}$ while their runtime is $m \cdot e^{O(\log^{2/3}(m)\log\log m)} / \phi^2$ (this is implied in particular by Theorem 5.3 in \cite{chuzhoy2020deterministicArxiv}, the second ArXiv version of \cite{chuzhoy2020deterministic}). 

Using this algorithm internally in the framework from \cite{sulser2024}, we obtain a deterministic algorithm implementing \Cref{thm:uncapExpDecomp} with $c_0 = e^{O(\log^{1/3}(m)\log\log m)}$, $c_1 = e^{O(\log^{1/3}(m)\log\log m)}$ and total update time $m \cdot e^{O(\log^{2/3}(m)\log\log m)} / \phi^3$. We thus obtain a deterministic version of \Cref{thm:weightedExpDecomp} with the same values for $c_0$ and $c_1$ and total update time $m \cdot e^{O(\log^{2/3}(m)\log\log m)} / \phi^4$.

By carefully re-parameterizing the algorithm in \Cref{subsec:decrTreeSpars} to use $\phi = 1/2^{\log^{2/3}(m)}$, we obtain that the number of levels of the expander hierarchy can be bounded by $O(\log^{1/3}(m)\log\log(m))$. We thus recover a quality of the final tree cut sparsifier of $\gamma_{quality} = 2^{O(\log^{2/3}(m)\log\log(m))}$ and a runtime of $m \cdot e^{O(\log^{2/3}(m)\log\log m)}$. The bound on the hop diameter of the tree cut sparsifier $T$ is again $O(\log m)$.

Finally, we use our deterministic algorithm to maintain a tree cut sparsifier of a decremental graph in lieu of the randomized algorithm, and re-parametrizing the algorithm in \Cref{subsec:fullydyntoDecrTreeSpars} to only use $L_{max} = \lceil \log^{1/6}(\tilde{m}) \rceil$ levels. This yields \Cref{thm:det_mainTreeSparsifier}.

\section{Dynamic Min-Ratio Cut}
\label{sec:min_ratio}

In this section, we build a data structure that allows us to toggle along approximate min-ratio cuts in a fully dynamic graph $G = (V, E, \uu, \gg)$. Unlike the previous section, every vertex $v$ now has an extra associated value: a gradient $\gg(v) \in \R$. When we compute a tree-cut sparsifier on such a graph $G = (V, E, \uu, \gg)$ it simply ignores these vertex gradients.  We first define the min-ratio cut problem. 

\begin{restatable}[Min-Ratio Cut]{definition}{minRatioCut}
    \label{def:minRatioCut}
    For a graph $G = (V, E, \uu, \gg)$, we refer to $\min_{\bDelta \in \R^V} \frac{\langle \gg , \bDelta  \rangle}{\norm{\UU \BB \bDelta}_1}$ as the min-ratio cut problem. 
\end{restatable}

Then, we define the main data structure this section is concerned with. Together with the interior point method in \Cref{sec:IPM}, this data structure is what enables us to solve decremental threshold min-cost flow.

\begin{restatable}[Min-Ratio Cut Data Structure]{definition}{MinRatioCutDSdef}    
    \label{def:min_ratio_cut_data_structure}
    For a dynamic graph $G = (V, E, \uu, \gg)$ where $\uu \in \R^E$ and $\gg \in \R^V$, $\gg \perp \vecone$, such that $\uu(e) \in [1, U]$ and $\gg(v) \in [-U, U]$ for $\log U = \tilde{O}(1)$, an initial potential vector $\yy \in R^V$, and a detection threshold parameter $\epsilon$, a $\alpha$-approximate min-ratio cut data structure $\mathcal{D}$ supports the following operations.
    \begin{itemize}
        \item $\textsc{InsertEdge}(e), \textsc{DeleteEdge}(e)$: Inserts/deletes edge $e$ to/from $G$ with capacity $\uu(e)$.  
        \item $\textsc{UpdateGradient}(u, v, \delta)$: Updates $\gg(u) = \gg(u) + \delta$ and $\gg(v) = \gg(v) - \delta$. 
        \item $\textsc{InsertVertex}(v)$: Inserts isolated vertex $v$ to $G$.
        \item $\textsc{Potential}(v)$: Returns $\yy(v)$.
    \end{itemize}
    After every update the data structure $\mathcal{D}$ the data structure returns a tuple $(g, u)$ where $g \in \R_{\leq 0}$ and $u \in \R_{\geq 0}$ such that for some implicit cut $\vecone_{C}$ for $C \subseteq V$ we have $\l \gg, \vecone_{C} \r = g$ and $\norm{\UU \BB \vecone_{C}}_1 \leq u$, and 
    \begin{align*}
        \frac{g}{u} \leq \frac{1}{\alpha} \min_{\bDelta \in \R^V} \frac{\langle \gg, \bDelta \rangle}{\norm{\UU \BB \bDelta}_1}.
    \end{align*}
    The data structure additionally allows updates of the following type based on the most recently returned tuple $(g, u)$.
    \begin{itemize}
        \item $\textsc{ToggleCut}(\eta)$: Given a parameter $\eta \leq 1/u$, the data structure 
        implicitly updates $\yy$ with $\yy^{(new)}$ such that $\BB\yy^{(new)} = \BB\yy + \eta \BB \vecone_{C}$. 

        Then, the data structure returns some edge set $E'$ such that every edge $e = (u,v)$ for which $\uu(e)(\yy(u) - \yy(v))$ has changed by at least $\epsilon$ since it was inserted/last returned in $E'$.
    \end{itemize}
\end{restatable}

We then state two separate theorems showing that there is both a randomized and a deterministic algorithm implementing a min-ratio cut data structure.

\begin{restatable}{theorem}{MinRatioCutDSrand}
    \label{thm:min_ratio_cut_data_structure_rand}
    There is a randomized min-ratio cut data structure (\Cref{def:min_ratio_cut_data_structure}) given a graph $G = (V, E, \uu, \gg)$ and $\epsilon$ for $\alpha = 2^{O(\log^{3/4} m \log \log m)}$ such that every update/query is processed in amortized time $2^{O(\log^{3/4} m \log \log m)} \log U$. Furthermore, the total number of edges returned by the algorithm after $t$ calls to $\textsc{ToggleCut}$ is at most $2^{O(\log^{1/4} \log \log m)} \cdot t/\epsilon$. The algorithm works against an adaptive adversary and succeeds with high probability. 
\end{restatable}

\begin{restatable}{theorem}{MinRatioCutDSdet}
    \label{thm:min_ratio_cut_data_structure_det}
    There is a deterministic min-ratio cut data structure (\Cref{def:min_ratio_cut_data_structure}) given a graph $G = (V, E, \uu, \gg)$ and $\epsilon$ for $\alpha = 2^{O(\log^{5/6} m \log \log m)}$ such that every update/query is processed in amortized time $2^{O(\log^{5/6} m \log \log m)} \log U$. Furthermore, the total number of edges returned by the algorithm after $t$ calls to $\textsc{ToggleCut}()$ is at most $2^{O(\log^{1/6} \log \log m)} \cdot t/\epsilon$. 
\end{restatable}

\subsection{Toggling Min-Ratio Cuts on a Tree Cut Sparsifier}

Before we prove \Cref{thm:min_ratio_cut_data_structure_rand} and the deterministic version \Cref{thm:min_ratio_cut_data_structure_det}, we show that a cut is a solution to the min-ratio cut problem, which explains the nomenclature. 

\begin{lemma}
    \label{lem:cut_suffices}
    $\min_{C \subset V} \frac{\langle \gg, \vecone_C \rangle}{\norm{\UU \BB \vecone_C}_1} = \min_{\bDelta \in \R^V} \frac{\langle \gg, \bDelta \rangle}{\norm{\UU \BB \bDelta}_1}$ 
\end{lemma}
\begin{proof}
    First observe that since the right hand side is a minimum over all $\bDelta,$ thus the minimum objective value achieved must be negative. Further, considering $\bDelta = \vecone_C$ for all $C \subseteq V$ gives us that the right hand side is less than or equal to left hand side.
    To establish that left hand side is less than equal to right hand side, we consider a vector $\bDelta^\star$ minimizing the right hand side. Without loss of generality, we assume that the minimum entry of $\bDelta^\star$ is $0$ and the maximum entry is $1$ by shifting and scaling.

    Let $t$ be a random variable uniformly distributed on $[0,1].$ Let $C_t$ denote the set $\{v \in V| \bDelta^{\star}(v) > t\}.$
    Observe that $\expct{t}{\vecone_{C_t}} = \bDelta^{\star}.$ By linearity of expectations, we have $\expct{t}{\langle \gg, \vecone_{C_t} \rangle} = \langle \gg,  \bDelta^{\star} \rangle.$ 

    Moreover, $\expct{t}{\norm{\UU \BB \vecone_{C_t}}_1} = \sum_e \uu_e \expct{t}{\abs{\chi_e^{\top} \vecone_{C_t}|}}.$ Observe that for each edge $e,$ $\chi_e^{\top} \vecone_{C_t}$ has the same sign for all $t.$ Thus, $\expct{t}{\abs{\chi_e^{\top} \vecone_{C_t}|}} = \abs{\expct{t}{\chi_e^{\top} \vecone_{C_t}|}} = \abs{\chi_e^{\top}\bDelta^{\star}},$ and hence $\expct{t}{\norm{\UU \BB \vecone_{C_t}}_1} = \norm{\UU \BB \bDelta^{\star}}_1.$
    Thus, we have,
    \begin{align*}
        \frac{\langle \gg,  \bDelta^{\star} \rangle}{\norm{\UU \BB \bDelta^{\star}}_1} = \frac{\expct{t}{\langle \gg, \vecone_{C_t} \rangle}}{\expct{t}{\norm{\UU \BB \vecone_{C_t}}_1}}.
    \end{align*}
    Hence there exists a $t$ where $\norm{\UU \BB \vecone_{C_t}}_1 \neq 0$ and 
    $\frac{{\langle \gg, \vecone_{C_t} \rangle}}{{\norm{\UU \BB \vecone_{C_t}}_1}} \le \frac{\langle \gg, \bDelta^{\star} \rangle}{\norm{\UU \BB \bDelta^{\star}}_1}$ by the well known fact that $\min_{i \in [n]} \aa(i)/\bb(i) \leq \sum_{i = 1}^n \aa(i)/\sum_{i = 1}^n \bb(i)$ for $\aa \in \R^n$ and $\bb \in \R_{> 0}^n$.  
    This concludes our proof. 
\end{proof}

Next, we show that given a tree-cut sparsifier $T$ of $G$, there exists a tree cut that corresponds to an approximate min-ratio cut.

\begin{lemma}
\label{lem:tree_cut_suffices}
Given a tree cut sparsifier $T = (V(T), E(T), \uu_T)$ of quality $q$ of a graph $G = (V, E, \uu, \gg)$ there exists an edge $e_T \in E(T)$ that induces a cut $(C, V(T) \setminus C)$ such that 
\begin{align*}
    \frac{\langle \gg, \vecone_{C \cap V} \rangle}{\uu_T(e_T)} \leq \frac{1}{q}\min_{\bDelta \in \R^V} \frac{\langle \gg, \bDelta \rangle}{\norm{\UU \BB \bDelta}_1}
\end{align*}
\end{lemma}
\begin{proof}
     By the well known fact that $\min_{i \in [n]} \aa(i)/\bb(i) \leq \sum_{i = 1}^n \aa(i)/\sum_{i = 1}^n \bb(i)$ for $\aa \in \R^n$ and $\bb \in \R_{> 0}^n$ it suffices to consider cuts in a single connected component. We therefore without loss of generality assume that $G$ and thus $T$ are connected. 
    
    By \Cref{lem:cut_suffices} there exists some cut $C'$ such that $\frac{\langle \gg, \vecone_{C'} \rangle}{\norm{\UU \BB \vecone_{C'}}_1} = \min_{\bDelta \in \R^V} \frac{\langle \gg, \bDelta \rangle}{\norm{\UU \BB \bDelta}_1}$. Then, by \Cref{def:tree_cut_sparsifer} the cut $\mincut_{T}(C', V \setminus C') \leq q \cdot \norm{\UU \BB \vecone_{C'}}_1$, and therefore this cut achieves the min-ratio up to a factor of $\frac{1}{q}$. We next show that the quality of this tree cut can be realized by an individual cut edge in $T$. To do so we arbitrarily root the tree at some vertex $r$, where we denote $\yy = \vecone_{C'}$ and assume $\yy(r) = 0$. Notice that such a vertex always exists. 

    For every edge $e = (u, v)$ where $u$ is the parent of $v$, we then set $\aa(e) = \yy(v) - \yy(u)$. Notice that $\aa(e) \in \{-1, 0, 1\}$ and $|\aa(e)| = 1$ if and only if the edge $e$ is in the cut found by $\yy$. Furthermore, we let $\ss_e$ denote the indicator vector of all the vertices in the sub-tree rooted at $v$.   
    We next show that $\hat{\yy} \defeq \sum_{e \in T} \aa(e) \ss_{e} = \yy$. Let $u$ be an arbitrary vertex. Then 
    \begin{align*}
        \hat{\yy}(u) &= \sum_{e \in T} \aa(e)\ss_{e}(u)\\
        &= \sum_{e \in T[u, r]} \aa(e)\ss_{e}(u) \\
        &= \sum_{e \in T[u, r]} \aa(e) = \sum_{(v, w) \in T[u, r]} \yy(v) - \yy(w) = \yy(u)
    \end{align*}
    where the first equality is by definition, the second follows from the fact that sub-trees not containing $u$ do not affect its value in $\hat{\yy}$, the third follows since $\ss_{e}(u) = 1$ for ever edge on the path $T[u, r]$, the forth follows by definition of $\aa(e)$ and the final equality follows from $\yy(r) = 0$ after cancellation. 

    Therefore we have $\l \gg, \yy \r = \sum_{e \in T} \aa(e) \l \gg, \ss_e \r$ and $\norm{\UU \BB \yy}_1 = \sum_{e \in T} |\aa(e)| \norm{\UU \BB \ss_e}_1$, because $\BB\ss_e$ are indicators of single edges. The result then again follows from the well known fact that $\min_{i \in [n]} \aa(i)/\bb(i) \leq \sum_{i = 1}^n \aa(i)/\sum_{i = 1}^n \bb(i)$ for $\aa \in \R^n$ and $\bb \in \R_{> 0}^n$.
\end{proof} 

As a final ingredient, we need a data structure that detects when potential difference may have changed significantly. 

\begin{restatable}[Detection Algorithm]{definition}{DetectionDef}
    \label{def:detection_alg}
    Given a fully dynamic tree cut sparsifier $T$ and a corresponding fully dynamic directed layer graph $H$ as in \Cref{thm:mainTreeSparsifier} and \Cref{rem:directed_layer_graph} respectively, and an edge significance function $s: E \mapsto \R_{> 0}$, a $\gamma$-approximate detection algorithm supports the following operations. 
    \begin{itemize}
        \item $\textsc{AddDelta}(e_T, \delta)$: Given an edge $e_T \in E(T)$ and a value $\delta \in \R_{> 0}$, it adds $\delta$ to the accumulated change of each edge $e$ in $E'_{e_T}$, i.e. to each edge reachable from $v_{e_T}$ in $H$. 
        Then, it reports a set $E'$ of edges such that 
        \begin{itemize}
            \item Every reported edge $e \in E'$ has accumulated a change of at least $s(e)/\gamma$.
            \item Every edge $e$ that has accumulated a change of at least $s(e)$ is in $E'$.
        \end{itemize}
        Then the accumulated change of the edges in $E'$ is re-set to $0$.
        \item $\textsc{Reset}(e)$: Resets the accumulated change of edge $e$ to $0$.
    \end{itemize}
    Furthermore, we let $D$ be the total number of detected edges throughout the course of the algorithm, $C$ be the number of updates to $H$ and $R$ be the total number of calls to $\textsc{Reset}()$.
\end{restatable}

We next state the main theorem of \Cref{sec:detection}, which describes our detection algorithm.

\begin{restatable}{theorem}{DetectionThm}
    \label{thm:detection_alg}
    There exists a $\gamma$-approximate deterministic detection algorithm (\Cref{def:detection_alg}) for $\gamma = d^k$ with total update time $\tilde{O}(d^k (D + R + C))$. Recall that $d$ is a bound on the in-degree of $H$, and $k$ is a bound on the depth of $H$ (See \Cref{def:directed_layer_graph}).
\end{restatable}

The proof of \Cref{thm:detection_alg} is deferred to \Cref{sec:detection}.

\paragraph{Proof of \Cref{thm:min_ratio_cut_data_structure_rand} and \Cref{thm:min_ratio_cut_data_structure_det}.} 

We first prove \Cref{thm:min_ratio_cut_data_structure_rand} using the slightly faster randomized three cut sparsifiers, and then proceed with the analogous proof of \Cref{thm:min_ratio_cut_data_structure_det} using deterministic tree cut sparsifiers. 

Since the tree cut sparsifiers require the capacities to be polynomially bounded, our algorithm internally maintains data structures for $\log(U)$ different levels. We first describe the objects the data structure maintains at level $i = 0, \ldots, \log(U) - 1$. 

\begin{itemize}
    \item \underline{Rounded Graph:} We let $G_i = (V, E, \uu_i)$ be the graph $G$ with altered capacities 
    \begin{align*}
    \uu_i(e) = 
    \begin{cases}
        \ceil{\uu(e)/n^i} & \text{if } \ceil{\uu(e)/n^i} \leq n^{10} \\
        n^{20} & \text{otherwise}
    \end{cases}
    \end{align*} 
    \item \underline{Tree Cut Sparsifer:} We maintain a tree cut sparsifer $T_i$ of $G_i$ with quality $q = 2^{O(\log^{3/4} m \log \log m)}$ and the stated update time in \Cref{thm:mainTreeSparsifier} (See \Cref{thm:det_mainTreeSparsifier} for the deterministic version). 
    \item \underline{Min-Ratio Cut:} We maintain the ratio achieved by every tree cut and keep them in a sorted list according to minimum ratio/quality, where we discard every cut that has capacity larger than $n^{20}$. Notice that this effectively contracts edges of capacity $n^{20}$, and does not affect cuts that do not contain such edges since other edges have capcity at most $n^{10}$, and therefore no cuts only involving such edges can reach capacity $n^{20}$.

    To explicitly maintain the ratio a tree cut achieves we directly have access to the capacities, and therefore only need to worry about maintaining the gradient sums. To do so, we use that the depth of $T_i$ is bounded by $\tilde{O}(1)$. We then maintain the value $\gg(v)$ at every vertex $v$, and maintain the sum of these values on either side of each edge in the tree $T_i$ (vertices that are not in $G_i$ but in $T_i$  contribute $0$). These two values are the sum of the gradients of crossing edges with opposite sign. Whenever a gradient between two vertices gets updated, only two vertices are affected in their $\gg(v)$ value and can be updated explicitly. Since the vertices change by the same amount with opposite sign, only tree cuts that place them in different components are relevant. Therefore, only the values stored on the edges on the tree path between the two endpoints need to be updated. Notice that there are only $\tilde{O}(1)$ such edges by the depth bound on the tree cut sparsifiers (See \Cref{rem:DiamTreeSparsifier}).  The value of all the edges on the path between the endpoints of the edge for which the gradient was updated can be updated accordingly. 
    
    Finally, whenever a tree edge $(u, v)$ in the tree that contains vertices from $G$ on either side is updated, we simulate it as moving the two endpoints separately from the old tree edge to the new tree edge. Then after moving a single endpoint, say $v$, only the edges on the path between the new endpoint and the old endpoint need to be updated with the sum of the values stored for the component containing $u$ after removing edge $(u, v)$. This value is readily available on edge $e$. The update is then analogous to the case where a gradient is updated. 
    
    Finally, new additional vertices and edges to them can be inserted and always store $0$ since they do not contain crossing edges.  
    \item \underline{Detection Algorithm:} Furthermore, every level $i$ initializes a detection algorithm to detect whenever the quantity $\uu(u,v)(\yy(u) - \yy(v))$ has changed by an additive $\epsilon$. It will ignore cancellations in-between calls to $\textsc{ToggleCut}()$ and track the difference $\yy(u) - \yy(v)$ while looking for changes of the size $\epsilon/\uu(e)$ instead, which is equivalent. To this avail, it initializes a detection algorithm $\mathcal{D}_i$ (\Cref{def:detection_alg} and \Cref{thm:detection_alg}) using the directed layer graph associated with $T_i$. Every level sets the significance of edge $e$ to $s(e) = \epsilon/(\uu(e) \log U)$, such that the change summed up over the levels is still bounded by $\epsilon/\uu(e)$.
\end{itemize}

After each update the algorithm goes through the best min-ratio tree cut found at each level and scales the quality of the best cut at level $i$ by dividing it by $n^{i}$. Then, it outputs the gradient sum $g$ of this cut, and the cut estimate $u = n^{i} \cdot \uu_{T_i}(e_{T_i})$ where $e_{T_i}$ is the edge that induces the cut. 

We next show that the best quality such tree cut across all levels is $\alpha$ competitive for $\alpha = 2^{\O(\log^{3/4} m \log \log m)}$. To do so, we fix a cut $C$ optimal for the min-ratio cycle problem which exists by \Cref{lem:cut_suffices}. Consider the edge $e'$ in $E(C, V \setminus C)$ with highest capacity. Let $i$ be the smallest index such that $\ceil{\uu(e)/n^{i}} \leq n^{10}$. If that $i$ is $0$, then the cut $C$ is captured by the tree-cut sparsifier $T_0$ up to a factor of $2 q$, and therefore one if its tree cuts is a $2q$ approximate min-ratio cut by \Cref{lem:tree_cut_suffices}. For higher $i$, the cut is again approximately preserved, because capacities with $\uu(e) \geq n^{i}$ are correct up to a factor $2$, and capacities $\uu(e) \leq n^{i}$ contribute less than  $\uu(e')$ altogether since there are at most $n^2$ such capacities. Therefore, the tree $T_i$ again captures the cut up to a $2 q$ factor (after re-scaling), and therefore one if its re-scaled tree cuts is $2q$ approximately min-ratio by \Cref{lem:tree_cut_suffices}. 

The potential vector $\yy$ after updates can be tracked via a link-cut tree on $T_i$ \cite{ST83} and queries can be supported by querying the potential change on each tree and adding them up.

It remains to show that we can detect changes in potential difference of the right magnitude when they happen without returning too many edges. We can update the detection thresholds of the set $E'_{e_T}$ from the directed layer graph $H$ (\Cref{def:directed_layer_graph}) via the routine $\textsc{AddDelta}(e_T, \eta)$ of the detection algorithm. Then the detection algorithm clearly reports all edges that have changed by the required margin, because the set $E'_{e_T}$ is a super-set of the edges the tree cut actually cuts and it does not factor in cancellations that happen across calls. 

We finally bound the total number of returned edges. Since the potential change of an edge can only be detected whenever it has accumulated $\epsilon \uu(e) /(\gamma \log(U))$ change in potential difference, and the total amount of change per update is $|E'| \cdot \eta \leq |E'| /\uu_T(e_T) \leq |E'|/(\sum_{e \in E'_{e_T}} \uu(e))$ we have that at most $(\gamma \cdot t \log U)/\epsilon$ edges get reported after processing $t$ updates. 

The runtime guarantee follows from \Cref{thm:detection_alg} and \Cref{thm:mainTreeSparsifier}. 

We finally remark that the completely analogous proof using deterministic tree cut sparsifiers (See \Cref{thm:det_mainTreeSparsifier}) yields \Cref{thm:min_ratio_cut_data_structure_det}.

\subsection{Detection of Large Potential Changes Across Edges}
\label{sec:detection}

\paragraph{Overview. } In this section, we describe our algorithm for detecting significant potential difference changes across edges. To do so, we ignore cancellations between updates. Concretely, if for an edge $(u, v)$ the potential $\yy(u)$ first changes by adding $+\delta$, and then the \emph{next} update changes $\yy(v)$ by adding $+\delta$, we treat this as a possible potential difference change of $2\delta$, although the real change is $0$. On the other hand, if the same update adds $\delta$ to both $u$ and $v$, we typically consider this as a change of $0$. Since our tree cut sparsifier sometimes over-reports the edges in a cut, there are exceptions where we also count potential changes to such edges. 

We first recall the definition of a directed layer graph as introduced in \Cref{def:directed_layer_graph} and the two lemmas about its maintenance in the randomized and deterministic setting. 

\DirectedLayerGraphDef*

\DirectedLayerGraphRand*

\DirectedLayerGraphDet*

We then state the definition of the detection algorithm and the corresponding theorem we prove in this section. 

\DetectionDef*

Recall that to maintain stable lengths, we set $s(e) \approx \epsilon/\uu(e)$ when using the detection algorithm as part of our min-ratio cycle data structure. In the rest of this section, we prove the previously stated main theorem about detection algorithms. 

\DetectionThm*

\subsubsection{Reduction to Trees}

We first reduce the detection of changes from a direct layer graph to a collection of $n^{o(1)}$ trees. Recall that directed layer graphs change in terms of edge insertions/deletions, and isolated vertex insertions. The directed layer trees described in this section receive the same update types. We emphasise that edges in $G$ are represented as leaf nodes in directed layer graphs $H$. We will refer to them as $G$-edges in the remainder of this section to disambiguate them from edges in the directed layer graph $H$. 

\begin{definition}[Directed Layer Tree]
    We call a directed layer graph a directed layer tree if it is a sub-graph of a directed layer graph that has in-degree at most one.
\end{definition}

\begin{lemma}[DAG to Tree]
    \label{lem:dag_to_tree}
    Assume a $1$-approximate detection algorithm (See \Cref{def:detection_alg}) for a tree-cut sparsifier $T$ with a directed layer tree $H_T$ with total runtime $\tilde{O}(R + D + C)$, then there is a $\gamma$-approximate detection algorithm for $\gamma = d^k$ with total runtime $\tilde{O}(d^k \cdot (R + D + C))$.
\end{lemma}
\begin{proof}
    We reduce from directed layer graphs to several directed layer trees. We first define a collection $\mathcal{C}$ of trees, and will then run one detection algorithm for each such tree.

    For every $\dd \in [d]^k$, we let $H_{\dd}$ denote the sub-tree of $H$ for which every vertex in $V_i$ only picks its $\dd(i)$-th in edge. The collection $\mathcal{C}$ consists of all such trees $H_{\dd}$. Notice that every path from a leaf to the root in $H$ is contained in one of the trees by selecting the corresponding in-edge at every level in the vector $\dd$. We then run one detection algorithm for each tree, where we use the sensitivity function $s': E \mapsto \R_{>}$ where $s'(e) := s(e)/d^k$. Whenever the $\textsc{Reset}$ operation is called, we simply call it for all of these algorithm and whenever the $\textsc{AddDelta}()$ operation is called, we also forward it to all the data structures. 

    Whenever a $G$-edge $e$ is reported by one of the tree data structures, it is reported, and we call $\textsc{Reset}(e)$ on all tree data structures. 

    We first show that whenever an edge is reported, it accumulated $s(e)/d^k = s(e)/\gamma$ change. This is evidently the case since it accumulated this amount of change in one of the trees, and  the trees are sub-graphs of $H$ and thus strictly under-estimate the amount of change. 
    
    We then show that every $G$-edge $e$ (corresponding to a leaf node in $H$) that accumulated $s(e)$ change is reported. Every time the $G$-edge $e$ is in $E_{e_T}$ for a call of $\textsc{AddDelta}(e_T, \delta)$ there is a path from $v_{e_T}$ to the edge in $H$, and therefore there is such a path is in at least one of the trees $H_{\dd}$, the tree that always picks the in edge of said path. Therefore, at least on of the tree data structures experiences the $\delta$ change. But if there are $d^k$ data strucutres and we distribute $s(e)$ total change among them, at least one of them has to experience $s(e)/d^k$ change by the pigeonhole principle. Therefore the algorithm is correct. 

    The runtime guarantee directly follows since we run $d^k$ copies of the tree data structure. 
\end{proof}

\subsubsection{Detection on Trees}

In this section, we describe the exact detection algorithm on directed layer trees experiencing edge insertions/deletions, and isolated vertex insertions. Via the reduction presented in the previous section, this suffices to arrive at a full detection algorithm. 

In this section, we prove the following theorem. We remark that we strongly believe that all the required operations can be implemented in a standard way using top-trees \cite{05alstrup_top_tree}. Indeed, detection can be solved with a data structure that can maintain a (significance) function $s: V(T) \to \R$ on the vertices of a tree $T$, that supports the following operations in deterministic $\O(1)$ time:
\begin{itemize}
    \item Change the tree $T$ via edge insertions and deletions.
    \item Update $s(v) \gets s(v) + \delta$ for all vertices $s$ in a subtree of $T$.
    \item Return $\max_{s \in V(T)} s(v)$.
    \item Update the value of a vertex $v$: $s(v) \gets C$.
\end{itemize}

However, we can give a more self-contained proof of the following statement, by leveraging that the directed layer graphs has low depth (at most $O(\log^{1/4} m)$).
\begin{theorem}
    \label{thm:tree_detection}
    Given a tree cut sparsifier $T$ with directed layer tree $H$ there is a $1$-approximate detection algorithm with total runtime $\tilde{O}(R + C + D)$.
\end{theorem}

Our data structure is a standard lazy heap construction. We first state a simple heap theorem that we will use in our data structure. 

\begin{theorem}[Heap]
    \label{thm:heap}
    There is a data structure $\textsc{Heap}()$ that stores an initially empty collection $\mathcal{C}$ of comparable elements from some universe and supports the following operations:
    \begin{itemize}
        \item $\textsc{Insert}(e, \val(e))$: Adds $e$ to $\mathcal{C}$.
        \item $\textsc{Delete}(e)$: Removes $e$ from $\mathcal{C}$.
        \item $\textsc{Min}()$: Returns an element $e$  such that $e \leq e'$ for all $e' \in \mathcal{C}$. 
        \item $\textsc{RemoveMin}()$: Removes and returns an element $e$ such that $e \leq e'$ for all $e' \in \mathcal{C}$.
    \end{itemize}
    All the above operations take worst-case update time $O(\log |\mathcal{C}|)$. 
\end{theorem}
\begin{proof}
    Directly follows from a self-balancing binary search tree. 
\end{proof}

\paragraph{Terminology for Recording Changes. } Before we describe our algorithm, we lay out the terminology used in the remainder of this section. 
\begin{itemize}
    \item \underline{Total change:} We say that whenever $\textsc{AddDelta}(e_T, \delta)$ is called, the total change  all vertices in the directed layer tree that are currently children of the node corresponding to $e_T$ received (See \Cref{def:directed_layer_graph}) is increased by $\delta$. This quantity is therefore monotonically increasing. 
    \item \underline{Experienced change:} Since our data structure is lazy, not all the change is pushed down the tree at once. Therefore, we refer say that the change a node has experienced is given by the amount that has been pushed down to said node so far. This quantity is upper bounded by the total change of the node. We refer to the experienced change of node $v$ as $c(v)$. 
    \item \underline{Passed change:} Every edge $e = (u, v)$ records some passed change $p(e) \leq c(u)$. Intuitively this corresponds to the amount of change node $u$ pushed to node $v$, but whenever an edge $(u,v)$ is inserted $p(e)$ is initialized to the total change of $u$ encoding that this change should not be passed to $v$ since it was issued before the edge $e$ existed. We refer to $c(u) - p(u, v)$ as the amount of change accumulated at vertex $u$ with respect to child $v$.
    \item \underline{Stored change:} We say that the change stored for a node $v$ is equal to the amount of change it would experience if all the nodes on the path to its root would push all their experienced change down. Our data structure makes sure that the change stored for a node $v$ is equal to the total change it received. 
    \item \underline{Difference since last Reported:} For leaf nodes, we sometimes additionally refer to the difference of the above quantities measured against the last time the $G$-edge got reported. In particular, we let $r(l)$ store the amount of experienced change of leaf $l$ at the time it was last reported. Our algorithm will ensure that this coincides with the total change at that point in time. 
    \item \underline{Significance Threshold:} Every vertex $v \in V(H)$ stores a sensitivity threshold $t(v)$, which encodes the magnitude of stored change it would like to be notified about
    \item \underline{Significance Function:} For leaf vertices $l$ we denote with $s'(l)$ the value $s'(e)$ where $e$ is the $G$-edge associated with $l$ and $s'(e) = s(e)/d^k$ is the scaled significance function and thus an input to the algorithm.
\end{itemize}

We first define the tree path of each leaf. 

\begin{definition}
    For every node $v$, we let the path $P_v$ denote the maximum length directed path ending at $v$ in $H$.
\end{definition}

We then describe our tree detection algorithm $\textsc{TreeDetection}()$. See \Cref{alg:tree_detection_heap} for pseudocode and \Cref{fig:alg_example} for a small example of two updates. We first describe the routines separately, and start with the internal routine $\textsc{FlushVertex}(v)$ which is used internally by all other routines. 
\begin{itemize}
    \item \underline{$\textsc{FlushVertex}(v)$:} First determines the path $P_v$ from $v$ to its root. It then considers one edge $(u, u')$ at a time, going from the root down to $v$.  It then sets $c(u') \gets c(u') + c(u) - p(u, u')$ and $p(u, u') \gets c(u')$. It then traverses the path a second time, this time from $v$ to the root, and for every edge $(u, u')$ updates the significance threshold $t(u) \gets -c(u) + \min_{w:(u, w) \in E(H)} t(w) + p(u, w)$. This pushes all the updates stored above $v$ to vertex $v$ and updates the thresholds accordingly. Thereafter, $c(v)$ contains the sum of all the updates vertex $v$ has received so far and therefore the stored change is equal to the experienced change for vertex $v$ (and every vertex in $P_v$). 
    \item \underline{$\textsc{InsertEdge}(u, v)$:} Calls $\textsc{FlushVertex}(u)$ such that vertex $u$  contains all the change it has received so far, then inserts the edge $(u, v)$ and sets $p(u, v) = c(u)$ encoding that these $c(u)$ flow were inserted before and therefore not destined to go across edge $(u,v)$. Finally, calls $\textsc{FlushVertex}(v)$ to update all the affected 
    thresholds $t(\cdot)$ in the path from $v$ to the root. 
    \item \underline{$\textsc{DeleteEdge}(u, v)$:} Calls $\textsc{FlushVertex}(v)$ to propagate all experienced change on the path from $v$ to its root to vertex $v$ before removing the edge. Then removes the edge, and calls $\textsc{FlushVertex}(u)$ to update all the affected thresholds $t(\cdot)$ in the path from $u$ to the root.
    \item \underline{$\textsc{AddDelta}(v, \delta)$:} Adds $\delta$ to $c(v)$. Then it updates $t(v) = \min_{(v, w)} p(v,w) + t(w) - c(v)$, and if $t(v) < 0$ it follows the path of edges that minimize $p(u,v) + t(v)$ until it finds a leaf $l$. Then it calls $\textsc{FlushVertex}(l)$, to propagate all additional changes to that leaf, and returns the leaf and re-sets its values $r(l) \gets c(l)$ and $t(l) \gets s'(l)$. This ensures that whenever a leaf is returned, it has $0$ stored change since the last reporting thereafter. Finally, it calls $\textsc{FlushVertex}(l)$ to update significance thresholds. If there is another node $v'$ for which $t(v') < 0$, the algorithm repeats the above with $v \gets v'$ without adding $\delta$. 
\end{itemize}
This concludes the description of our algorithm. We remark that the threshold function $t(\cdot)$ is monotonically decreasing along paths from leaves to roots, and that they encode the magnitude of an update that needs to reach a node to trigger it to propagate to at least one of its children. The threshold for a node is then  given by the minimum threshold among its children after subtracting the amount of flow the vertex has yet to send to said children. In particular, this ensures that for a tree edge $(u,v)$ we have $t(u) \leq t(v) - c(u) + p(u,v)$, where we recall that $c(u) - p(u,v)$ is the change accumulated at $u$ with respect to $v$. 

\begin{algorithm}
\SetKwProg{myalg}{Procedure}{}{}
\myalg{$\textsc{Init}(d)$}{
    $V \gets \emptyset$; 
    $E \gets \emptyset$
}
\myalg{$\textsc{InsertEdge}(u, v)$}{
    \tcp{If $u$ or $v$ not in $V$, add them with $c(u) = 0$ or $c(v) = 0$ respectively. If $v$ leaf either set $t(v) = s'(v)$, or $t(v) = \infty$ if $v$ not associated with $G$-edge.}
    $\textsc{FlushVertex}(u)$ \\
    $E \gets E \cup (u, v)$; 
    $p(u, v) \gets c(u)$; \\
    $\mathcal{D}_u.\textsc{Insert}((u,v), t(v) + p(u, v))$; \\
    $\textsc{FlushVertex}(v)$
}

\myalg{$\textsc{DeleteEdge}(u, v)$}{
    $\textsc{FlushVertex}(v)$ \\
    $E \gets E \setminus (u, v)$; $\mathcal{D}_u.Delete(u, v)$; $t(u) \gets \mathcal{D}_u.\textsc{Min}() - c(u)$ \\
    $\textsc{FlushVertex}(u)$;
}
\myalg{$\textsc{AddDelta}(v, \delta)$}{
    $c(v) \gets c(v) + \delta$; 
    $t(v) \gets  \mathcal{D}_v.\textsc{Min}() - c(v)$;
    $R \gets \emptyset$ \\
    \While{There exists a vertex $u$ with $t(u) \leq 0$}{
        $w \gets u$ \\
        \While{$w$ not a leaf}{
            $(w, \cdot) \gets \mathcal{D}_w.\textsc{Min}()$
        }
        $R \gets R \cup w$ \\
        $\textsc{FlushVertex}(w)$ \\
        $r(w) \gets c(w)$; 
        $t(w) \gets s'(w) - r(w)$ \\
        $\textsc{FlushVertex}(w)$ \\
    }
    \Return{$R$}
}
\myalg{$\textsc{FlushVertex}(v)$}{
    \tcp{Initiates an exact recompute for the path $P_v$}
    For all $j$, let $\mathcal{S}^X_j$ be the subset of $\mathcal{S}_j$ that contains all vertices $Y \in \mathcal{S}_j$ such that there is a path from $Y$ to $X$. Let $E^X$ be the set of edges $(Y, Z)$ such that $Y \in {\mathcal{S}^X_j}$ and $Z \in {\mathcal{S}^X_k}$ for some $j$ and $k$. \\
    \For{$(w,w') \in P_v$ in order of the path direction}{
        $c(w') \gets c(w') + c(w) - p(w, w')$ \\
        $p(w, w') = c(w)$
    }
    \For{$(w, w') \in P_v$ in reverse order of the path direction}{
        $\mathcal{D}_w.\textsc{Delete}((w,w'))$ \\
        $\mathcal{D}_w.\textsc{Insert}((w,w'), t(w') + p(w, w'))$ \\
        $t(w) \gets \mathcal{D}_w.\textsc{Min}() - c(w)$;
    }
}
\caption{$\textsc{TreeDetection}()$}
\label{alg:tree_detection_heap}
\end{algorithm}

\begin{figure}
    \centering
    \includegraphics[width = 16.5cm]{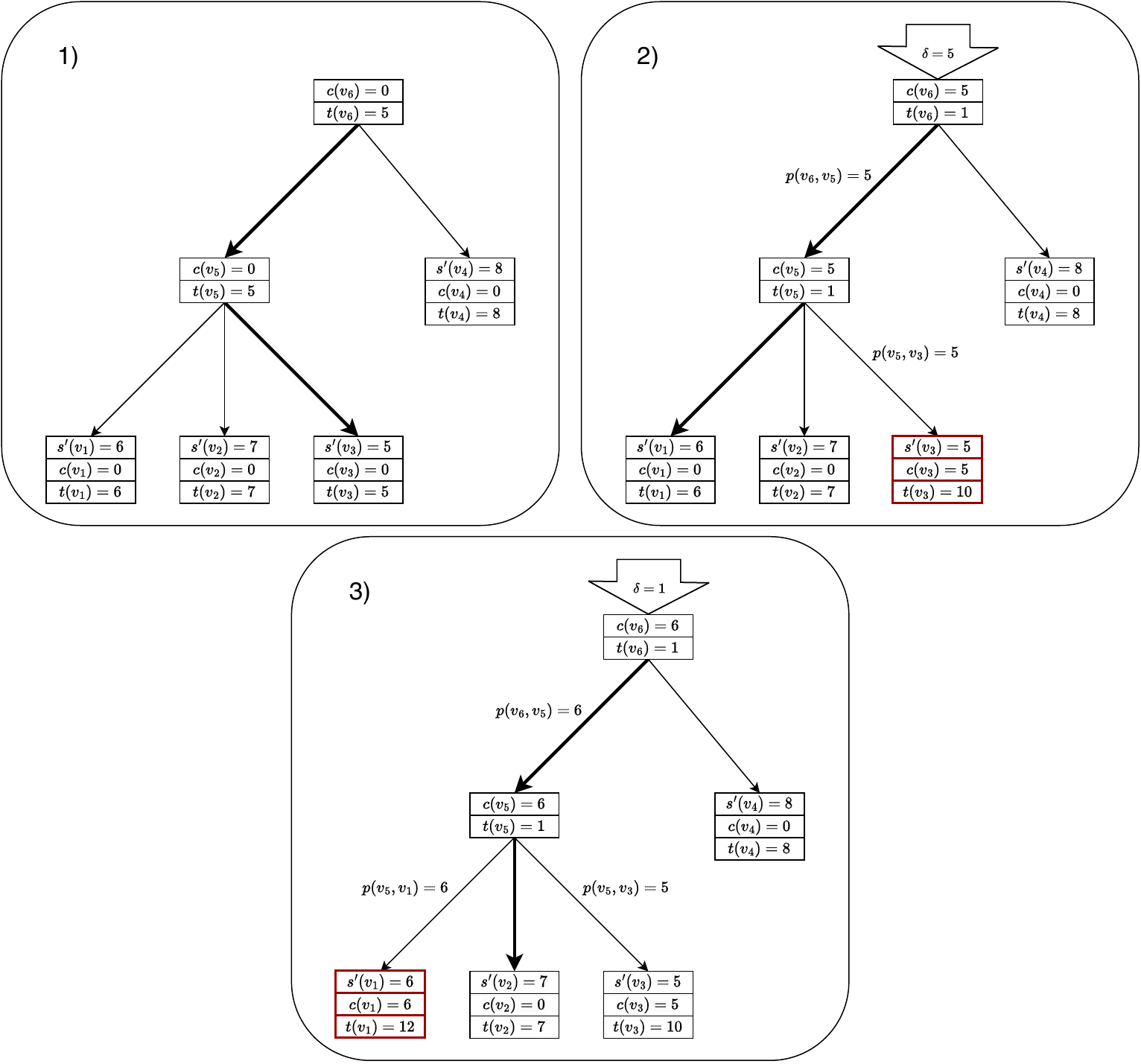}
    \caption{This figure displays a static directed layered tree experiencing two updates $\textsc{AddDelta}(v_6, \delta)$. Edges $(v_i, v_j)$ with $p(v_i, v_j) = 0$ are not labeled. Each update returns one leaf depicted in red, and causes the significance thresholds to be updated. Notice that for leaf vertices $v$ the threshold $t(v)$ is always $s'(v) + r(v) - c(v)$ where $r(v)$ is the value of $c(v)$ when $v$ was last returned. For internal nodes, $t(v) = \min_{(v, u)} t(u) + p(v, u) - c(v)$ at all times. Every internal vertex has a bold outgoing arrow, which points to the child that has the lowest value in its heap and therefore gives rise to its significance threshold. 
    } 
    \label{fig:alg_example}
\end{figure}

Our algorithm then maintains the following invariant. 

\begin{definition}[Algorithm Invariant]
    \label{def:alg_inv}
    The following are invariant in our algorithm. 
    \begin{enumerate}
        \item \underline{Path stores change:}
        For every leaf node $v$ corresponding to an edge $e$ we have that the total change since it last got reset/reported is \begin{align*}
            \left(\sum_{v' \in P_v} c(v')\right) - \left(\sum_{e \in P_v} p(e)\right) - r(v).
        \end{align*}
        \item \underline{Significance thresholds are positive:} The signifcance thresholds $t(\cdot)$ are $t(v) = s'(v) + r(v) - c(v)$  for $v$ being a leaf, and 
        \begin{align*}
            t(v) = - c(v) + \min_{(v, u) \in E(H)} t(u) + p(u).
        \end{align*}
        We have $t(v) > 0$ for all $v$.
    \end{enumerate}
\end{definition}

\begin{lemma}
    \label{lem:inv_maint}
    The tree detection algorithm $\textsc{TreeDetection}()$ (\Cref{alg:tree_detection_heap}) maintains the Algorithm Invariant (\Cref{def:alg_inv}). 
\end{lemma}
\begin{proof}
    Initially, the invariant holds since the graph is empty. Next, we argue that each operation preserves the invariant. 
    \begin{itemize}
        \item \underline{$\textsc{FlushVertex}(v)$:} This operation does not change the sum in invariant $1$ for every vertex $v'$ by the definition of the algorithm, and therefore preserves it. The significance thresholds all become strictly smaller, but no smaller than the topmost one. Therefore they stay positive. After a call to flush vertex, it is the case that for all edges $(u, u')$ on the path to the root $P_v$ we have $p(u, u') = c(u)$, and therefore vertex $v$ has received all the previous updates and stored them in $c(v)$. Therefore, this operation pushes all the accumulated change on vertices on the path from $v$ to its root to $v$, and does not change the magnitude of the stored change. 
        \item \underline{$\textsc{InsertEdge}(u, v)$:} Since the vertex $u$ is flushed first all the accumulated change on the path is pushed to vertex $u$. Then the edge is added with $p(u,v) = c(u)$ thereafter and it correctly stores that it has not received any updates going across it yet. Finally, $\textsc{FlushVertex}(v)$ triggers a re-compute of all significance thresholds $t(\cdot)$ that might have changed. This also preserves invariant $2$, since the significance threshold only decreases to the significance threshold of the root. 
        \item \underline{$\textsc{DeleteEdge}(u, v)$:} Since vertex $v$ is flushed first, all the updates that needed to pass through edge $(u, v)$ get handed to $v$. Thereafter, no more updates to $v$ are stored on the path to its root. Therefore, the edge can safely be deleted and the call to $\textsc{FlushVertex}(u)$ ensures that significance thresholds $t(\cdot)$ are updated accordingly. This again preserves invariant $2$, since the significance threshold only decreases to the significance threshold of the root. 
        \item \underline{$\textsc{AddDelta}(v, \delta)$:} The algorithm first increases $c(v)$ by $\delta$ for vertex $v$, and it then updates $t(v) \gets \min_{(v, u)} p(v, u) + t(u) - c(v)$ for non-leaf vertices and $t(l) \gets c(l) - r(l)$ leaf vertices. This may lead to a violation of item $2$ of the algorithm invariant, i.e. $t(v) < 0$. If so, then the algorithm follows the edge that minimizes $p(u, u') + t(u')$ until it finds a leaf. Then, it flushes this vertex, which ensures that $p(u, u') = c(u)$ for all edges on said path, and therefore the path does not contain any accumulated change. Then it returns the vertex (and the corresponding $G$-edge) and sets $t(l) = 0$ where $l$ is the found leaf. These two points ensure that the leaf has no more accumulated change, and can therefore not be returned anymore. To re-compute all the significance thresholds $t(\cdot)$, the algorithm calls $\textsc{FlushVertex}(l)$ again. If there is another violation where $t(v') \leq 0$ for some vertex, it repeats the above procedure.
    \end{itemize}
\end{proof}

\begin{lemma}
    \label{lem:changed_reported}
    An algorithm that maintains the invariant in \Cref{def:alg_inv} reports every $G$-edge $e$ that has accumulated more than $s'(e)$ change. 
\end{lemma}
\begin{proof}
    We aim to show that the stored change for a leaf node $v$ is strictly bounded by $s'(v)$. This is equivalent to showing 
    \begin{align}
        \label{eq:item1}
        \left(\sum_{v' \in P_v} c(v')\right) - \left(\sum_{e \in P_v} p(e)\right) - r(v) < s'(v)
    \end{align}
    by item 1 of \Cref{def:alg_inv}. 
    
    We denote the sub-path of $P_v$ from $v$ to $u$ as $P_v^u$. We show that \begin{align*}
        t(u) \leq s'(v) - \left(\left(\sum_{v' \in P_v^u} c(v')\right) - \left(\sum_{e \in P_v^u} p(e)\right) - r(v)\right).
    \end{align*}
    For $u = v$, we have $t(v) = s'(v) - r(v)$ by definition of $t(v)$. Therefore the base case of the induction holds. 

    We then assume the claim to be shown for the path $P_v^u$, and show it for $P_v^w$ where $(w, u)$ is in $P_v$. We conclude 
    \begin{align*}
        t(w) &= - c(w) + \min_{(w, w') \in E(H)} t(w') + p(w, w') \\ &\leq -c(w) + t(u) + p(w, u) \\
        &\leq -c(w) + p(w, u) + s'(v) - \left(\left(\sum_{v' \in P_v^u} c(v')\right) - \left(\sum_{e \in P_v^u} p(e)\right) - r(v)\right) \\
        &= s'(v) - \left(\left(\sum_{v' \in P_v^w} c(v')\right) - \left(\sum_{e \in P_v^w} p(e)\right) - r(v)\right).
    \end{align*}
    We conclude that
    \begin{align*}
        0 < s'(v) - \left(\left(\sum_{v' \in P_v} c(v')\right) - \left(\sum_{e \in P_v} p(e)\right) - r(v)\right)
    \end{align*}
    by item 2 of \Cref{def:alg_inv} which establishes the lemma. 
\end{proof}

\begin{lemma}
    \label{lem:reported_changed}
    Whenever a vertex is returned by $\textsc{TreeDetection}()$ after a call to $\textsc{AddDelta}(v, \delta)$, it has accumulated a change of at least $s'(v)$. 
\end{lemma}
\begin{proof}
   By the first item in \Cref{def:alg_inv} the path stores the change. Then, an item is only returned when the change exceeds the $s'(v)$ by the definition of our algorithm. 
\end{proof}

The proofs of \Cref{thm:tree_detection} and \Cref{thm:detection_alg} follow directly. 

\begin{proof}[Proof of \Cref{thm:tree_detection}]
    The correctness follows from \Cref{lem:inv_maint}, \Cref{lem:reported_changed} and \Cref{lem:changed_reported}. The runtime follows from \Cref{lem:inv_maint}, \Cref{thm:heap} and the description of our algorithm.  
\end{proof}

\begin{proof}[Proof of \Cref{thm:detection_alg}]
    Follows directly from \Cref{thm:tree_detection} and \Cref{lem:dag_to_tree}.
\end{proof}
\section{L1-IPM on the Dual}
\label{sec:IPM}
In this section, we develop a dual $\ell_1$-IPM for the min-cost flow problem and prove \Cref{thm:main}.

\mainTheorem*

Without loss of generality, we may assume that the problem is uncapacitated:
\begin{align}
\label{eq:transPrimal}
\min_{\BB^\top \ff = \bd, \ff \ge 0} \l \cc, \ff \r
\end{align}
Specifically, any capacitated min-cost flow problem on a graph of $n$ vertices and $m$ edges can be reduced to the uncapacitated case on $O(m)$ vertices and edges.
Notice that the reduction loses density and is therefore useless when density is key, as in~\cite{van2020bipartite,van2021minimum}.
However, the reduction works for decremental graphs, i.e., decremental min-cost flow can be reduced to decremental transshipment.
We present the reduction in \Cref{sec:reduceToTrans}, and it is summarized as \Cref{lem:reduceToTrans}.

The dual of \eqref{eq:transPrimal} is
\begin{align}
\label{eq:transDual}
\max_{\cc - \BB \yy \ge 0} \l\dd, \yy\r
\end{align}
Let $F$ be the cost-threshold for which we want to decide if $\min_{\BB^\top \ff = \bd, \ff \ge 0} \cc^\top \ff \ge F$ and let $\alpha \defeq 1/(1000\log(mC)).$
We define the potential $\Phi(\yy)$ for any $\yy$ s.t. $\cc - \BB \yy \ge 0$ as follows:
\begin{align}
\label{eq:dualPot}
\Phi(\yy) \defeq 100 m \log\left(F - \l\dd, \yy\r\right) + \sum_{e = (u, v) \in G} \left(\cc(e) - (\BB\yy)(e)\right)^{-\alpha}
\end{align}
where we notice that $(\BB\yy)(e) = \yy(u) - \yy(v)$.

Our potential reduction framework solves \eqref{eq:transDual} by finding a feasible dual solution $\yy \in \R^V$ with small potential.
Next, we show that if the potential is smaller than $-\O(m)$, then $\yy$ is a high-accuracy solution.

\begin{lemma}[Small Potential implies Small Gap]
\label{lem:smallPot}
If $\Phi(\yy) \le -1000m\log(mC)$, $\l\dd, \yy\r \ge F - (mC)^{-10}.$
\end{lemma}
\begin{proof}
By definition \eqref{eq:dualPot}, we have
\begin{align*}
100 m \log\left(F - \l\dd, \yy\r\right) \le -1000m\log(mC)
\end{align*}
This implies that $F - \l\dd, \yy\r \le (mC)^{-10}.$
\end{proof}

To find $\yy$ with small $\Phi(\yy)$, we start with an initial dual and iteratively update the solution $\yy \gets \yy + \bDelta$ by a $\bDelta$ that gives sufficient decrease in the potential, e.g., $\Phi(\yy + \bDelta) \le \Phi(\yy) - m^{-o(1)}.$
Because $\Phi(\cdot)$ is continuous, we can approximate the change in $\Phi(\yy + \bDelta)$ around a small neighborhood of $\yy$ using the 1st and 2nd order derivative information.
\begin{definition}[Slacks, Residual, Gradients and Capacities]
\label{def:potGradLen}
Given a feasible potential $\yy \in \R^V$, we define its \emph{edge slacks} as $\ss(\yy) \defeq \cc - \BB \yy \in \R^E_{\ge 0}$ and its \emph{residual} as $r(\yy) \defeq F - \l\dd, \yy\r$.

Given edge slacks $\ss \in \R^{E}_{\ge 0}$, which might be only an estimation of the true slack $\cc - \BB \yy$, approximate residual $r$, we define \emph{edge capacities} $\uu(\ss) \in \R^E_{\ge 0}$ as
\begin{align*}
\uu(\ss) \defeq \ss^{-1-\alpha} \in \R^E_{\ge 0}
\end{align*} and \emph{vertex gradients} $\gg(\ss, r) \in \R^V$ as 
\begin{align*}
\gg(\ss, r) \defeq \frac{-100m}{r} \dd + \alpha \BB^\top \ss^{-1-\alpha} \in \R^V
\end{align*}
We write $\uu(\yy)$ and $\gg(\yy)$ when they are both defined w.r.t. the exact slacks and the residual.
We sometimes ignore the parameters from $\ss$, $r$, $\uu$, and $\bg$ when it is clear from the context.
\end{definition}

At each $\yy$, we want to find a small $\bDelta$ that approximately minimizes the 2nd order Taylor approximation of $\Phi(\yy + \bDelta)$:
\begin{align*}
\Phi(\yy + \bDelta) \approx \Phi(\yy) + \l\bg(\yy), \bDelta\r + \norm{\UU(\yy) \BB\bDelta}_2^2 \le \Phi(\yy) + \l\bg(\yy), \bDelta\r + \norm{\UU(\yy) \BB\bDelta}_1^2
\end{align*}
Notice that $\gg(\yy) = \g \Phi(\yy)$.
Minimizing the right-hand side corresponds to a cut optimization problem on $G$, namely, min-ratio cuts.
\minRatioCut*

We show that $m^{1+o(1)}$ iterations are enough to find a dual solution of small potential.
Each iteration asks for an approximate min-ratio cut which will eventually be maintained efficiently using dynamic graph algorithms.
Moreover, we show that both edge gradients $\gg(\yy)$ and capacity $\uu(\yy)$ are stable over the course of the algorithm.
This is summarized as follows:
\begin{theorem}[Dual L1 IPM]
\label{thm:dualL1IPM}
Consider a decremental uncapacitated min-cost flow instance \eqref{eq:transPrimal} with an associated dual \eqref{eq:transDual}, a cost threshold $F$, and an approximation parameter $\kappa = m^{o(1)}.$
There is a potential reduction framework for the dual problem that runs in $\O(m \kappa^2)$ iterations in total.

We start with a feasible dual $\yy^{(0)}$ such that $\Phi(\yy^{(0)}) = \O(m)$ and edge slacks $\wt{\ss} \defeq \ss(\yy^{(0)})$ and residual $\wt{r} \defeq r(\yy^{(0)}).$
At each iteration, the following happens:
\begin{enumerate}
\item We receive updates in $U \subseteq E$ to $\wt{\ss}$ so that
\begin{align}
\label{eq:estAcc}
    \wt{\ss} \approx_{1 + 1 / 10\kappa} \ss(\yy^{(t)})
\end{align}
We also update $\wt{\rr}$ so that $\wt{r} \approx_{1 + 1 / 10\kappa} r(\yy^{(t)}).$
\item Compute a $\kappa$-approximate min-ratio cut $\bDelta \in \R^V$, i.e., an $\kappa$-approximate solution to the problem:
\begin{align*}
    \min_{\bDelta \in \R^V}\frac{\l\wt{\gg}, \bDelta\r}{\norm{\wt{\UU}\BB\bDelta}_1}
\end{align*}
where $\wt{\gg} = \gg(\wt{\ss}, \wt{r})$ and $\wt{\uu} = \uu(\wt{\ss}).$
If the ratio is larger than $-\O(1/\kappa)$, we certify that the optimal value to \eqref{eq:transPrimal} and \eqref{eq:transDual} is less than $F.$
\item Scale $\bDelta$ so that $\l\wt{\bg}, \bDelta\r = -1 / (100\kappa^2)$ and update $\yy^{(t+1)} \gets \yy^{(t)} + \bDelta.$
\end{enumerate}
After $\O(m \kappa^2)$ iterations, we have $\l\dd, \yy\r \ge F - (mC)^{-10}.$

Over the course of the algorithm, the slacks $\ss(\yy^{(t)})$ stay quasi-polynomially bounded. That is, $\ss(\yy^{(t)})(e) \in [2^{-O(\log^2(mC))}, (mC)^{O(1)}]$ for any edge $e$ at any iteration $t.$
\end{theorem}

\begin{remark}
Note that under an edge deletion, $\yy$ remains feasible and $\Phi(\yy)$ decreases.
\end{remark}

We are now ready to prove the almost-linear time threshold min-cost flow algorithm using the Dual $\ell_1$-IPM (\Cref{thm:dualL1IPM}) and the adaptive min-ratio cut data structure (\Cref{thm:min_ratio_cut_data_structure_rand} and \Cref{thm:min_ratio_cut_data_structure_det}, which is deterministic).
The proof of \Cref{thm:dualL1IPM} is deferred to \Cref{sec:analysisIPM}.

\begin{proof}[Proof of \Cref{thm:main}]
We use \Cref{lem:reduceToTrans} to reduce to a corresponding transshipment problem, which we then handle by \Cref{thm:dualL1IPM}. We may assume that $Q \le m$, because after $m$ updates we can rebuild the whole data structure.

We first present the randomized algorithm.
Let $\kappa = 2^{O(\log^{3/4} m \log\log m)}$ be the approximation ratio of the min-ratio cut data structure from \Cref{thm:min_ratio_cut_data_structure_rand}.
Let $\yy^{(0)}$ be the dual solution guaranteed by \Cref{lem:initialDual} where $\cc - \BB \yy^{(0)} \ge C$ and, thus, $\Phi(\yy^{(0)}) = \O(m)$.
We initialize the edge slacks $\wt{\ss} \defeq \ss(\yy^{(0)})$ and the residual $\wt{r} \defeq r(\yy^{(0)}).$
In addition, we maintain and initialize the gradients $\wt{\gg} \defeq \gg(\wt{\ss}, \wt{r})$ and the capacities $\wt{\uu} \defeq \uu(\wt{\ss}).$ 
Then, we initialize the min-ratio cut data structure $\cD$ from \Cref{thm:min_ratio_cut_data_structure_rand} with $\eps \defeq 1 / (10 \kappa)$.
We rebuild everything after every $\eps m$ iterations. 

At each iteration $t$, let $U^{(t)}$ be the set of edges output by $\cD$ at the end of the previous iteration (or the empty set when $t = 0$).
For each $e = (u, v) \in U^{(t)}$, we update its slack $\wt{\ss}(e) \gets \ss(\yy^{(t)}, e)$ as well as its capacity $\wt{\uu}(e)$ and the gradients $\gg$ on vertex $u$ and $v$. Notice that this update always corresponds to adding some $\delta$ to $\gg(u)$ and subtracting it from $v$ in accordance with the definition of the min-ratio cut data structure. 
This can be done by querying $\cD$ about the current value of $\yy(u)$ and $\yy(v)$ in $2^{O(\log^{3/4} m \log\log m)}$-amortized time.
The updates are passed to $\cD$ and handled in the same amount of time.

Then, $\cD$ computes a $\kappa$-approximate min-ratio cut $\bDelta$ of ratio at least $g / u$, and reports $g$ and $u$. 
If the ratio is larger than $-\O(1/\kappa)$, we certify that the optimal transshipment value is less than $F$ and process the next deletion.  
Otherwise, we compute the scalar $\eta \defeq -1 / (100 \kappa^2 g)$ so that $\eta \l\wt{\gg}, \BB\bDelta\r = -1 / (100\kappa^2)$ and update $\yy^{(t+1)} \gets \yy^{(t)} + \eta \bDelta$.
This is handled using $\cD.\textrm{ToggleCut}(\eta)$. Edge deletions can also be implemented using the $\textsc{DeleteEdge}$ and $\textsc{UpdateGradient}$ operations of the min-ratio cycle data structure and do not increase the potential $\phi(\yy)$.

The correctness comes from \Cref{thm:dualL1IPM}, and it suffices to show that \eqref{eq:estAcc} holds in each iteration.
We first show that $\wt{r} \approx_{1 + 1/(10\kappa)} F - \l\dd, \yy^{(t)}\r$ all the time.
At the $t$-th iteration, let $\eta^{(t)} \bDelta^{(t)}$ be the cut that updates $\yy^{(t)}$, i.e. $ \eta^{(t)} \bDelta^{(t)}$ is the scaled output of $\cD$ that is used to obtain $\yy^{(t+1)} = \yy^{(t)} + \eta^{(t)} \bDelta^{(t)}.$
By \Cref{lem:resStable}, we have
\begin{align*}
    \frac{|\l\dd, \eta^{(t)} \bDelta^{(t)}\r|}{F - \l\dd, \yy^{(t)}\r} \le \frac{|\l\wt{\gg}, \eta^{(t)} \bDelta^{(t)}\r|}{50\kappa m} = \frac{1}{5000\kappa^3 m}
\end{align*}
because $\l\wt{\gg}, \eta^{(t)} \bDelta^{(t)}\r = -1/(100\kappa^2).$
Therefore, over $\eps m$ iterations, $F - \l\dd, \yy^{(t)}\r$ can change by at most a $(1 + 1 / (5000\kappa^3 m))^{\eps m} \le 1 + \eps$ factor.

Next, we show that $\wt{\ss} \approx_{1 + 1/(10\kappa)} \ss(\yy^{(t)}).$
At the $t$-th iteration, for any edge $e$,  let $p$ be the previous iteration when we update its slack estimation $\wt{\ss}(e)$.
That is, we have $\wt{\ss}(e) = \ss(\yy^{(p)}, e).$
\Cref{thm:min_ratio_cut_data_structure_rand} guarantees that
\begin{align*}
\wt{\uu}(e)\left|(\BB \yy^{(t)})(e) - (\BB \yy^{(p)})(e)\right|
= \ss(\yy^{(p)}, e)^{-1-\alpha}\left|\ss(\yy^{(t)}, e) - \ss(\yy^{(p)}, e)\right|
\le \eps = \frac{1}{100\kappa}
\end{align*}
Because $\ss(\yy^{(p)}, e)^{\alpha} \le 2$, this implies that $\ss(\yy^{(p)}, e)$ approximates $\ss(\yy^{(t)}, e)$ up to a factor of $1 + 2\eps \le 1 + 1/(10\kappa)$ and \eqref{eq:estAcc} follows.

In order to bound the total runtime, we need to bound the total size of updates $\{U^{(t)}\}$ output by the min-ratio cut data structure $\cD.$
Again, let $\bDelta^{(t)}$ be the cut output by $\cD$ during the $t$-th iteration and $\uu^{(t)}$ be the capacities at the time.
For an edge $e$ which is previously updated during iteration $p$, it is included in $U^{(t)}$ if the accumulated potential difference exceeds $\eps$, i.e., $\sum_{p \le x < t} \wt{\uu}_e |\eta^{(t)}(\BB \bDelta^{(x)})(e)| > \eps.$
Therefore, we can bound the total size by
\begin{align*}
    \sum_t |U^{(t)}| \le \sum_t \norm{\eta^{(t)}\wt{\UU}^{(t)} \BB \bDelta^{(t)}}_1 \eps^{-1} \le \O(m\kappa^2)
\end{align*}
because there are $\O(m \kappa^2)$ iterations and $\|\eta^{(t)} \wt{\UU}^{(t)} \BB \bDelta^{(t)}\|_1 = \O(1/\kappa)$ due to scaling.
This concludes the $m \cdot 2^{O(\log^{3/4} m \log\log m)}$ total runtime.


The deterministic version follows from the analogous argument after replacing the min-ratio cut data structure with the deterministic variant \Cref{thm:min_ratio_cut_data_structure_det}.
This yields a total runtime of $m \cdot 2^{O(\log^{5/6} m \log\log m)}.$
\end{proof}

\subsection{Analysis of L1-IPM on the Dual}
\label{sec:analysisIPM}

To prove \Cref{thm:dualL1IPM}, we analyze how the 1st and the 2nd order derivatives affect the potential value $\Phi(\yy)$ and extensively use the following facts regarding the 2nd order Taylor approximation.

\begin{lemma}[Taylor Expansion for $x^{-\alpha}$]
\label{lem:taylorAlphaPower}
If $|\delta| \le 0.1 x, x > 0$, we have
\begin{align*}
(x + \delta)^{-\alpha} &\le x^{-\alpha} - \alpha x^{-1 - \alpha} \delta + \alpha x^{-2 - \alpha} \delta^2, \text{and} \\
\left|(x + \delta)^{-1-\alpha} - x^{-1-\alpha}\right| &\le 2 |\delta|x^{-2-\alpha}
\end{align*}
\end{lemma}

\begin{lemma}[Taylor Expansion for $\log x$]
\label{lem:taylorLog}
For $x > 0, \delta > -x$, we have
\begin{align*}
    \log(x + \delta) \le \log(x) + \delta/x
\end{align*}
\end{lemma}
We then let $\zz \defeq \BB \yy$. We give a reduction to the case where $\cc(e) - \zz(e)$ is polynomially bounded for each edge $e \in E$. 

\begin{lemma}
\label{lem:slackUB}
Given any un-capacitated min-cost flow instance \eqref{eq:transPrimal}, one can modify the instance so that any new optimal solution is also optimal in the original instance and any feasible dual solution $\by$ satisfies $0 \le \cc(e) - \zz(e) \le 3mC$ where we recall $\zz = \BB \yy$. 
\end{lemma}
\begin{proof}
We can add a super source $s$ and bidirectional edges toward every other vertex of cost $mC.$
This won't affect the optimal primal solution as the residue graph induced by any optimal flow of the original instance contains no negative cycle.
Furthermore,  this ensures that for any feasible dual solution $\yy$ and for any vertex $u$, we have $|\yy(u) - \yy(s)| \le mC.$
Thus, for any edge $e = (u, v)$, we have $\cc(e) - \zz(e) \le |\cc(e)| + |\yy(u) - \yy(s)| + |\yy(v) - \yy(s)| \le 3mC.$
\end{proof}

First, we show that a cut of a small ratio makes enough progress.

\begin{lemma}[Progress from Cuts of Small Ratio]
\label{lem:progress}
Let $\kappa \ge 1$, $\gg = \gg(\yy)$ and $\uu = \uu(\yy).$
Given $\bDelta \in \R^V$ s.t. $\l\bg, \bDelta\r / \|\UU \BB \bDelta\|_1 \le -1/\kappa$ and $\l\gg, \bDelta\r = -1 / (100\kappa^2)$, we have
\begin{align*}
    \Phi(\yy + \bDelta) \le \Phi(\yy) - \frac{1}{1000\kappa^2}
\end{align*}
\end{lemma}
\begin{proof}
First, observe that 
\begin{align*}
\norm{\UU \BB\bDelta}_1 = \sum_e \ss(e)^{-1-\alpha} |(\BB\bDelta)(e)| \le \frac{1}{100\kappa}
\end{align*}
Therefore, for any edge $e$, we have
\begin{align}
\label{eq:changePerEdge}
|(\BB\bDelta)(e)| \le \frac{1}{100\kappa} \ss(e)^{1+\alpha} \le \frac{1}{100\kappa} \ss(e)(3mU)^{\alpha} \le \frac{1}{50\kappa} \ss(e)
\end{align}
where we use \Cref{lem:slackUB} to bound $\ss(e) \le 3mU.$ 

Recall the definition of $\Phi(\yy + \bDelta)$ and $\l\bg, \bDelta\r$:
\begin{align*}
\Phi(\yy + \bDelta)
&= 100m \log (F - \l\dd, \yy\r - \l\dd, \bDelta\r) + \sum_{e} (\ss(e) - (\BB\bDelta)(e))^{-\alpha} \\
\l\gg, \bDelta\r
&= \frac{-100m}{F - \l\dd, \yy\r} \l\dd, \bDelta\r + \alpha \sum_e \ss(e)^{-1-\alpha} (\BB\bDelta)(e)
\end{align*}
We will use \Cref{lem:taylorLog} and \Cref{lem:taylorAlphaPower} to analyze the change in both terms respectively.

To bound $\log (F - \l\dd, \yy\r - \l\dd, \bDelta\r)$ using \Cref{lem:taylorLog}, we first need to argue that $|\l\dd, \bDelta\r|$ is small so that $F - \l\dd, \yy\r - \l\dd, \bDelta\r$ stays positive:
\begin{align*}
\frac{100m}{F - \l\dd, \yy\r} |\l\dd, \bDelta\r|
&= \left|\l\gg, \bDelta\r - \alpha \sum_e \ss(e)^{-1-\alpha} (\BB\bDelta)(e)\right| \\
&\le \frac{1}{100\kappa^2} + \alpha \norm{\UU(\yy) \BB\bDelta}_1 \le \frac{1}{100\kappa^2} + \frac{\alpha}{100 \kappa} \le \frac{1}{100\kappa}
\end{align*}
Therefore, we have, by \Cref{lem:taylorLog}, 
\begin{align*}
100m \log (F - \l\dd, \yy\r - \l\dd, \bDelta\r)
&\le 100m\log (F - \l\dd, \yy\r) - \frac{100m}{F - \l\dd, \yy\r} \l\dd, \bDelta\r
\end{align*}

Next, we want to bound $\sum_e (\ss(e) - (\BB\bDelta)(e))^{-\alpha}.$
By \Cref{lem:taylorAlphaPower} and \eqref{eq:changePerEdge}, we have
\begin{align*}
\sum_e (\ss(e) - (\BB\bDelta)(e))^{-\alpha}
&\le \sum_e \ss(e)^{-\alpha} + \alpha \ss(e)^{-1-\alpha} (\BB\bDelta)(e) + \alpha \ss(e)^{-2-\alpha} (\BB\bDelta)(e)^2 \\
&\le \sum_e \ss(e)^{-\alpha} + \alpha \ss(e)^{-1-\alpha} (\BB\bDelta)(e) + \alpha \frac{1}{50\kappa} \sum_e \ss(e)^{-1-\alpha} |(\BB\bDelta)(e)| \\
&= \sum_e \ss(e)^{-\alpha} + \alpha \ss(e)^{-1-\alpha} (\BB\bDelta)(e) + \alpha \frac{1}{50\kappa} \norm{\UU \BB\bDelta}_1 \\
&\le \sum_e \ss(e)^{-\alpha} + \alpha \ss(e)^{-1-\alpha} (\BB\bDelta)(e) + \frac{\alpha }{5000\kappa^2}
\end{align*}

Combining both yields
\begin{align*}
\Phi(\yy + \bDelta)
\le \Phi(\yy) + \l\gg, \bDelta\r +  \frac{\alpha }{5000\kappa^2} 
= \Phi(\yy) - \frac{1}{100\kappa^2} +  \frac{\alpha}{5000 \kappa^2} 
\le \Phi(\yy) - \frac{1}{1000\kappa^2}
\end{align*}
\end{proof}

Next, we show that when the optimal value of \eqref{eq:transDual} is at least the given threshold $F$ and the current solution is far from it, the direction to the optimal solution $\yy^{\star} - \yy$ has a small ratio, which establishes that a small ratio update exists. 

\begin{lemma}[Existence of Cuts with Small Ratio]
\label{lem:existSmallRatioCut}
Consider $\yy \in \R^V$ s.t. $\Phi(\yy) \le 1000m \log(mC)$ and $\log(F - \l\dd, \yy\r) \ge -10\log(mC)$. 
Suppose there is some $\yy^{\star}$ s.t. $\l\dd, \yy^{\star}\r \ge F$ and $\cc - \BB\yy^{\star} \ge 0$, we have
\begin{align*}
\frac{\l\bg(\yy), \yy^{\star} - \yy\r}{\norm{\UU(\yy) \BB (\yy^{\star} - \yy)}_1} 
\le -\alpha
\end{align*}
\end{lemma}
\begin{proof}
We let $\zz = \BB \yy$ and $\zz^{\star} = \BB \yy^{\star}$. Writing down the definition yields:
\begin{align*}
\l\bg(\yy), \yy^{\star} - \yy\r
&= \frac{-100m}{F - \l\dd, \yy\r} \l\dd, \yy^{\star} - \yy\r + \alpha \sum_e \ss(e)^{-1-\alpha} (\zz^{\star}(e) - \zz(e)) \\
&\le -100m \frac{F - \l\dd, \yy\r}{F - \l\dd, \yy\r} + \alpha \sum_e \ss(e)^{-1-\alpha} (\zz^{\star}(e) - \zz(e)) \\
&= -100m + \alpha \sum_e \ss(e)^{-1-\alpha} (\zz^{\star}(e) - \zz(e))
\end{align*}
where we use the fact $\l\dd, \yy^{\star}\r \ge F.$

Observe that for any edge $e$, we have
\begin{align*}
\ss(e)^{-1-\alpha} (\zz^{\star}(e) - \zz(e))
&= \ss(e)^{-1-\alpha} (\ss(e)- (\cc(e) - \zz^{\star}(e))) \\
&= \ss(e)^{-\alpha} - \ss(e)^{-1-\alpha}(\cc(e) - \zz^{\star}(e)) \\
&= \ss(e)^{-\alpha} - \ss(e)^{-1-\alpha}|\cc(e) - \zz(e) + \zz(e) - \zz^{\star}(e)| \\
&\le \ss(e)^{-\alpha} + \ss(e)^{-1-\alpha}|\cc(e) - \zz(e)| - \ss(e)^{-1-\alpha}|\zz(e) - \zz^{\star}(e)| \\
&= 2\ss(e)^{-\alpha} - \ss(e)^{-1-\alpha}|\zz(e) - \zz^{\star}(e)|
\end{align*}
where we use the fact $-|a+b| \le |a| - |b|.$
Then, we have
\begin{align*}
\l\bg(\yy), \yy^{\star} - \yy\r
&\le -100m + \alpha \sum_e \ss(e)^{-1-\alpha} (\zz^{\star}(e) - \zz(e)) \\
&\le -100m + 2\alpha \sum_e \ss(e)^{-\alpha} - \alpha \norm{\UU(\yy) \BB (\yy^{\star} - \yy)}_1 \\
&= -100m - \alpha \norm{\UU(\yy) \BB (\yy^{\star} - \yy)}_1 + 2\alpha(\Phi(\yy)- 100m \log(F - \l\dd, \yy\r)) \\
&\le -100m - \alpha \norm{\UU(\yy) \BB (\yy^{\star} - \yy)}_1 + 2000\alpha m \log(mC) \\
&\le -100m - \alpha \norm{\UU(\yy) \BB (\yy^{\star} - \yy)}_1 + 2m \le -50m - \alpha \norm{\UU(\yy) \BB (\yy^{\star} - \yy)}_1
\end{align*}
where we use the fact that $\Phi(\yy) \le 1000m \log(mC)$ and $\log(F - \l\dd, \yy\r) \ge -10\log(mC).$
\end{proof}

If we compute the slack $\ss(\yy)$ and write down both $\bg(\yy)$ and $\uu(\yy)$ explicitly and exactly at every iteration, this would take $\Omega(m)$ time per iteration and, hence, a $\Omega(m^2)$ total runtime over $m^{1+o(1)}$ iterations.
However, we allow approximation in the min-ratio cut computation at each iteration.
We instead use crude estimation of the slacks $\wt{\ss} \approx \ss(\yy)$ and the residual $\wt{r} \approx r(\yy)$ to define the gradients and the capacities.
In this setting, the min-ratio cut is preserved up to a constant factor.
\begin{lemma}[Slack and residual estimations suffice]
\label{lem:approxSlack}
Let $\kappa \ge 1$ and $\yy \in \R^V$ be a feasible dual potential.
Suppose we are also given estimations of its slack $\wt{\ss} \approx_{1 + 1 / (10 \kappa)} \ss(\yy)$ and its residual $\wt{r} \approx_{1 + 1 / (10 \kappa)} r(\yy).$
Define $\wt{\gg} = \gg(\wt{\ss}, \wt{r})$ and $\wt{\uu} = \uu(\wt{\ss})$ to be the corresponding gradients and capacities respectively.

Given some $\bDelta \in \R^V$ s.t. $\l\wt{\gg}, \bDelta\r / \|\wt{\UU}\BB\bDelta\|_1 \le -1/\kappa$, we have
\begin{align*}
    \frac{\l\gg(\yy), \bDelta\r}{\norm{\UU(\yy) \BB \bDelta}_1} \le \frac{-1}{100\kappa}
\end{align*}
\end{lemma}
\begin{proof}
Because $\wt{\ss} \approx_{1 + 1 / (10 \kappa)} \ss(\yy)$, we have $\wt{\uu} \approx_{1 + 1 / (5 \kappa)} \uu(\yy)$ and, thus, $\|\wt{\UU}\BB\bDelta\|_1 \approx_2 \|\UU(\yy)\BB\bDelta\|_1.$

For the numerator part, we have, by \Cref{def:potGradLen}, 
\begin{align*}
\l\wt{\gg}, \bDelta\r
= \left\l\frac{-100m}{\wt{r}} \dd + \alpha \BB^\top \wt{\uu}, \bDelta\right\r \text{, and }
\l\gg, \bDelta\r
= \left\l\frac{-100m}{r} \dd + \alpha \BB^\top \uu, \bDelta\right\r
\end{align*}
Scaling $\wt{\gg}$ by $r / \wt{r}$ yields
\begin{align*}
\left| \frac{r}{\wt{r}} \l\wt{\gg}, \bDelta\r - \l\gg, \bDelta\r \right|
&= \left|\left\l\alpha \BB^\top \left(\frac{r}{\wt{r}} \wt{\uu} - \uu \right) , \bDelta \right\r\right| \\
&= \alpha \left|\left\l \frac{r}{\wt{r}} 1 - \frac{\uu}{\wt{\uu}} , \wt{\UU} \BB^\top \bDelta \right\r\right| \\
&\le \alpha \norm{\frac{r}{\wt{r}} 1 - \frac{\uu}{\wt{\uu}}}_{\infty} \norm{\wt{\UU} \BB^\top \bDelta}_1 \\
&\le \alpha \frac{1}{2\kappa} \cdot \kappa \left|\l\wt{\gg}, \bDelta\r\right| = \frac{\alpha}{2} \left|\l\wt{\gg}, \bDelta\r\right|
\end{align*}
where we use the facts that $\wt{r} \approx_{1+1/(10\kappa)} r$, $\wt{\uu} \approx_{1 + 1 / (5 \kappa)} \uu$, and $\l\wt{\gg}, \bDelta\r / \|\wt{\UU}\BB\bDelta\|_1 \le -1/\kappa.$
Because $|r/\wt{r} - 1| \le 1/(10\kappa)$, we have
\begin{align*}
\left|\l\wt{\gg}, \bDelta\r - \l\gg, \bDelta\r \right| \le \left| \frac{r}{\wt{r}} \l\wt{\gg}, \bDelta\r - \l\gg, \bDelta\r \right| + \left|\frac{r}{\wt{r}} - 1\right| \left|\l\wt{\gg}, \bDelta\r\right| \le 0.1 \left|\l\wt{\gg}, \bDelta\r\right|
\end{align*}
This yields $|\l\gg, \bDelta\r| \approx_{1.1} |\l\wt{\gg}, \bDelta\r|$ and the claim follows.
\end{proof}

Next, we show that both the slack $\ss(\yy)$ and the residual $r(\yy)$ changes slowly over time.
\begin{lemma}[Residual Stability]
\label{lem:resStable}
Let $\kappa \ge 1$ and $\yy \in \R^V$ be a feasible dual potential.
Suppose we are also given estimations of its slack $\wt{\ss} \approx_{1 + 1 / (10 \kappa)} \ss(\yy)$ and its residual $\wt{r} \approx_{1 + 1 / (10 \kappa)} r(\yy).$
Define $\wt{\gg} = \gg(\wt{\ss}, \wt{r})$ and $\wt{\uu} = \uu(\wt{\ss})$ to be the corresponding gradients and capacities respectively.

Given some $\bDelta \in \R^V$ s.t. $\l\wt{\gg}, \BB \bDelta\r / \|\wt{\UU}\BB\bDelta\|_1 \le -\kappa$, we have 
\begin{align*}
    \frac{|\l\dd, \bDelta\r|}{F - \l\dd, \yy\r} \le \frac{|\l\wt{\bg}, \bDelta\r|}{50\kappa m}
\end{align*}
\end{lemma}
\begin{proof}
From definition of $\gg$, we know
\begin{align*}
\frac{|\l\dd, \bDelta\r|}{F - \l\dd, \yy\r}
&= \frac{1}{100m} \left|\l\bg, \bDelta\r - \alpha \sum_e \ss(e)^{-1-\alpha} (\BB\bDelta)(e)\right| \\
&\le \frac{|\l\bg, \bDelta\r| + \alpha \norm{\UU \BB \bDelta}_1}{100m} 
\le \frac{|\l\wt{\bg}, \bDelta\r| + \alpha \norm{\wt{\UU} \BB \bDelta}_1}{50m} 
\le \frac{1+\alpha/\kappa}{50m} |\l\wt{\bg}, \bDelta\r| \le \frac{|\l\wt{\bg}, \bDelta\r|}{50\kappa m}
\end{align*}
\end{proof}

The final lemma we need is that under edge deletions, cost increases, and capacity decreases, the potential does not increase.
\begin{lemma}
\label{lemma:dec}
If a graph $G$ undergoes an edge deletion, cost increase, or capacity decrease, then a potential $\yy$ which is feasible for the dual transshipment problem is still feasible, and the potential $\Phi(\yy)$ does not increase.
\end{lemma}
\begin{proof}
We only discuss cost increases and capacity decreases, because an edge deletion can be modeled as setting a capacity to $0$ (or setting a cost to $+\infty$). When a cost is increased, by inspecting \Cref{def:reduction}, in the transshipment instance on graph $H$ the value of $\cc^H(e)$ increases for a single edge $e$. Thus, in the potential function on the dual instance, the term $(\cc^H(e) - \BB\yy(e))^{-\alpha}$ decreases, and no other terms change, as desired.

When the capacity of an edge $e = (u, v)$ decreases by $\delta$, this changes the demand of the transshipment instance in \Cref{def:reduction} in the following way. $\dd^H(v)$ decreases by $\delta$, and $\dd^H(x_e)$ increases by $\delta$. This changes $\l \dd, \yy \r$ by $\delta(\yy_{x_e} - \yy_v)$. We will argue that this quantity is nonnegative, and thus in the potential function $\Phi$, the term $100m \log(F - \l d, y \r)$ does not increase, and no other terms change. Indeed, note that the cost of the edge $(v, x_e)$ in $H$ is $0$ in \Cref{def:reduction}, so because $\cc - \BB\yy \ge 0$, we know that $0 \ge \yy_v - \yy_{x_e}$, so $\yy_{x_e} - \yy_v \ge 0$, as desired.
\end{proof}
Now, we are ready to prove the main theorem and conclude the section.
\begin{proof}[Proof of \Cref{thm:dualL1IPM}]
\Cref{lem:progress} and \Cref{lem:approxSlack} implies that $\Phi(\yy^{(t+1)}) \le \Phi(\yy^{(t)}) - \Omega(\kappa^{-2}).$ Each decremental update, the potential does not increase due to \Cref{lemma:dec}.
Therefore, after $\O(m \kappa^2)$ iterations, $\Phi(\yy^{(t)}) < -\O(m)$ and $\l\dd, \yy^{(t)}\r \ge F - (mC)^{-10}$ by \Cref{lem:smallPot}.
Whenever the approximate min ratio cut has ratio larger than $-\O(1/\kappa)$ at some iteration, \Cref{lem:existSmallRatioCut} implies that $\l\dd, \yy^{\star}\r < F.$


Finally, we argue about the range of the slacks $\ss^{(t)} \defeq \ss(\yy^{(t)})$ at any point in time.
\Cref{lem:slackUB} ensures that $\|\ss^{(t)}\|_{\infty} \le 3mC.$
Also, an upper bound on $(\ss^{(t)}(e))^{-1}$ for all edges $e$ follows from:
\begin{align*}
(\ss^{(t)}(e))^{-1} 
&\le (\Phi(\yy^{(t)}) - 100m\log(F - \l\dd, \yy^{(t)}\r))^{1/\alpha} \\
&\le (2000m \log(mC))^{1000 \log(mC)} = 2^{O(\log^2(mC))}
\end{align*}
where we use the fact that $\Phi(\yy^{(t)}) \le 1000 m \log(mC)$ and $\l\dd, \yy^{(t)}\r < F - (mC)^{-10}$
\end{proof}

\subsection{Extracting a Flow}
\label{sec:flow}
In this section we discuss how to use the potential reduction IPM to extract a flow in the end. Let $\eps = \frac{1}{2}(mC)^{-10}$, and set $F = F^{\star} + \eps$, a bit higher than the optimal value $F^{\star}$ of our instance. This ensures that our IPM does not run forever. We then run the IPM until the min-ratio cut problem has value less than $\kappa \defeq m^{-o(1)} < \alpha$. We first show that $\l \dd, \yy \r \ge F - 2\eps$ when this happens by adapting the analysis of \Cref{lem:existSmallRatioCut}.

\begin{claim}[Adapted from \Cref{lem:existSmallRatioCut}]
\label{clm:existSmallRatioCutAlmost}
Consider $\yy \in \R^V$ s.t. $\Phi(\yy) \le 1000m \log(mC)$ and $\log(F - \l\dd, \yy\r) \ge -10\log(mC)$. 
Suppose there is some $\yy^{\star}$ s.t. $\l\dd, \yy^{\star}\r \ge F - \epsilon$ and $\cc - \BB\yy^{\star} \ge 0$, we have
\begin{align*}
\frac{\l\bg(\yy), \yy^{\star} - \yy\r}{\norm{\UU(\yy) \BB (\yy^{\star} - \yy)}_1} 
\le -\alpha
\end{align*}
as long as $F - \l \dd, \yy \r \geq 2\epsilon$.
\end{claim}
\begin{proof}
We let $\zz = \BB \yy$ and $\zz^{\star} = \BB \yy^{\star}$. Writing down the definition yields:
\begin{align*}
\l\bg(\yy), \yy^{\star} - \yy\r
&= \frac{-100m}{F - \l\dd, \yy\r} \l\dd, \yy^{\star} - \yy\r + \alpha \sum_e \ss(e)^{-1-\alpha} (\zz^{\star}(e) - \zz(e)) \\
&\le -100m \frac{F - \epsilon - \l\dd, \yy\r}{F - \l\dd, \yy\r} + \alpha \sum_e \ss(e)^{-1-\alpha} (\zz^{\star}(e) - \zz(e)) \\
&\le -50m + \alpha \sum_e \ss(e)^{-1-\alpha} (\zz^{\star}(e) - \zz(e))
\end{align*}
where we use the fact $\l\dd, \yy^{\star}\r \ge F - \epsilon$ in the first inequality and $F - \l \dd, \yy \r \geq 2 \epsilon$ in the second inequality. The remainder of the proof is identical to the proof of \Cref{lem:existSmallRatioCut}.
\end{proof}

Recall that $\uu(e) = \ss(e)^{-1-\alpha}$, and the vertex gradient is
\[ \gg = \frac{-100m}{F - \l \dd, \yy \r} \bd + \alpha \BB^\top (\cc - \BB\yy)^{-1-\alpha}. \]

The min-ratio optimization problem the algorithm solves is $\min_{\DDelta\neq0} \frac{\l \gg, \DDelta \r}{\|\UU\BB\DDelta\|_1}$. If we cannot make progress then this problem has value at least $-\kappa$, for some small $\kappa = m^{-o(1)}$.

The dual of this problem is: $- \min_{\BB^\top\tilde{\xx} = \gg} \|\UU^{-1}\tilde{\xx}\|_\infty$. Let $\xx$ be a $2$-approximate solution (the constant $2$ is not important), which can be computed by running an \emph{undirected maxflow} oracle. This can be computed in nearly-linear time via previous work \cite{Peng16}. Let
\begin{align}
  \ff = \frac{F-\l \dd,\yy \r}{100m}\left(\alpha(\cc - \BB\yy)^{-1-\alpha} - \xx \right). \label{eq:f}  
\end{align} We will show that $\ff$ is a nearly-optimal solution to the transshipment problem $\min_{\BB^\top\ff=\dd,\ff\ge0} \cc^\top\ff$.
\begin{lemma}[Demands]
\label{lemma:demand}
As defined in \eqref{eq:f}, we have $\BB^\top\ff=\dd$.
\end{lemma}
\begin{proof}
Follows from $\BB^\top\xx=\gg$ and the definition of $\ff$.
\end{proof}

\begin{lemma}[Positivity]
\label{lemma:positive}
As defined in \eqref{eq:f}, $\ff \ge 0$.
\end{lemma}
\begin{proof}
It suffices to argue that $|\zz(e)| \le \alpha(\cc(e) - (\BB\yy)(e))^{-1-\alpha}$ for all $e$. Indeed, this follows because $\|\UU^{-1}\xx\|_\infty \le 2\kappa < \alpha$, by our termination condition.
\end{proof}

\begin{lemma}[Cost]
\label{lemma:cost}
As defined in \eqref{eq:f}, $\cc^\top\ff \le F^{\star} + O(\eps)$.
\end{lemma}
\begin{proof}
Because our demand $\yy$ is feasible, recall that $\l \dd, \yy \r \le F^{\star}$. Then
\[ \cc^\top\ff = (\cc - \BB\yy)^\top \ff + \yy^\top \BB^\top\ff = (\cc - \BB\yy)^\top \ff + \l \dd, \yy \r \le F^{\star} + (\cc - \BB\yy)^\top \ff. \]
We now bound the first term $(\cc - \BB\yy)^\top \ff.$ To start, we bound $(\cc - \BB\yy)^\top\xx$ as follows.
\begin{align*}
    (\cc - \BB\yy)^\top\xx &\le \sum_{e \in E} (\cc(e) - (\BB\yy)(e))|\xx(e)| \le \|\UU^{-1}\xx\|_\infty \sum_{e \in E} (\cc(e) -  (\BB\yy)(e))^{-\alpha} \\
    &\le 2\kappa \sum_{e \in E} (\cc(e) -  (\BB\yy)(e))^{-\alpha}.
\end{align*}
Now, we obtain
\begin{align*}
(\cc - \BB\yy)^\top \ff &= \frac{F - \l \dd, \yy \r}{100m}\left(\alpha \sum_{e \in E} (\cc(e) -  (\BB\yy)(e))^{-\alpha} + (\cc - \BB\yy)^\top \xx \right) \\
&\le \frac{2\eps}{100m}(\alpha+2\kappa)\sum_{e \in E} (\cc(e) -  (\BB\yy)(e))^{-\alpha} \\
&\le \frac{2\eps}{100m}(\alpha+2\kappa)(\Phi(\yy) - 100m \log(F-\l \dd,\yy \r)) \le O(\eps).
\end{align*}
where the first inequality follows from the bound on $(\cc - \BB\yy)^\top \xx$ derived above and the fact that $\l \dd, \yy \r \geq F - 2 \eps$ by \Cref{clm:existSmallRatioCutAlmost}, and the second inequality follows from the definition of $\Phi(\yy)$. The last bound then follows for small enough $\alpha$ and $\kappa$ since $\Phi(\yy)$ is upper bounded by $\tilde{O}(m)$ and the second term is also upper bounded by $\tilde{O}(m)$ because $F - \l \dd, \yy \r \geq \epsilon = 1/m^{O(1)}$.
\end{proof}

We conclude with the final statement of our static maximum flow algorithm.

\begin{theorem}
    There is a randomized algorithm for exact minimum cost flow on a graph $G = (V, E, \uu, \cc)$ with polynomially bounded integer edge costs and capacities that runs in time $m \cdot 2^{O(\log^{3/4} m \log \log m)}$ and a deterministic version that runs in time $m \cdot 2^{O(\log^{5/6} m \log \log m)}$.
\end{theorem}
\begin{proof}
    We first reduce the problem to the un-capacitated version via the reduction presented in \Cref{sec:reduceToTrans}, and then run the dual algorithm described in \Cref{sec:IPM} to figure out the value $F^{\star}$ via binary search. Finally, we extract the $2\epsilon = 1/(mC)^{10}$ approximate flow $\ff$ with the procedure described in this section and round it to an exact solution \cite{kang2015flow}.   
\end{proof}

\subsection{Deterministic Decremental SSR and SCCs}
\label{subsec:scc}

In this section, we show that the IPM developed in the previous section can be used to prove \Cref{thm:scc}.

\thmSCC*

\begin{proof}
Consider instantiating the potential reduction IPM (see \eqref{eq:dualPot}) on the dual transshipment problem with costs $\cc = \eps \cdot \vecone$ for $\eps = m^{-3}$, demand $\dd \defeq \sum_{e = (u, v)} 1_u - 1_v$, and threshold $F = 2m^2$. In other words, a feasible vector $\yy \in \R^V$ satisfies that $\yy(v) - \yy(u) \le \eps$ for every directed edge $u \to v$. Note that for any vertices $a, b$ in the same strongly connected component, that $|\yy(a) - \yy(b)| \le \eps m$, because there is a directed path between $a \to b$ and $b \to a$ of lengths at most $m$.

Start running the algorithm described in this section (solving min-ratio cuts) to decrease the potential $\Phi(\yy)$. We claim that the algorithm can find a min-ratio cut to decrease the potential if the graph is not strongly connected. Indeed, it suffices to find a feasible vector $\yy$ with $\l \dd, \yy \r \ge F$, by \Cref{lem:existSmallRatioCut}. Let $C$ be a strongly connected component in $G$, and let $\yy_v = -F$ for $v \in C$, and $\yy_v = 0$ for $v \notin C$. $\yy$ is feasible, and because there is at least one edge into $C$, $\l \dd, \yy \r \ge F$.

We now describe when to split off vertices when the graph is not strongly connected. Eventually, the algorithm will encounter a feasible potential vector $\yy$ with $|\yy(a) - \yy(b)| \ge 1$ for some vertices $a, b$. Otherwise, note that $|\l \dd, \yy \r| \le m$ always, while the threshold $F = 2m^2$. In this case, define the set $\hat{E} \defeq \{e = (u, v) : \tilde{s}(e) \le 2\eps m\}$, where $\tilde{s}(e) \approx \cc(e) - (\yy(v) - \yy(u))$ are approximate slacks that the algorithm maintains. Note that for all $e \in \hat{E}$, $|\yy(v) - \yy(u)| \le 3\eps m$, and any edge with $|\yy(v) - \yy(u)| \le \eps m$ is contained in $\hat{E}$. Let $C_a$ be the vertices that are connected to $a$ via edges in $\hat{E}$ (treated as undirected edges), and define $C_b$ similarly. We find one of $C_a, C_b$ in $\min\{\vol(C_a), \vol(C_b)\}$ time. Wlog, say we find $C_a$. Then, we delete $C_a$ from $G$, and instantiate a new decremental SCC data structure on $C_a$.

To argue correctness, it suffices to prove that $C_a$ is not the whole graph, and that any SCC in $G$ is either contained in $C_a$ or disjoint from $C_a$. For the former claim, we argue that $C_a$ does not contain $b$. By the definition of $\hat{E}$ for any $a' \in C_a$, we know that $|\yy(a') - \yy(a)| \le 3\eps m^2 < 1$, so $b \notin C_a$. Because vertices $u, v$ within the same SCC satisfy $|\yy(u) - \yy(v)| \le \eps m$, if $u \in C_a$ then also $v \in C_a$. Thus, $C_a$ does not nontrivially intersect an SCC.

Finally, we discuss minor changes needed to implement this algorithm compared to our IPM for thresholded min-cost flow, and the runtime. The
main difference is that the demand $\dd$ may change. Consider edge $e = (u, v)$ being deleted. This may increase the $100m \log(F - \l \dd, \yy \r)$ term of the potential. However, the increase is very small: because we split the graph whenever $|\yy(a) - \yy(b)| \ge 1$, we know that $\l \dd, \yy \r \le 2m$ always, and that deleting edge $(u, v)$ can only increase $\l \dd, \yy \r$ by $O(1)$. Also, changing the demand $\dd$ may affect the gradient $\gg \defeq \g \Phi(\yy)$, but only in two vertices, so this is acceptable. The cost of initializing decremental SCC algorithm on $C_a$ recursively at most increases our runtime by $O(\log m)$, since there are $O(\log m)$ recursive layers (recall that $\vol(C_a) \le \vol(C_b)$). Overall, the runtime is the same as our deterministic algorithm, which is $m \cdot e^{O((\log m)^{5/6} \log \log m)}$.
\end{proof}


\pagebreak

\printbibliography[heading=bibintoc]

@misc{GH22,
  title = {Incremental {{Approximate Maximum Flow}} in $m^{1/2+o(1)}$ Update Time},
  author = {Goranci, Gramoz and Henzinger, Monika},
  year = {2022},
  month = nov,
  number = {arXiv:2211.09606},
  eprint = {2211.09606},
  primaryclass = {cs},
  publisher = {{arXiv}},
  doi = {10.48550/arXiv.2211.09606},
  urldate = {2023-11-10},
}

@inproceedings{BlikstadBEMN22,
  author       = {Joakim Blikstad and
                  Brand, Jan van den and
                  Yuval Efron and
                  Sagnik Mukhopadhyay and
                  Danupon Nanongkai},
  title        = {Nearly Optimal Communication and Query Complexity of Bipartite Matching},
  booktitle    = {{FOCS}},
  pages        = {1174--1185},
  publisher    = {{IEEE}},
  year         = {2022}
}

@inproceedings{sulser2024,
  title={A Simple and Near-Optimal Algorithm for Directed Expander Decompositions},
  author={{Probst Gutenberg}, Maximilian and Sulser, Aurelio},
  year={2024}
}

@inproceedings{brand2023incremental,
  title={Incremental Approximate Maximum Flow on Undirected Graphs in Subpolynomial Update Time},
  author={Brand, Jan van den and Chen, Li and Kyng, Rasmus and Liu, Yang P and Peng, Richard and {Probst Gutenberg}, Maximilian and Sachdeva, Sushant and Sidford, Aaron},
  booktitle={Proceedings of the 2024 Annual ACM-SIAM Symposium on Discrete Algorithms (SODA)},
  pages={2980--2998},
  year={2024},
  organization={SIAM}
}

@misc{KMP23,
  title = {A {{Dynamic Shortest Paths Toolbox}}: {{Low-Congestion Vertex Sparsifiers}} and Their {{Applications}}},
  shorttitle = {A {{Dynamic Shortest Paths Toolbox}}},
  author = {Kyng, Rasmus and Meierhans, Simon and {Probst Gutenberg}, Maximilian},
  year = {2023},
  month = nov,
  number = {arXiv:2311.06402},
  eprint = {2311.06402},
  primaryclass = {cs},
  publisher = {{arXiv}},
  doi = {10.48550/arXiv.2311.06402},
  urldate = {2023-11-15},
  note = {Available at \url{https://arxiv.org/abs/2311.06402}}
}

@inproceedings{gutenberg2020deterministic,
  title={Deterministic algorithms for decremental approximate shortest paths: Faster and simpler},
  author={{Probst Gutenberg}, Maximilian and Wulff-Nilsen, Christian},
  booktitle={Proceedings of the Fourteenth Annual ACM-SIAM Symposium on Discrete Algorithms},
  pages={2522--2541},
  year={2020},
  organization={SIAM}
}

@inproceedings{bernstein2013maintaining,
  title={Maintaining shortest paths under deletions in weighted directed graphs},
  author={Bernstein, Aaron},
  booktitle={Proceedings of the forty-fifth annual ACM symposium on Theory of computing},
  pages={725--734},
  year={2013}
}

@inproceedings{chechik2020dynamic,
  title={Dynamic low-stretch spanning trees in subpolynomial time},
  author={Chechik, Shiri and Zhang, Tianyi},
  booktitle={Proceedings of the Fourteenth Annual ACM-SIAM Symposium on Discrete Algorithms},
  pages={463--475},
  year={2020},
  organization={SIAM}
}

@inproceedings{forster2021dynamic,
  title={Dynamic maintenance of low-stretch probabilistic tree embeddings with applications},
  author={Forster, Sebastian and Goranci, Gramoz and Henzinger, Monika},
  booktitle={Proceedings of the 2021 ACM-SIAM Symposium on Discrete Algorithms (SODA)},
  pages={1226--1245},
  year={2021},
  organization={SIAM}
}

@inproceedings{forster2019dynamic,
  title={Dynamic low-stretch trees via dynamic low-diameter decompositions},
  author={Forster, Sebastian and Goranci, Gramoz},
  booktitle={Proceedings of the 51st Annual ACM SIGACT Symposium on Theory of Computing},
  pages={377--388},
  year={2019}
}

@inproceedings{BrandCPKLPSS23,
  author       = {Brand, Jan van den and
                  Li Chen and
                  Rasmus Kyng and
                  Yang P. Liu and
                  Richard Peng and
                  {Probst Gutenberg}, Maximilian and
                  Sushant Sachdeva and
                  Aaron Sidford},
  title        = {A Deterministic Almost-Linear Time Algorithm for Minimum-Cost Flow},
  booktitle    = {64th {IEEE} Annual Symposium on Foundations of Computer Science, {FOCS}
                  2023},
  pages        = {503--514},
  publisher    = {{IEEE}},
  year         = {2023},
  url          = {https://doi.org/10.1109/FOCS57990.2023.00037},
  doi          = {10.1109/FOCS57990.2023.00037}
}

@article{shiloach1981line,
  title={An on-line edge-deletion problem},
  author={Shiloach, Yossi and Even, Shimon},
  journal={Journal of the ACM (JACM)},
  volume={28},
  number={1},
  pages={1--4},
  year={1981},
  publisher={ACM New York, NY, USA}
}

@inproceedings{rozhon2022undirected,
  title={Undirected (1+ eps)-shortest paths via minor-aggregates: near-optimal deterministic parallel and distributed algorithms},
  author={Rozho{\v{n}}, V{\'a}clav and Grunau, Christoph and Haeupler, Bernhard and Zuzic, Goran and Li, Jason},
  booktitle={Proceedings of the 54th Annual ACM SIGACT Symposium on Theory of Computing},
  pages={478--487},
  year={2022}
}

@inproceedings{chen2022maximum,
  title={Maximum flow and minimum-cost flow in almost-linear time},
  author={Chen, Li and Kyng, Rasmus and Liu, Yang P and Peng, Richard and {Probst Gutenberg}, Maximilian and Sachdeva, Sushant},
  booktitle={2022 IEEE 63rd Annual Symposium on Foundations of Computer Science (FOCS)},
  pages={612--623},
  year={2022},
  organization={IEEE}
}

@inproceedings{vdBrand23incr, author = {Brand, Jan van den and Liu, Yang P. and Sidford, Aaron}, title = {Dynamic Maxflow via Dynamic Interior Point Methods}, year = {2023}, isbn = {9781450399135}, publisher = {Association for Computing Machinery}, address = {New York, NY, USA}, url = {https://doi.org/10.1145/3564246.3585135}, doi = {10.1145/3564246.3585135}, booktitle = {Proceedings of the 55th Annual ACM Symposium on Theory of Computing}, pages = {1215–1228}, numpages = {14}, 
series = {STOC 2023} }

@article {ST83,
    AUTHOR = {Sleator, Daniel D. and Tarjan, Robert Endre},
     TITLE = {A data structure for dynamic trees},
   JOURNAL = {J. Comput. System Sci.},
  FJOURNAL = {Journal of Computer and System Sciences},
    VOLUME = {26},
      YEAR = {1983},
    NUMBER = {3},
     PAGES = {362--391},
      ISSN = {0022-0000},
   MRCLASS = {68P05 (05-04 05C05 05C35)},
  MRNUMBER = {710253},
       DOI = {10.1016/0022-0000(83)90006-5},
       URL = {https://doi.org/10.1016/0022-0000(83)90006-5},
}

@inproceedings{HenzingerKNS15,
	author    = {Monika Henzinger and
	Sebastian Krinninger and
	Danupon Nanongkai and
	Thatchaphol Saranurak},
	title     = {Unifying and Strengthening Hardness for Dynamic Problems via the Online
	Matrix-Vector Multiplication Conjecture},
	booktitle = {{STOC}},
	pages     = {21--30},
	publisher = {{ACM}},
	year      = {2015}
}

@inproceedings{BKS23,
  author       = {Sayan Bhattacharya and
                  Peter Kiss and
                  Thatchaphol Saranurak},
  title        = {Dynamic Algorithms for Packing-Covering LPs via Multiplicative Weight
                  Updates},
  booktitle    = {Proceedings of the 2023 {ACM-SIAM} Symposium on Discrete Algorithms,
                  {SODA} 2023},
  pages        = {1--47},
  publisher    = {{SIAM}},
  year         = {2023},
  url          = {https://doi.org/10.1137/1.9781611977554.ch1},
  doi          = {10.1137/1.9781611977554.CH1}
}

@inproceedings{Gupta14,
  author    = {Manoj Gupta},
  title     = {Maintaining Approximate Maximum Matching in an Incremental Bipartite
               Graph in Polylogarithmic Update Time},
  booktitle = {{FSTTCS}},
  series    = {LIPIcs},
  volume    = {29},
  pages     = {227--239},
  publisher = {Schloss Dagstuhl - Leibniz-Zentrum f{\"{u}}r Informatik},
  year      = {2014}
}

@inproceedings{M10,
  author    = {Aleksander M{\k{a}}dry},
  title     = {Fast Approximation Algorithms for Cut-Based Problems in Undirected
               Graphs},
  booktitle = {51th Annual {IEEE} Symposium on Foundations of Computer Science, {FOCS}
               2010, October 23-26, 2010, Las Vegas, Nevada, {USA}},
  pages     = {245--254},
  publisher = {{IEEE} Computer Society},
  year      = {2010},
  url       = {https://doi.org/10.1109/FOCS.2010.30},
  doi       = {10.1109/FOCS.2010.30}
}

@inproceedings{S13,
  author    = {Jonah Sherman},
  title     = {Nearly Maximum Flows in Nearly Linear Time},
  booktitle = {54th Annual {IEEE} Symposium on Foundations of Computer Science (FOCS)
               },
  pages     = {263--269},
  publisher = {{IEEE} Computer Society},
  year      = {2013},
  url       = {https://doi.org/10.1109/FOCS.2013.36},
  doi       = {10.1109/FOCS.2013.36}
}

@inproceedings{bernstein2020deterministic,
  title={Deterministic decremental reachability, scc, and shortest paths via directed expanders and congestion balancing},
  author={Bernstein, Aaron and Gutenberg, Maximilian Probst and Saranurak, Thatchaphol},
  booktitle={2020 IEEE 61st Annual Symposium on Foundations of Computer Science (FOCS)},
  pages={1123--1134},
  year={2020},
  organization={IEEE}
}

@inproceedings{bernstein2022deterministic,
  title={Deterministic decremental sssp and approximate min-cost flow in almost-linear time},
  author={Bernstein, Aaron and {Probst Gutenberg}, Maximilian and Saranurak, Thatchaphol},
  booktitle={2021 IEEE 62nd Annual Symposium on Foundations of Computer Science (FOCS)},
  pages={1000--1008},
  year={2022},
  organization={IEEE}
}

@misc{chuzhoy2020deterministicArxiv,
      title={A Deterministic Algorithm for Balanced Cut with Applications to Dynamic Connectivity, Flows, and Beyond}, 
      author={Julia Chuzhoy and Yu Gao and Jason Li and Danupon Nanongkai and Richard Peng and Thatchaphol Saranurak},
      year={2020},
      eprint={1910.08025},
      archivePrefix={arXiv},
      primaryClass={cs.DS}
}

@inproceedings{chuzhoy2020deterministic,
  title={A deterministic algorithm for balanced cut with applications to dynamic connectivity, flows, and beyond},
  author={Chuzhoy, Julia and Gao, Yu and Li, Jason and Nanongkai, Danupon and Peng, Richard and Saranurak, Thatchaphol},
  booktitle={2020 IEEE 61st Annual Symposium on Foundations of Computer Science (FOCS)},
  pages={1158--1167},
  year={2020},
  organization={IEEE}
}

@article{khandekar2009graph,
  title={Graph partitioning using single commodity flows},
  author={Khandekar, Rohit and Rao, Satish and Vazirani, Umesh},
  journal={Journal of the ACM (JACM)},
  volume={56},
  number={4},
  pages={1--15},
  year={2009},
  publisher={ACM New York, NY, USA}
}

@inproceedings{KLOS14,
  author    = {Jonathan A. Kelner and
               Yin Tat Lee and
               Lorenzo Orecchia and
               Aaron Sidford},
  title     = {An Almost-Linear-Time Algorithm for Approximate Max Flow in Undirected
               Graphs, and its Multicommodity Generalizations},
  booktitle = {Proceedings of the Twenty-Fifth Annual {ACM-SIAM} Symposium on Discrete
               Algorithms (SODA)},
  pages     = {217--226},
  publisher = {{SIAM}},
  year      = {2014}
}

@inproceedings{CGHPS20,
  title={Fast dynamic cuts, distances and effective resistances via vertex sparsifiers},
  author={Chen, Li and Goranci, Gramoz and Henzinger, Monika and Peng, Richard and Saranurak, Thatchaphol},
  booktitle={2020 IEEE 61st Annual Symposium on Foundations of Computer Science (FOCS)},
  pages={1135--1146},
  year={2020},
  organization={IEEE}
}

@inproceedings{incrTreeCutSparsifier,
  title={Fast dynamic flows, cuts, distances via vertex sparsifers},
  author={Goranci, Gramoz and Henzinger, Monika and Saranurak, Thatchaphol},
  year={2019}
}

@misc{kang2015flow,
      title={Flow Rounding}, 
      author={Donggu Kang and James Payor},
      year={2015},
      eprint={1507.08139},
      archivePrefix={arXiv},
      primaryClass={cs.DS}
}

@inproceedings{Chu21,
  author       = {Julia Chuzhoy},
  title        = {Decremental all-pairs shortest paths in deterministic near-linear
                  time},
  booktitle    = {{STOC} '21: 53rd Annual {ACM} {SIGACT} Symposium on Theory of Computing,
                  Virtual Event, Italy, June 21-25, 2021},
  pages        = {626--639},
  publisher    = {{ACM}},
  year         = {2021},
  url          = {https://doi.org/10.1145/3406325.3451025},
  doi          = {10.1145/3406325.3451025}
}

@inproceedings{van2020bipartite,
  title={Bipartite matching in nearly-linear time on moderately dense graphs},
  author={Brand, Jan van den and Lee, Yin-Tat and Nanongkai, Danupon and Peng, Richard and Saranurak, Thatchaphol and Sidford, Aaron and Song, Zhao and Wang, Di},
  booktitle={2020 IEEE 61st Annual Symposium on Foundations of Computer Science (FOCS)},
  pages={919--930},
  year={2020},
  organization={IEEE}
}

@inproceedings{van2021minimum,
  title={Minimum cost flows, mdps, and ℓ1-regression in nearly linear time for dense instances},
  author={Brand, Jan van den and Lee, Yin Tat and Liu, Yang P and Saranurak, Thatchaphol and Sidford, Aaron and Song, Zhao and Wang, Di},
  booktitle={Proceedings of the 53rd Annual ACM SIGACT Symposium on Theory of Computing},
  pages={859--869},
  year={2021}
}

@inproceedings{hua2023maintaining,
  title={Maintaining expander decompositions via sparse cuts},
  author={Hua, Yiding and Kyng, Rasmus and {Probst Gutenberg}, Maximilian and Wu, Zihang},
  booktitle={Proceedings of the 2023 Annual ACM-SIAM Symposium on Discrete Algorithms (SODA)},
  pages={48--69},
  year={2023},
  organization={SIAM}
}

@inproceedings{goranci2021expander,
  title={The expander hierarchy and its applications to dynamic graph algorithms},
  author={Goranci, Gramoz and R{\"a}cke, Harald and Saranurak, Thatchaphol and Tan, Zihan},
  booktitle={Proceedings of the 2021 ACM-SIAM Symposium on Discrete Algorithms (SODA)},
  pages={2212--2228},
  year={2021},
  organization={SIAM}
}

@inproceedings{HenzingerKN14,
  author       = {Monika Henzinger and
                  Sebastian Krinninger and
                  Danupon Nanongkai},
  title        = {Sublinear-time decremental algorithms for single-source reachability
                  and shortest paths on directed graphs},
  booktitle    = {Symposium on Theory of Computing, {STOC} 2014},
  pages        = {674--683},
  publisher    = {{ACM}},
  year         = {2014},
  url          = {https://doi.org/10.1145/2591796.2591869},
  doi          = {10.1145/2591796.2591869}
}

@inproceedings{HenzingerKN15,
  author       = {Monika Henzinger and
                  Sebastian Krinninger and
                  Danupon Nanongkai},
  title        = {Improved Algorithms for Decremental Single-Source Reachability on
                  Directed Graphs},
  booktitle    = {Automata, Languages, and Programming - 42nd International Colloquium,
                  {ICALP} 2015, Kyoto, Japan, July 6-10, 2015, Proceedings, Part {I}},
  series       = {Lecture Notes in Computer Science},
  volume       = {9134},
  pages        = {725--736},
  publisher    = {Springer},
  year         = {2015},
  url          = {https://doi.org/10.1007/978-3-662-47672-7\_59},
  doi          = {10.1007/978-3-662-47672-7\_59}
}

@inproceedings{BernsteinPW19,
  author       = {Aaron Bernstein and
                  Maximilian Probst and
                  Christian Wulff{-}Nilsen},
  title        = {Decremental strongly-connected components and single-source reachability
                  in near-linear time},
  booktitle    = {Proceedings of the 51st Annual {ACM} {SIGACT} Symposium on Theory
                  of Computing, {STOC} 2019, Phoenix, AZ, USA, June 23-26, 2019},
  pages        = {365--376},
  publisher    = {{ACM}},
  year         = {2019},
  url          = {https://doi.org/10.1145/3313276.3316335},
  doi          = {10.1145/3313276.3316335}
}

@inproceedings{ChechikHILP16,
  author       = {Shiri Chechik and
                  Thomas Dueholm Hansen and
                  Giuseppe F. Italiano and
                  Jakub Lacki and
                  Nikos Parotsidis},
  title        = {Decremental Single-Source Reachability and Strongly Connected Components
                  in $\widetilde{O}(m\sqrt{n})$ Total Update Time},
  booktitle    = {{IEEE} 57th Annual Symposium on Foundations of Computer Science, {FOCS}
                  2016, 9-11 October 2016, Hyatt Regency, New Brunswick, New Jersey,
                  {USA}},
  pages        = {315--324},
  publisher    = {{IEEE} Computer Society},
  year         = {2016},
  url          = {https://doi.org/10.1109/FOCS.2016.42},
  doi          = {10.1109/FOCS.2016.42}
}

@inproceedings{ItalianoKLS17,
  author       = {Giuseppe F. Italiano and
                  Adam Karczmarz and
                  Jakub Lacki and
                  Piotr Sankowski},
  title        = {Decremental single-source reachability in planar digraphs},
  booktitle    = {Proceedings of the 49th Annual {ACM} {SIGACT} Symposium on Theory
                  of Computing, {STOC} 2017, Montreal, QC, Canada, June 19-23, 2017},
  pages        = {1108--1121},
  publisher    = {{ACM}},
  year         = {2017},
  url          = {https://doi.org/10.1145/3055399.3055480},
  doi          = {10.1145/3055399.3055480}
}

@inproceedings{GutenbergW20a,
  author       = {{Probst Gutenberg}, Maximilian and
                  Christian Wulff{-}Nilsen},
  title        = {Decremental {SSSP} in Weighted Digraphs: Faster and Against an Adaptive
                  Adversary},
  booktitle    = {Proceedings of the 2020 {ACM-SIAM} Symposium on Discrete Algorithms,
                  {SODA} 2020, Salt Lake City, UT, USA, January 5-8, 2020},
  pages        = {2542--2561},
  publisher    = {{SIAM}},
  year         = {2020},
  url          = {https://doi.org/10.1137/1.9781611975994.155},
  doi          = {10.1137/1.9781611975994.155}
}

@inproceedings{BernsteinGW20,
  author       = {Aaron Bernstein and
                  {Probst Gutenberg}, Maximilian  and
                  Christian Wulff{-}Nilsen},
  title        = {Near-Optimal Decremental {SSSP} in Dense Weighted Digraphs},
  booktitle    = {61st {IEEE} Annual Symposium on Foundations of Computer Science, {FOCS}
                  2020, Durham, NC, USA, November 16-19, 2020},
  pages        = {1112--1122},
  publisher    = {{IEEE}},
  year         = {2020},
  url          = {https://doi.org/10.1109/FOCS46700.2020.00107},
  doi          = {10.1109/FOCS46700.2020.00107}
}

@article{GargK07,
  author       = {Naveen Garg and
                  Jochen K{\"{o}}nemann},
  title        = {Faster and Simpler Algorithms for Multicommodity Flow and Other Fractional
                  Packing Problems},
  journal      = {{SIAM} J. Comput.},
  volume       = {37},
  number       = {2},
  pages        = {630--652},
  year         = {2007},
  url          = {https://doi.org/10.1137/S0097539704446232},
  doi          = {10.1137/S0097539704446232}
}

@inproceedings{Madry10,
  author       = {Aleksander M{\k{a}}dry},
  title        = {Faster approximation schemes for fractional multicommodity flow problems
                  via dynamic graph algorithms},
  booktitle    = {Proceedings of the 42nd {ACM} Symposium on Theory of Computing, {STOC}
                  2010, Cambridge, Massachusetts, USA, 5-8 June 2010},
  pages        = {121--130},
  publisher    = {{ACM}},
  year         = {2010},
  url          = {https://doi.org/10.1145/1806689.1806708},
  doi          = {10.1145/1806689.1806708}
}

@article{HenzingerKN18,
  author       = {Monika Henzinger and
                  Sebastian Krinninger and
                  Danupon Nanongkai},
  title        = {Decremental Single-Source Shortest Paths on Undirected Graphs in Near-Linear
                  Total Update Time},
  journal      = {J. {ACM}},
  volume       = {65},
  number       = {6},
  pages        = {36:1--36:40},
  year         = {2018},
  url          = {https://doi.org/10.1145/3218657},
  doi          = {10.1145/3218657}
}

@inproceedings{ChuzhoyS21,
  author       = {Julia Chuzhoy and
                  Thatchaphol Saranurak},
  title        = {Deterministic Algorithms for Decremental Shortest Paths via Layered
                  Core Decomposition},
  booktitle    = {Proceedings of the 2021 {ACM-SIAM} Symposium on Discrete Algorithms,
                  {SODA} 2021, Virtual Conference, January 10 - 13, 2021},
  pages        = {2478--2496},
  publisher    = {{SIAM}},
  year         = {2021},
  url          = {https://doi.org/10.1137/1.9781611976465.147},
  doi          = {10.1137/1.9781611976465.147}
}

@inproceedings{ChuzhoyK2024,
  title={A Faster Combinatorial Algorithm for Maximum Bipartite Matching},
  author={Chuzhoy, Julia and Khanna, Sanjeev},
  booktitle={Proceedings of the 2024 Annual ACM-SIAM Symposium on Discrete Algorithms (SODA)},
  pages={2185--2235},
  year={2024},
  organization={SIAM}
}

@article{chen2023almost,
  title={Almost-Linear Time Algorithms for Incremental Graphs: Cycle Detection, SCCs,$s$-$t$ Shortest Path, and Minimum-Cost Flow},
  author={Chen, Li and Kyng, Rasmus and Liu, Yang P and Meierhans, Simon and {Probst Gutenberg}, Maximilian},
  journal={arXiv preprint arXiv:2311.18295},
  year={2023},
}

@inproceedings{JambulapatiJST22,
  author       = {Arun Jambulapati and
                  Yujia Jin and
                  Aaron Sidford and
                  Kevin Tian},
  title        = {Regularized Box-Simplex Games and Dynamic Decremental Bipartite Matching},
  booktitle    = {49th International Colloquium on Automata, Languages, and Programming,
                  {ICALP} 2022, July 4-8, 2022, Paris, France},
  series       = {LIPIcs},
  volume       = {229},
  pages        = {77:1--77:20},
  publisher    = {Schloss Dagstuhl - Leibniz-Zentrum f{\"{u}}r Informatik},
  year         = {2022},
  url          = {https://doi.org/10.4230/LIPIcs.ICALP.2022.77},
  doi          = {10.4230/LIPICS.ICALP.2022.77}
}

@article{ChuzhoyK2024b,
  title={Maximum Bipartite Matching in $n^{2+o(1)}$ Time via a Combinatorial Algorithm},
  author={Julia Chuzhoy and Sanjeev Khanna},
  journal={STOC 2024},
  year={2024},
  note={Accepted to STOC 2024}
}

@inproceedings{zuzic2021simple,
  title={A Simple Boosting Framework for Transshipment},
  author={Zuzic, Goran},
  booktitle={31st Annual European Symposium on Algorithms (ESA 2023)},
  year={2023},
  organization={Schloss-Dagstuhl-Leibniz Zentrum f{\"u}r Informatik}
}

@inproceedings{HenzingerJPW23,
  author       = {Monika Henzinger and
                  Billy Jin and
                  Richard Peng and
                  David P. Williamson},
  title        = {A Combinatorial Cut-Toggling Algorithm for Solving Laplacian Linear
                  Systems},
  booktitle    = {14th Innovations in Theoretical Computer Science Conference, {ITCS}
                  2023, January 10-13, 2023, MIT, Cambridge, Massachusetts, {USA}},
  series       = {LIPIcs},
  volume       = {251},
  pages        = {69:1--69:22},
  publisher    = {Schloss Dagstuhl - Leibniz-Zentrum f{\"{u}}r Informatik},
  year         = {2023},
  url          = {https://doi.org/10.4230/LIPIcs.ITCS.2023.69},
  doi          = {10.4230/LIPICS.ITCS.2023.69}
}

@inproceedings{KelnerOSZ13,
  author       = {Jonathan A. Kelner and
                  Lorenzo Orecchia and
                  Aaron Sidford and
                  Zeyuan Allen Zhu},
  title        = {A simple, combinatorial algorithm for solving {SDD} systems in nearly-linear
                  time},
  booktitle    = {Symposium on Theory of Computing Conference, STOC'13, Palo Alto, CA,
                  USA, June 1-4, 2013},
  pages        = {911--920},
  publisher    = {{ACM}},
  year         = {2013},
  url          = {https://doi.org/10.1145/2488608.2488724},
  doi          = {10.1145/2488608.2488724}
}

@inproceedings{Li20,
  author       = {Jason Li},
  title        = {Faster parallel algorithm for approximate shortest path},
  booktitle    = {Proceedings of the 52nd Annual {ACM} {SIGACT} Symposium on Theory
                  of Computing, {STOC} 2020, Chicago, IL, USA, June 22-26, 2020},
  pages        = {308--321},
  publisher    = {{ACM}},
  year         = {2020},
  url          = {https://doi.org/10.1145/3357713.3384268},
  doi          = {10.1145/3357713.3384268}
}

@inproceedings{Madry13,
  author       = {Aleksander M{\k{a}}dry},
  title        = {Navigating Central Path with Electrical Flows: From Flows to Matchings,
                  and Back},
  booktitle    = {54th Annual {IEEE} Symposium on Foundations of Computer Science, {FOCS}
                  2013, 26-29 October, 2013, Berkeley, CA, {USA}},
  pages        = {253--262},
  publisher    = {{IEEE} Computer Society},
  year         = {2013},
  url          = {https://doi.org/10.1109/FOCS.2013.35},
  doi          = {10.1109/FOCS.2013.35}
}

@inproceedings{Madry16,
  author       = {Aleksander M{\k{a}}dry},
  title        = {Computing Maximum Flow with Augmenting Electrical Flows},
  booktitle    = {{IEEE} 57th Annual Symposium on Foundations of Computer Science, {FOCS}
                  2016, 9-11 October 2016, Hyatt Regency, New Brunswick, New Jersey,
                  {USA}},
  pages        = {593--602},
  publisher    = {{IEEE} Computer Society},
  year         = {2016},
  url          = {https://doi.org/10.1109/FOCS.2016.70},
  doi          = {10.1109/FOCS.2016.70}
}

@inproceedings{GaoLP21,
  author       = {Yu Gao and
                  Yang P. Liu and
                  Richard Peng},
  title        = {Fully Dynamic Electrical Flows: Sparse Maxflow Faster Than Goldberg-Rao},
  booktitle    = {62nd {IEEE} Annual Symposium on Foundations of Computer Science, {FOCS}
                  2021, Denver, CO, USA, February 7-10, 2022},
  pages        = {516--527},
  publisher    = {{IEEE}},
  year         = {2021},
  url          = {https://doi.org/10.1109/FOCS52979.2021.00058},
  doi          = {10.1109/FOCS52979.2021.00058}
}

@inproceedings{BrandGJLLPS22,
  author       = {Brand, Jan van den and
                  Yu Gao and
                  Arun Jambulapati and
                  Yin Tat Lee and
                  Yang P. Liu and
                  Richard Peng and
                  Aaron Sidford},
  title        = {Faster maxflow via improved dynamic spectral vertex sparsifiers},
  booktitle    = {{STOC} '22: 54th Annual {ACM} {SIGACT} Symposium on Theory of Computing,
                   2022},
  pages        = {543--556},
  publisher    = {{ACM}},
  year         = {2022},
  url          = {https://doi.org/10.1145/3519935.3520068},
  doi          = {10.1145/3519935.3520068}
}

@inproceedings{ChristianoKMST11,
  author       = {Paul F. Christiano and
                  Jonathan A. Kelner and
                  Aleksander Madry and
                  Daniel A. Spielman and
                  Shang{-}Hua Teng},
  title        = {Electrical flows, laplacian systems, and faster approximation of maximum
                  flow in undirected graphs},
  booktitle    = {Proceedings of the 43rd {ACM} Symposium on Theory of Computing, {STOC}
                  2011, San Jose, CA, USA, 6-8 June 2011},
  pages        = {273--282},
  publisher    = {{ACM}},
  year         = {2011},
  url          = {https://doi.org/10.1145/1993636.1993674},
  doi          = {10.1145/1993636.1993674}
}

@inproceedings{Peng16,
  author       = {Richard Peng},
  title        = {Approximate Undirected Maximum Flows in $O(m \mathrm{polylog}(n))$
                  Time},
  booktitle    = {Proceedings of the Twenty-Seventh Annual {ACM-SIAM} Symposium on Discrete
                  Algorithms, {SODA} 2016, Arlington, VA, USA, January 10-12, 2016},
  pages        = {1862--1867},
  publisher    = {{SIAM}},
  year         = {2016},
  url          = {https://doi.org/10.1137/1.9781611974331.ch130},
  doi          = {10.1137/1.9781611974331.CH130}
}

@inproceedings{LeeS14,
  author       = {Yin Tat Lee and
                  Aaron Sidford},
  title        = {Path Finding Methods for Linear Programming: Solving Linear Programs
                  in $\tilde{O}(\sqrt{\mathrm{rank}})$ Iterations and Faster Algorithms for Maximum Flow},
  booktitle    = {55th {IEEE} Annual Symposium on Foundations of Computer Science, {FOCS}
                  2014, Philadelphia, PA, USA, October 18-21, 2014},
  pages        = {424--433},
  publisher    = {{IEEE} Computer Society},
  year         = {2014},
  url          = {https://doi.org/10.1109/FOCS.2014.52},
  doi          = {10.1109/FOCS.2014.52}
}

@article{05alstrup_top_tree,
  author       = {Stephen Alstrup and
                  Jacob Holm and
                  Kristian de Lichtenberg and
                  Mikkel Thorup},
  title        = {Maintaining information in fully dynamic trees with top trees},
  journal      = {{ACM} Trans. Algorithms},
  volume       = {1},
  number       = {2},
  pages        = {243--264},
  year         = {2005},
  url          = {https://doi.org/10.1145/1103963.1103966},
  doi          = {10.1145/1103963.1103966},
  bibsource    = {dblp computer science bibliography, https://dblp.org}
}

@misc{li2021deterministic,
      title={Deterministic Weighted Expander Decomposition in Almost-linear Time}, 
      author={Jason Li and Thatchaphol Saranurak},
      year={2021},
      eprint={2106.01567},
      archivePrefix={arXiv},
      primaryClass={cs.DS}
}

@InProceedings{dahlgaard:LIPIcs.ICALP.2016.48,
  author =	{Dahlgaard, S{\o}ren},
  title =	{{On the Hardness of Partially Dynamic Graph Problems and Connections to Diameter}},
  booktitle =	{43rd International Colloquium on Automata, Languages, and Programming (ICALP 2016)},
  pages =	{48:1--48:14},
  series =	{Leibniz International Proceedings in Informatics (LIPIcs)},
  ISBN =	{978-3-95977-013-2},
  ISSN =	{1868-8969},
  year =	{2016},
  volume =	{55},
  publisher =	{Schloss Dagstuhl -- Leibniz-Zentrum f{\"u}r Informatik},
  address =	{Dagstuhl, Germany},
  URL =		{https://drops-dev.dagstuhl.de/entities/document/10.4230/LIPIcs.ICALP.2016.48},
  URN =		{urn:nbn:de:0030-drops-63289},
  doi =		{10.4230/LIPIcs.ICALP.2016.48}
}

@article{adil2024decremental,
  title={Decremental $(1+\eps)$-Approximate Maximum Eigenvector: Dynamic Power Method},
  author={Adil, Deeksha and Saranurak, Thatchaphol},
  journal={arXiv preprint arXiv:2402.17929},
  year={2024}
}

@inproceedings{bhattacharya2023dynamic,
  title={Dynamic $(1+\eps)$-approximate matching size in truly sublinear update time},
  author={Bhattacharya, Sayan and Kiss, Peter and Saranurak, Thatchaphol},
  booktitle={Annual Symposium on Foundations of Computer Science},
  year={2023},
  organization={IEEE}
}

@inproceedings{CK19,
  title={A new algorithm for decremental single-source shortest paths with applications to vertex-capacitated flow and cut problems},
  author={Chuzhoy, Julia and Khanna, Sanjeev},
  booktitle={Proceedings of the 51st Annual ACM SIGACT Symposium on Theory of Computing},
  pages={389--400},
  year={2019}
}

@article{chen2023entropy,
  title={Entropy Regularization and Faster Decremental Matching in General Graphs},
  author={Chen, Jiale and Sidford, Aaron and Tu, Ta-Wei},
  journal={arXiv preprint arXiv:2312.09077},
  year={2023}
}

@article{dudeja2023decremental,
  title={Decremental Matching in General Weighted Graphs},
  author={Dudeja, Aditi},
  journal={arXiv preprint arXiv:2312.08996},
  year={2023}
}

@article{assadi2022decremental,
  title={Decremental matching in general graphs},
  author={Assadi, Sepehr and Bernstein, Aaron and Dudeja, Aditi},
  journal={arXiv preprint arXiv:2207.00927},
  year={2022}
}

@inproceedings{chechik2018near,
  title={Near-optimal approximate decremental all pairs shortest paths},
  author={Chechik, Shiri},
  booktitle={2018 IEEE 59th Annual Symposium on Foundations of Computer Science (FOCS)},
  pages={170--181},
  year={2018},
  organization={IEEE}
}

@inproceedings{bernstein2017deterministic,
  title={Deterministic partially dynamic single source shortest paths for sparse graphs},
  author={Bernstein, Aaron and Chechik, Shiri},
  booktitle={Proceedings of the Twenty-Eighth Annual ACM-SIAM Symposium on Discrete Algorithms},
  pages={453--469},
  year={2017},
  organization={SIAM}
}

@inproceedings{bernstein2016deterministic,
  title={Deterministic decremental single source shortest paths: beyond the o (mn) bound},
  author={Bernstein, Aaron and Chechik, Shiri},
  booktitle={Proceedings of the forty-eighth annual ACM symposium on Theory of Computing},
  pages={389--397},
  year={2016}
}

@article{gupta2018simple,
  title={Simple dynamic algorithms for maximal independent set and other problems},
  author={Gupta, Manoj and Khan, Shahbaz},
  journal={arXiv preprint arXiv:1804.01823},
  year={2018}
}

@article{italiano2010improved,
  title={Improved minimum cuts and maximum flows in undirected planar graphs},
  author={Italiano, Giuseppe F and Sankowski, Piotr},
  journal={arXiv preprint arXiv:1011.2843},
  year={2010}
}

@article{karczmarz2021fully,
  title={Fully dynamic algorithms for minimum weight cycle and related problems},
  author={Karczmarz, Adam},
  journal={arXiv preprint arXiv:2106.11744},
  year={2021}
}

@inproceedings{chen2023simple,
  title={A Simple Framework for Finding Balanced Sparse Cuts via APSP},
  author={Chen, Li and Kyng, Rasmus and Gutenberg, Maximilian Probst and Sachdeva, Sushant},
  booktitle={Symposium on Simplicity in Algorithms (SOSA)},
  pages={42--55},
  year={2023},
  organization={SIAM}
}

\pagebreak

\appendix
\section{Derivation of Transshipment Dual}

\begin{lemma}[Transshipment Dual]
\label{lem:transDual}
\begin{align}
\min_{\BB^\top \ff = \bd, \ff \ge 0} \cc^\top \ff = \max_{\cc - \BB\yy \ge 0} \l\dd, \yy\r
\end{align}
\end{lemma}
\begin{proof}
The lemma follows from
\begin{align*}
\min_{\BB^\top \ff = \bd, \ff \ge 0} \cc^\top \ff
&= \min_{\ff} \max_{\yy, \ss \ge 0} \cc^\top \ff + \l\yy, \dd - \BB^\top \ff\r + \l\ss, - \ff\r \\
&= \max_{\yy, \cc - \BB\yy -\ss = 0, \ss \ge 0} \l\dd, \yy\r \\
&= \max_{\cc - \BB\yy \ge 0} \l\dd, \yy\r.
\end{align*}
\end{proof}


\section{Reduction to (Decremental) Sparse Transshipment}
\label{sec:reduceToTrans}

Given a (decremental) min-cost flow instance, we show that it is equivalent to solving a (decremental) transshipment problem.
In particular, it is equivalent to minimum weighted bipartite matching, which is a special case of transshipment.
We summarize the technical results as the following lemma:
\begin{lemma}
\label{lem:reduceToTrans}
Suppose there is an algorithm $\cA$ that, given any (decremental) transshipment instance and some threshold $F$, outputs either a feasible flow of cost at most $F + \eps$ or certifies that the minimum cost is at least $F$ after the initialization and each edge deletion in $T(n, m)$ total time.
Then, there is a thresholded min-cost flow algorithm $\cA'$ (\Cref{def:threshold}) that runs in $T(O(m), O(m))$ total time. Furthermore, if $\cA$ successes with probability $p$, so does $\cA'.$
\end{lemma}

We first present the reduction to weighted bipartite matching and then show that the reduction carries over to decremental instances.

Let $G$ be the given graph with costs $\cc \in \R^E$ and capacities $\uu \in \R^E_+$ and $\dd \in \R^V$ be the demand vector.
The associated min-cost flow problem is as follows:
\begin{align}
\label{eq:redMinCostFlow}
\min_{\BB^\top \ff = \dd, 0 \le \ff \le \uu} \l\cc, \ff\r
\end{align}
\begin{definition}[Induced Weighted Bipartite Matching]
\label{def:reduction}
Given a min-cost flow instance $(G, \cc, \uu, \dd)$, we define its induced weighted bipartite matching instance $(H, \cc^H, \dd^H)$ as follows:
\begin{itemize}
\item $H$ is a bipartite graph over two set of vertices $V(G)$ and $E(G).$
For any directed edge $e = (u, v) \in G$, we use $x_e$ to denote the corresponding vertex in $H$ and add directed edges $(u, x_e)$ and $(v, x_e)$ to $H.$
\item For any directed edge $e = (u, v) \in G$, we define the cost
\begin{align*}
\cc^H(u, x_e) &\defeq \cc(e), \text{ and} \\
\cc^H(v, x_e) &\defeq 0
\end{align*}
\item We define the vertex demand $\dd^H \in \R^{V(H)}$ as follows:
\begin{align*}
    \dd^H(u) &\defeq \dd(u) + \sum_{e = (v, u) \in G} \uu(e) = \dd(u) + \deg^{in}_{\uu}(u), \forall u \in V(G) \\
    \dd^H(x_e) &\defeq -\uu(e), \forall e \in E(G)
\end{align*}
\end{itemize}
\end{definition}

First, we observe that if the min-cost flow instance is decremental, so is its induced weighted bipartite matching instance.
\begin{lemma}
\label{lem:decrementalEquivalence}
When the input graph $G$ is decremental in the sense that we set $\cc(e) \gets \infty$ whenever $e$ is removed, its induced weighted bipartite matching instance is also decremental because we can set $\cc^H(u, x_e) \gets \infty.$
\end{lemma}

Next, on an induced weighted bipartite matching instance, we show that a feasible dual solution $\yy \in \R^{V(H)}$, i.e., $\cc^H - \BB \yy \ge 0$,  can be explicitly constructed \emph{without} running a real-weighted negative shortest path algorithm.
\begin{lemma}
\label{lem:initialDual}
Given an induced weighted bipartite matching instance $(H, \cc^H, \dd^H)$ with $C > 0$ being the largest absolute edge cost, let $\yy \in \R^{V(H)}$ be a vertex potential on $H$ where $\yy(u) \defeq 2C, u \in V.$
We have $\cc^H - \BB \yy \ge C$.
\end{lemma}
\begin{proof}
For any edge $e = (u, x_e)$ in $H$, the value of $(\BB\yy)(e)$ is $-2C$ and $\cc^H(e) - (\BB\yy)(e) \ge C$ because $\cc^H(e)$ is at least $-C.$
\end{proof}

We show that any capacitated min-cost flow instance \eqref{eq:redMinCostFlow} is equivalent to its induced weighted bipartite matching instance:
\begin{align}
\label{eq:redTrans}
    \min_{\BB(H)^\top \ff = \dd^H,~\ff \ge 0} \l\cc^H, \ff\r
\end{align}
The equivalence is based on the following lemma.
\begin{lemma}
\label{lem:equivMinCostFlowTrans}
Given any feasible flow $\ff \in \R^E(G)$ to the min-cost flow instance \eqref{eq:redMinCostFlow}, there is a feasible flow $\bar{\ff} \in \R^E(H)$ to \eqref{eq:redTrans} of the same cost.

Similarly, any feasible flow $\bar{\ff} \in \R^E(H)$ to \eqref{eq:redTrans} corresponds to a feasible flow $\ff' \in \R^E(G)$ to  \eqref{eq:redMinCostFlow} of the same cost.
\end{lemma}
\begin{proof}
For the first part of the claim, we define $\bar{\ff} \in \R^E(H)$ as follows:
For any edge $e = (u, v) \in G$ that corresponds to edges $(u, x_e)$ and $(v, x_e)$ in $H$, we define
\begin{align*}
    \bar{\ff}(u, x_e) &\defeq \ff(e) \\
    \bar{\ff}(v, x_e) &\defeq \uu(e) - \ff(e)
\end{align*}
One can check directly that $\bar{\ff} \ge 0$ and $\l\cc^H, \bar{\ff}\r = \l\cc, \ff\r.$
To see $\BB(H)^\top \bar{\ff} = \dd^H$, we first observe that there are $\uu(e)$ net-flow injected into any vertex $x_e.$
For any $u \in V(H)$ that corresponds to a vertex in $G$, its net-flow is
\begin{align*}
\sum_{e = (u, v) \in G} \bar{\ff}(u, x_e) + \sum_{e = (v, u) \in G}\bar{\ff}(u, x_e)
= \sum_{e = (u, v) \in G} \ff(e) + \sum_{e = (v, u) \in G} \uu(e) - \ff(e) = \dd(u) + \deg^{in}_{\uu}(u)
\end{align*}
because the net flow of $u$ in $G$ is $\dd(u).$

For the second part of the claim, we define the flow $\ff' \in \R^{E(G)}$ as follows:
\begin{align*}
    \ff'(e) \defeq \bar{\ff}(u, x_e), \forall e = (u, v) \in E(G)
\end{align*}

We first check $0 \le \ff' \le \uu$.
Because $\BB(H)^\top \bar{\ff} = \dd^H$ and $\bar{\ff} \ge 0$, we know for any edge $e = (u, v)$ that
\begin{align*}
    \bar{\ff}(u, x_e) + \bar{\ff}(v, x_e) = \uu(e)
\end{align*}
and both $\bar{\ff}(u, x_e)$ and $\bar{\ff}(v, x_e)$ are non-negative.
Therefore, $\bar{\ff}(u, x_e)$ lies within $[0, \uu(e)]$ and so does $\ff'(e).$

Next, we check $\l\cc, \ff'\r = \l\cc^H, \bar{\ff}\r.$
This comes from that $\cc^H(u, x_e) = \cc(e)$ and $\cc^H(v, x_e) = 0$ for any edge $e = (u, v) \in E(G).$

Finally, we check that $\ff'$ routes the demand, i.e., $\BB^\top \ff' = \dd.$
For any vertex $u$, its residue w.r.t. $\ff'$ is
\begin{align*}
\sum_{(u, v) \in G} \ff'(u, v) - \sum_{(v, u) \in G} \ff'(v, u)
&= \sum_{e = (u, v) \in G} \bar{\ff}(u, x_e) - \sum_{e = (v, u) \in G} \bar{\ff}(v, x_e) \\
&= \sum_{e = (u, v) \in G} \bar{\ff}(u, x_e) - \sum_{e = (v, u) \in G} \uu(e) - \bar{\ff}(u, x_e) \\
&= \sum_{e = (u, v) \in G} \bar{\ff}(u, x_e) + \sum_{e = (v, u) \in G} \bar{\ff}(u, x_e) - \deg^{in}_{\uu}(u) \\
&= \dd^H(u) - \deg^{in}_{\uu}(u) = \dd(u)
\end{align*}
where we use the fact $\bar{\ff}(u, x_e) + \bar{\ff}(v, x_e) = \uu(e)$ for any edge $e = (u, v) \in E(G).$
\end{proof}

\begin{proof}[Proof of \Cref{lem:reduceToTrans}]
The lemma follows from \Cref{lem:decrementalEquivalence} and \Cref{lem:equivMinCostFlowTrans}.
\end{proof}

\end{document}